\documentclass[12pt]{UniboBook}

\usepackage{hyperref}
\usepackage{color}
\usepackage{amssymb}
\usepackage{amstext}
\usepackage{graphicx}
\usepackage{subfigure}
\usepackage{url,xspace}
\usepackage{epsfig}
\usepackage{amsfonts}
\usepackage[leqno]{amsmath}
\usepackage{amsthm}
\usepackage{xcolor}

\newtheorem{theorem}{Theorem}[section]
\newtheorem{definition}[theorem]{Definition}

\newtheorem{example}[theorem]{Example}

\newtheorem{remark}[theorem]{Remark}
\newtheorem{notation}[theorem]{Notation}
\newtheorem{fact}[theorem]{Fact}
\newtheorem{model}[theorem]{CWC Modelling}

\newcommand{\CalculusShortName}{\textsc{CWC}}

\newcommand{\blank}{\sharp}

\newcommand{\qqop}[1]{\mathrel{\makebox[2em]{$#1$}}}

\newcommand{\agr}{\quad\big|\quad}

\newcommand{\mydots}{\cdot\cdot\cdot}

\newcommand{\AT}{\mathcal{A}}
\newcommand{\LT}{\mathcal{L}}
\newcommand{\AS}{\overline{\mathcal{A}}}
\newcommand{\TT}{\mathcal{T}}
\newcommand{\TS}{\overline{\mathcal{T}}}
\newcommand{\TSG}{\overline{\mathcal{T}_G}}

\newcommand{\Pat}{p}
\newcommand{\LeftPP}{\mathcal{P}}
\newcommand{\LeftPat}{p}

\newcommand{\OT}{\mathcal{O}}

\newcommand{\RightPat}{o}

\newcommand{\TOP}{\top}

\newcommand{\BR}{\textsc{BR}}

\newcommand{\VarOf}{\textit{Var}}
\newcommand{\LengthOf}[1]{|#1|}

\newcommand{\CC}{\mathcal{C}}

\newcommand{\QQ}{\mathcal{R}}
\newcommand{\RR}{\mathcal{S}}
\newcommand{\DD}{\mathcal{D}}
\newcommand{\BB}{\mathcal{B}}
\newcommand{\Ode}{\mathcal{E}}

\newcommand{\NN}{\mathcal{N}}
\newcommand{\OO}{\textsc{Occ}}

\newcommand{\ltrans}[1]{\stackrel{#1}{\longrightarrow}}

\makeatletter
\def\mapstofill@{%
\arrowfill@{\mapstochar\relbar}\relbar\rightarrow}
\newcommand*\xmapsto[2][]{%
\ext@arrow 0000\mapstofill@{#1}{#2}}
\makeatother

\newcommand{\srewrites}[1]{\stackrel{#1}{\longmapsto}}
\newcommand{\xsrewrites}[1]{\xmapsto{#1}}

\newcommand{\into}{\ensuremath{\,\rfloor}\,}

\newcommand{\phole}{\square}

\newcommand{\conc}{\;\,}
\newcommand{\emptyseq}{\bullet}

\newcommand{\TSV}{\VV_{\TS}}
\newcommand{\ASV}{\VV_{\AS}}
\newcommand{\VV}{\mathcal{V}}

\newcommand{\mapstoDesug}{\mapsto}

\newcommand{\red}{\mapstoDesug}

\newcommand{\st}{simple term}
\newcommand{\sts}{simple terms}
\newcommand{\St}{Simple term}
\newcommand{\Sts}{Simple terms}
\newcommand{\short}{\CalculusShortName}
\newcommand{\ov}[1]{\overline{#1}}
\newcommand{\set}[1]{\{#1\}} 

\newcommand{\s}{\conc}
\renewcommand{\wr}[3]{( #1 \into #2)^{#3}}
\newcommand{\wsrew}[1]{\stackrel{#1}{\longmapsto}_{\sf w}}
\newcommand{\csrew}[1]{\stackrel{#1}{\longmapsto}_{\sf c}}
\newcommand{\srew}[1]{\stackrel{#1}{\longmapsto}}
\newcommand{\srewRule}[3]{#1\srew{#3}#2}
\newcommand{\OTS}{\overline{\mathcal{O}}}
\newcommand{\timesSilent}{}

\begin{document}

\mainmatter  

\title{Modelling Biological and Ecological Systems with the Calculus of Wrapped Compartments}

\author{\and Marco Aldinucci \and Livio Bioglio \and Cristina Calcagno \and Mario Coppo \and Ferruccio Damiani \and Maurizio Drocco \and Elena Grassi \and Pablo Ram\'{o}n \and Eva Sciacca \and Salvatore Spinella \and Angelo Troina}

\pagestyle{smallheadings}

\maketitle

\begin{acknowledgments}
This document collects and slightly revises and extends original research presented in~\cite{CDDGGT_TCSB11,spatial_COMPMOD11,BCCDSST11,CWCTCS,CMC13}.

A substantial part of the research has been founded by the BioBITs Project (\emph{Converging Technologies}, area: Biotechnology-ICT), Regione Piemonte (Italy), and has been carried on at the Dipartimento di Informatica, Universit\`a di Torino.

The material presented in Chapter~\ref{ecolog} has been originally developed while Angelo Troina was visiting Pablo Ram\'{o}n at the
Instituto de Ecolog\'ia at Universidad Tecnica Particular de Loja (Ecuador), in the context of the Prometeo program founded by SENESCYT.
\end{acknowledgments}

\begin{abstract}
The modelling and analysis of biological systems has deep roots in Mathematics, specifically in the field of Ordinary Differential Equations. Alternative approaches based on formal calculi, often derived from process algebras or term rewriting systems, provide a quite complementary way to analyse the behaviour of biological systems. These calculi allow to cope in a natural way with notions like compartments and membranes, which are not easy (sometimes impossible) to handle with purely numerical approaches, and are often based on stochastic simulation methods.

The Calculus of Wrapped Compartments is a framework based on stochastic multiset rewriting in a compartmentalised setting used for the description of biological and ecological systems.

We provide an extended presentation of the Calculus of Wrapped Compartments, sketch a few modelling guidelines to encode biological and ecological interactions, show how spatial properties can be handled within the framework and define a hybrid simulation algorithm. Several applications in Biology and Ecology are proposed as modelling case studies.
\end{abstract}


\tableofcontents


\chapter{Introduction}
\label{intro}
\section{Modelling Biological Systems}

The most common approaches used by biologists to describe biological systems have been mainly based on the use of deterministic mathematical means like, e.g., Ordinary Differential Equations (ODEs for short).
ODEs make it possible to abstractly reason on the behaviour of biological systems and to perform a quantitative \emph{in silico} investigation.
However, this kind of modelling becomes more and more difficult, both in the specification phase and in the analysis processes, when
the complexity of the biological systems taken into consideration increases. More recently, the observation that biological systems (for example in the case of chemical instability) are inherently stochastic \cite{science02}, has led a growing interest in the stochastic modelling of chemical kinetics.

Besides, the concurrently interacting structure of biological systems has inspired the possibility to describe them by means of formalisms developed in Computer Science for the description of computational entities~\cite{RS02}.
Different formalisms have either been applied to (or have been inspired from) biological systems. Automata-based models~\cite{ABI01,MDNM00} have the advantage of allowing the direct use of many verification tools such as model checkers.
Rewrite systems~\cite{DL04,P02,BMMT06} usually allow describing biological systems with a notation that can be easily understood by
biologists.
Process calculi, including those commonly used to describe biological systems~\cite{RS02,PRSS01,Car04}, have the advantage of being compositional, but their way of describing biological entities is often less intuitive.
Quantitative simulations of biological models represented with these kind of frameworks (see, e.g.~\cite{PRSS01,DPPQ06,KMT08,BMMTT08,DPR08,preQAPL2010}) are usually developed via a stochastic method derived from Gillespie's algorithm~\cite{G77}.

The ODE description of biological systems determines \emph{continuous} and {deterministic} models in which variables describe the concentrations of the species involved in the system as functions of the time. These models are based on average reaction rates, measured from real experiments which relate to the change of concentrations over time, taking into account the known properties of the involved chemicals, but possibly abstracting away some unknown mechanisms.

In contrast to the deterministic model, \emph{discrete} and {stochastic} simulations involve random variables. 
Since the basic steps of a molecular reaction are described in terms of their probability of occurrence, the behaviour
of a reaction is not determined a priori but characterised statistically. Thus, biological reactions fall in the category of stochastic systems and stochastic models for their kinetics are widely accepted as the best way to represent and simulate genetic and biochemical networks. In particular, when the system to be described is based on the interaction of few molecules, or we want to simulate the functioning of a little pool of cells, the system may expose several different behaviour, observed with different probabilities. As a consequence, the stochastic approach is always valid when the deterministic one is (i.e. when the system is stable and exposes only one possible behaviour), and it may be valid when the ordinary deterministic is not (i.e. in a nonlinear system in the neighbourhood of a chemical instability).

\section{Modelling Ecological Systems}

Answers to ecological questions could rarely be formulated as general laws: ecologists deal with \emph{in situ} methods and experiments which cannot be controlled in a precise way since the phenomena observed operate on much larger scales (in time and space) than man can effectively study. Actually, to carry on ecological analyses, there is the need of a ``macroscope''! 

Theoretical and Computational Ecology, the scientific disciplines devoted to the study of ecological systems using theoretical methodologies together with empirical data, could be considered as a fundamental component of such a macroscope. Within these disciplines, quantitative analysis, conceptual description techniques, mathematical models,  and computational simulations are used to understand the  fundamental biological conditions and processes that affect populations dynamics (given the underlying assumption that phenomena observable across species and ecological environments are generated by common, mechanistic processes)~\cite{Pie77}.

A model in the Calculus of Wrapped Compartments (CWC for short) consists of a term, representing a (biological or ecological) system and a set of rewrite rules which model the transformations determining the system's evolution~\cite{preQAPL2010,CDDGGT_TCSB11}. Terms are defined from a set of atomic elements via an operator of compartment construction. Each compartment is labelled with a nominal type which identifies the set of rewrite rules that may be applied into it.
The CWC framework is based on a stochastic semantics and models an exact scenario able to capture the stochastic fluctuations that can arise in the system.

We present some modelling guidelines to describe, within CWC, some of the main common features and models used to represent ecological interactions and population dynamics. A few generalising examples illustrate the abstract effectiveness of the application of CWC to ecological modelling.

\section{Motivation and Methodology}

At the beginning of the second half of the twentieth century, when ecology was still a young science and mathematical models for ecological systems were in their infancy, Elton~\cite{Elt58}, acknowledging the influence of Lotka \cite{Lot25} and Volterra~\cite{Vol26}, wrote: ``Being mathematicians, they did not attempt to contemplate a whole food--chain with all the complications of five stages. They took two: a predator and its prey''.
Nowadays, in the era of computational ecological modelling, deterministic systems based on ordinary differential equations for two variables, or even a whole food chain, appear like simple idealisations quite distant from the real complexity of nature. Predator--prey interactions are now considered  as ``consumer--resource'' interactions embedded within the large ecological networks that underlie biodiversity. Lotka--Volterra equations and their many descendants assume that individuals within a system are well mixed and interact at mean population abundances. They are mean-field equations that use the mass--action law to describe the dynamics of interacting populations, and ignore both the scale of individual interactions and their spatial distribution. However, because ecological systems are typically nonlinear, they often cannot be solved analytically, and, in order to obtain sensible results, nonlinear, stochastic computational techniques must be used.

The formal framework to be used as the modelling core of this project should thus be able to manage several features which are typical of ecological systems. Namely, complex ecological systems are multilevel, they follow non linear, stochastic dynamics and involve a distributed spatial organisation.

\subsection{Multilevel Modelling} The role of the computational methodology used to model and simulate ecological systems is to address questions on the relationship between systems dynamics at different temporal, spatial, and organisational (or structural) scales. In particular, it is important to address the variability at small, local scales and its effects on the dynamics of the aggregated quantities measured at large, global scales~\cite{Pas05}.

\subsection{Stochastic Modelling} Biological and ecological models can be deterministic or stochastic~\cite{Bol08}. Given an initial system, deterministic simulations always evolve in the same way, producing a unique output~\cite{SM90}. Deterministic methods give a picture of the average, expected behaviour of a system, but do not incorporate random fluctuations. On the other hand, stochastic models allow to describe the random perturbations that may affect natural living systems, in particular when considering small populations evolving at slow interactions. Actually, while deterministic models are approximations of the real systems they describe, stochastic models, at the price of an higher computational cost, can describe exact scenarios. Stochastic models, such as interacting particle systems, can also help us examine new approaches for scaling up individual--based dynamics.\footnote{Note that the impact of stochastic factors and the corresponding level of a system's uncertainty are much higher in Ecology than in other natural sciences. Common sources of uncertainty are, e.g., the poor accuracy of ecological data and their transient nature. Noise that is inevitably present in ecosystems can significantly change the properties of an ecological model and this fundamental uncertainty affects the accuracy of ecological data \cite{PP12}.}

\subsection{Spatial Modelling} The impact of space tends to make a system's dynamics significantly more complicated compared with its non--spatial counterpart and to bring new and bigger challenges to simulations. Formal models dealing explicitly with spatial coordinates are able to depict more precise localities in a biological system and/or \emph{ecological niches}, describing, for example, how organisms or populations respond to the distribution of resources and competitors~\cite{LB98}.

\section{The Calculus of Wrapped Compartments} While the Calculus of Wrapped Compartments has been originally developed to deal with biomolecular interactions and cellular communications, it appears to be particularly well suited also to model and analyse interactions in ecology.
The Calculus of Wrapped Compartments satisfies the main requirements addressed in the previous sections. Namely, CWC is able to model and simulate: (i) multilevel systems, (ii) stochastic dynamics, (iii) explicit spatial systems.

The compartment operator of the calculus can be used to describe the topological organisation of a systems. It also allows to deal with multilevel systems by defining different set of rules for different compartments, reflecting the interactions taking place at the different levels of the system.

The evolution of a system described in CWC follows a stochastic simulation model defined by incorporating a collision-based framework along the lines of the one presented by Gillespie in~\cite{G77}, which is, \emph{de facto}, the standard way to model quantitative aspects of biological systems. The basic idea of Gillespie's algorithm is that a rate is associated with each considered reaction. This rate is used as the parameter of an exponential probability distribution modelling the time needed for the reaction to take place. In the standard approach, the reaction \emph{propensity} is obtained by multiplying the rate of the reaction by the number of possible combinations of reactants in the compartment in which the reaction takes place, modelling the law of mass action.

The calculus has been extensively used to model real biological scenarios, in particular related to the AM-symbiosis~\cite{CDDGGT_TCSB11,spatial_COMPMOD11}.\footnote{Arbuscular  Mycorrhiza (AM) is a class of fungi constituting a vital mutualistic interaction for terrestrial ecosystems. More than 48\% of land plants actually rely on mycorrhizal relationships to get inorganic compounds, trace elements, and resistance to several kinds of pathogens.} An hybrid semantics for CWC, combining stochastic transitions with deterministic steps, modelled by Ordinary Differential Equations, has been proposed in~\cite{HCWC_mecbic10,CWCTCS}.

A spatial extension of CWC has been proposed in~\cite{BCCDSST11}, incorporating a two--dimensional spatial description of the elements in the system through axial coordinates and special rules for the movement of system components in space. The spatial extension of the calculus can be generalised to deal with spaces defined in more than two dimensions.

\section{Case Studies}

In Chapter~\ref{case} we present a CWC model describing a newly discovered ammonium
transporter. This transporter is believed to play a fundamental
role for plant mineral acquisition, which takes place in the
arbuscular mycorrhiza, the most wide-spread plant-fungus symbiosis
on earth.

In our experiments the passage of NH3 / NH4+ from the fungus to
the plant has been dissected in known and hypothetical mechanisms;
with the model so far we have been able to simulate the behaviour
of the system under different conditions. Our simulations
confirmed some of the latest experimental results about the
LjAMT2;2 transporter. The initial simulation results of the
modelling of the symbiosis process are promising and indicate new
directions for biological investigations.

Our second case study is taken from Ecology: in Section~\ref{Sec:Croton}, we model within CWC the distribution of height of \emph{Croton wagneri}, a shrub in the dry ecosystem of southern Ecuador, and investigate how it could adapt to global climate change.

In Chapter~\ref{spatial} we explore a few scenarios in which topological properties of the system are to be taken into account.

In Chapter~\ref{hybrid} we present an hybrid simulation algorithm for the CWC framework and apply it to a variant of Lotka-Volterra dynamics and an HIV-1 transactivation mechanism.

\chapter{The Calculus of Wrapped Compartments}
\label{cwc}
Like most modelling languages based on term rewriting (notably CLS), a CWC (biological) model consists of a term, representing the system and a set of
rewrite rules which model the transformations determining the system's evolution. Terms are defined from a set atomic elements via an operator of
compartment construction. Compartments are enriched with a nominal type, represented as a label, which identifies the set of rewrite rules that may be
applied to them.

\section{Terms and Structural Congruence}\label{CLS_syntax}

\emph{Terms} of the \short\ calculus are intended to represent a biological system.
A \emph{term} is a multiset of \emph{\st s}. \Sts, ranged over by $t$,
$u$, $v$, $w$, are built by means of the \emph{compartment} constructor, $(-\into -)^-$, from a set $\AT$ of \emph{atomic elements} (\emph{atoms} for short), ranged over by $a$, $b$, $c$, $d$, and from a set $\LT$ of \emph{compartment types} (represented as \emph{labels} attached to compartments),
ranged over by $\ell,\ell',\ell_1,\ldots$ and containing a distinguished element $\top$ which  characterizes the top level compartment. The syntax of {\st s} is given in Figure~\ref{fig:CWM-syntax}. We write $\overline{t}$ to
denote a (possibly empty) multiset of \sts\ $t_1\mydots t_n$. Similarly, with $\overline{a}$ we denote a (possibly empty) multiset of atoms. The set of simple terms will be denoted by $\TT$.

Then, a \st\ is either an atom or a compartment $(\overline{a}\into \overline{t})^\ell$ consisting of a \emph{wrap} (represented by the multiset of atoms
$\overline{a}$), a \emph{content} (represented by the term $\overline{t}$) and a \emph{type} (represented by the label $\ell$). Note that we do not allow
nested structures in wraps but only in compartments.
%
We write  $\emptyseq$ to represent the empty multiset 
 and denote the union of two multisets $\overline{u}$ and $\overline{v}$ as
$\overline{u}\conc\overline{v}$. The notion of inclusion between multisets, denoted as usual by $\subseteq$, is the natural extension of the analogous
notion between sets. The set of terms (multisets of simple terms) and the set of multisets of atoms will be denoted by $\TS$ and $\AS$, respectively. Note
that $\AS \subseteq \TS$.

Since a term $\overline{t}=t_1\mydots t_n$ is intended to represent a multiset we introduce a relation
of structural congruence between terms of CWC defined as the least equivalence relation on terms satisfying the rules given in Figure~\ref{fig:CWM-syntax}. From now on we will always consider terms modulo structural
congruence. To denote multisets of atomic elements we will sometime use the compact notation $n a$ where $ a $ is an atomic element and $ n $ its multiplicity so for instance $3a\,2b$ is a notation for the multiset $a\,a\,a\,b\,b$.

\begin{figure}
\hrule $\;$ \\
\textbf{\St s syntax}
\\
 $
 \begin{array}{lcl}
  t & \;\qqop{::=}\; & a \!\agr\! (\overline{a}\into\overline{t})^\ell
   \\
 \end{array}
 $
\\
$\;$
\\
\textbf{Structural congruence}
\\
$
\begin{array}{l}
\overline{t} \conc u \conc w \conc \overline{v} \equiv
\overline{t} \conc w \conc u \conc \overline{v} \\  
  \mbox{if } \; \ov{a}\equiv\ov{b}~~\text{and}~~\ov{t}\equiv\ov{u} ~~\mbox{ then }~~
    (\ov{a}\into\ov{t})^\ell\equiv(\ov{b}\into\ov{u})^\ell
%
%
\\
\end{array}
$ 
 \caption{ \CalculusShortName\ term syntax and structural congruence rules}
\label{fig:CWM-syntax}
\end{figure}

An example of term is $\ov{t} = 2a \conc 3b \conc (c \conc d \into e \conc f)^\ell$ representing a multiset consisting of two atoms $a$ and three $b$ (for instance five
molecules) and an $\ell$-type compartment $(c \conc d \into e \conc f)^\ell$ which, in turn, consists of a wrap (a membrane) with two atoms $c$ and $d$
(for instance, two proteins) on its surface, and containing the atoms $e$ (for instance, a molecule) and $f$ (for instance a DNA strand whose functionality can be modelled as an atomic element). See
Figure~\ref{fig:example CWM} for some graphical representations.



\begin{figure*}
\centering
\subfigure[] {
\includegraphics[height=28mm]{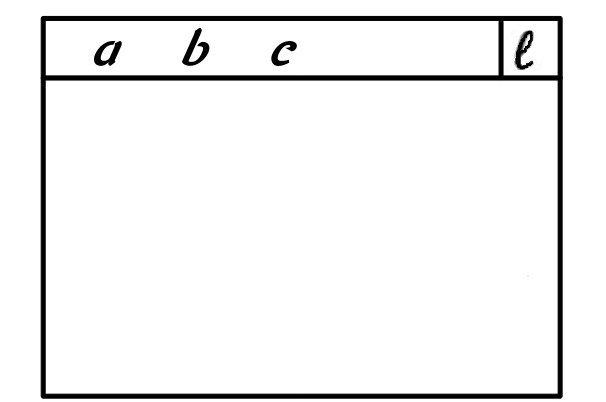}
}
\centering
\subfigure[] {
\includegraphics[height=28mm]{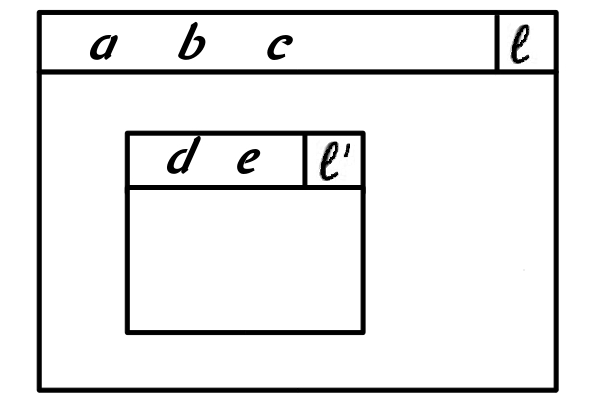}
}
\subfigure[] {
\includegraphics[height=28mm]{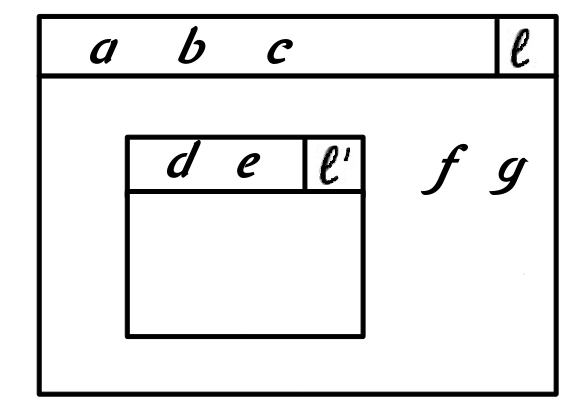}
}

\caption{\textbf{(a)} represents $(a \conc b \conc c \into \emptyseq)^\ell$; \textbf{(b)} represents $(a \conc b \conc c  \into (d \conc e \into
\emptyseq)^{\ell'})^\ell$; \textbf{(c)} represents $(a \conc b \conc c \into (d \conc e \into \emptyseq)^{\ell'} \conc f \conc g)^\ell$}
\label{fig:example CWM}
\end{figure*}
\medskip

\begin{notation}[Top-level compartment]\label{nota:WholeTerm} For sake of uniformity we assume that the term representing the whole
system is always a single compartment labelled $\top$ with an empty wrap, i.e., all systems are represented by a term of the shape
$\wr{\emptyseq}{\ov{t}}{\top}$, which we will also write as $\ov{t}$ for simplicity.
\end{notation}

\section{Contexts}
The notion of reduction for CWC systems is formalized via the notion of reduction context. To define them, the syntax of terms is enriched with a new
element $\phole$ representing a hole which can be filled only by a \emph{single} compartment. \emph{Reduction contexts} (ranged over by $C$) are defined
by:
$$
 C  \;\qqop{::=}\; \phole  \agr (\overline{a} \into C \conc \overline{t})^\ell
$$
where $\overline{a} \in \AS$, $\overline{t}\in \TS$ and $\ell \in \LT$.
Note that, by definition, every context contains a single hole $\phole$.
The infinite set of contexts is denoted by $\CC$.

Given a compartment $t$ and a context $C$, the compartment obtained by filling the hole in $C$ with $t$ is denoted  by $C[t]$. For instance, if
$t=\wr{\ov{a}}{\ov{u}}{\ell}$ and $C=\wr{\ov{b}}{\phole \conc \ov{c}}{\ell'}$, then $C[t]=\wr{\ov{b}}{\wr{\ov{a}}{\ov{u}}{\ell} \conc \ov{c}}{\ell'}$.

The composition of two contexts $C$ and $C'$, denoted by $C[C']$, is the context obtained by replacing $\phole$ with $C'$ in $C$. For example, given $C=(a
\into \phole\conc b )^\ell$, $C'= (c \into \phole \conc d\conc e)^{\ell'}$, we get $C[C']= (a \into  (c \into \phole \conc d \conc
e)^{\ell'} \conc b)^\ell$.

\section{Rewrite Rules and Qualitative Reduction Semantics}
A rewrite rule is defined by a pair of compartments (possibly containing variables), which represent the patterns along which the system transformations
are defined. The choice of defining rules at the level of compartments allows to simplify the formal treatment, allowing a uniform presentation of the
system semantics.




In order to formally define the rewriting semantics, we introduce the notion of open term (a term containing variables) and pattern (an open term that may
be used as left part of a rewrite rule). To respect the syntax of terms, we distinguish between ``wrap variables'' which may occur only in compartment
wraps (and can be replaced only by multisets of atoms) and ``content variables'' which may only occur in compartment contents or at top level (and can be
replaced by arbitrary terms)

 Let $\TSV$ be a set of \emph{content variables}, ranged over by $X,Y,Z$, and $\ASV$ a set of \emph{wrap variables}, ranged over by $x,y,z$ such that $\TSV \cap \ASV = \emptyset$.
 We denote by $\VV$ the set of all variables $\TSV \cup \ASV$, and with $\rho$ any variable in $\VV$. \emph{Open terms} are terms
 which may contain occurrences of wrap variables in compartment
 wraps and content variables in compartment contents or at top level.
  Similarly to terms, open terms are defined as multisets $\ov{o}$ of simple open terms defined in the following way:
\[\begin{array}{lcl}
o           & \; \qqop{::=} \; & a \agr X \agr   (\overline{q}  \into \overline{o} )^\ell
\\
q  & \; \qqop{::=} \; &  a \agr x
\end{array}  \]
(i.e. $\ov{q}$ denotes a multiset formed only of atomic elements and wrap variables).
 Let $\OT$ and $\OTS$ denote the set of simple open terms and the set of open terms (multisets of simple open terms), respectively.
 An open term is \emph{linear} if each variable occurs in it at most once.

An \emph{instantiation} (or substitution) is defined as a partial function $\sigma : \VV \rightarrow \TS$. An instantiation must preserve the type of
variables, thus for $X \in \TSV$ and $x \in \ASV$ we have $\sigma(X) \in \TS$ and $\sigma(x) \in \AS$, respectively.  Given $\ov{o}\, \in \OTS$, with
$\ov{o}\, \sigma$ we denote the term obtained by replacing each occurrence of each variable $\rho \in \VV$ appearing in $\ov{o}$ with the corresponding
term $\sigma(\rho)$.

Let $\Sigma$ denote the set of all the possible instantiations and $\VarOf(\ov{o})$ denote the set of variables appearing in $\ov{o}\,\in \OTS$.


 To define the rewrite rules, we first introduce the notion of patterns, which are particular simple open terms representing the left hand side of a rule.
\emph{Patterns}, ranged over by $\LeftPat$, are the linear simple open terms defined in the following way:
$$
\begin{array}{lcl}
   \LeftPat & \;\qqop{::=}\; & (\overline{a} \conc x \into \overline{b} \conc \overline{\LeftPat}\conc X)^\ell 
\end{array}
$$
where $\overline{a}$ and $\overline{b}$ are multisets of atoms, $\overline{\LeftPat}$ is a multiset of pattern, $x$ is  a wrap variable, $X$ is a content
variable and the label $\ell$ is called the \emph{type of the pattern}.
 The set of patterns is denoted by $\LeftPP$.  Patterns are intended to match with compartments. Note that we force \emph{exactly} one variable to occur in each compartment content and wrap.
 This prevents ambiguities in the instantiations needed to match a given compartment.\footnote{The presence of two (or more) variables  in the same compartment content or wrap, like in $\wr{x}{a \conc X \conc Y}{\ell}$, would introduce the possibility of matching the same path in different although equivalent ways. For instance we could match a term $\wr{x}{a \conc a \conc b \conc b}{\ell}$ with $X = a,~Y = b \conc b$ or $X = a \conc b,~Y =  b$, etc.}
 The linearity condition, in biological terms, corresponds to excluding that a transformation can depend on the presence of two (or more) identical (and generic) components in different compartments (see also~\cite{OP11}).

Some examples of patters are:
\begin{itemize}
 \item $\wr{x}{a \s b\s X}{\ell}$ which matches with all compartments of type $\ell$ containing at least one occurrence of $a$ and one of $b$.
 \item  $\wr{x}{\wr{a \s y}{Y}{\ell_1} \s X}{\ell_2}$ which matches with compartments of type $\ell_2$ containing a compartment of type $\ell_1$ with at least an
 $a$ on its wrap.
 \end{itemize}

A \emph{rewrite rule} is a pair $(\LeftPat,\RightPat)$, denoted by $\srewRule{\LeftPat}{\RightPat}{}$, where $\LeftPat= (\overline{a} \conc x \into
\overline{b} \conc \overline{\LeftPat}\conc X)^\ell$ $\in$ $\LeftPP$ and $\RightPat = \wr{\ov{q}}{ \ov{\RightPat}}{\ell} \in \OT$ are such that $\VarOf(\RightPat) \subseteq
\VarOf(\LeftPat)$. Note that rewrite rules must respect the type of the involved compartments.
 A rewrite rule $\srewRule{\LeftPat}{\RightPat}{}$ then states that a compartment $\LeftPat \sigma$,
obtained by instantiating variables in $\LeftPat$ by some instantiation function $\sigma$, can be transformed into the compartment $\RightPat\sigma$ with
the same type $\ell$ of $\LeftPat$. Linearity is not required in the r.h.s. of a rule, thus allowing duplication.

A \emph{\short\ system} over a set $\AT$ of atoms and a set $\LT$ of labels is represented by a set $\QQ_{\AT,\LT}$ ($\QQ$ for short when $\AT$ and $\LT$
are understood) of rewrite rules over $\AT$ and $\LT$.

A \emph{transition} between two terms $t$ and $u$ of a  \short\ system $\QQ$ (denoted $t \ltrans{} u$) is defined by the following rule:
$$
\frac{ \srewRule{\LeftPat}{\RightPat}{} \in \QQ \quad\sigma \in \Sigma
           \quad C\in \CC }   
    { C[\LeftPat\sigma] \ltrans{}  C[\RightPat\sigma]}
$$
where $C[\LeftPat\sigma] = t$ and $C[\RightPat\sigma] = u$ (modulo $\equiv$).

In a rule $\Pat \! \srewrites{} \! o$ the pattern $\Pat$ represents a compartment containing the reactants of the reaction that will be simulated.
 The crucial
point for determining an application of the rule to the whole biological
ambient $t$ is the choice of the occurrences of compartments
matching with $\Pat$ (determining the
compartment in which this reaction will take place).

By our definitions of contexts, patterns, and substitutions, in a match $t=C[p\sigma]$ the substitution $\sigma$ defines the closer environment around the
pattern, while the context $C$ defines the more external environment. Note that the applicability of a rewrite rule depends on the type of the involved
compartments but not on the context in which it occur.  This correspond to the assumption that only the compartment type can influence the kind of
reaction that takes place in them but not their position in the system.

As shown in the following example, a same pattern can have more then one match in a term.

\begin{example}\label{ex:diff_cont}
In Figure~\ref{fig:ex_cont} we depict the matching of the pattern $p$ $=$ $\wr{ x}{a \conc b\conc X}{\ell}$  to the term $t = \wr{\emptyseq}{5a \conc 5b
\conc 10c \conc \wr{12m}{23a \conc 6b \conc 10c}{\ell} \conc \wr{14m}{17a \conc 6b \conc 10c}{\ell}}{\top}$. In the figure, the two different contexts and
instantiations are showed for this matching. In particular, we depicted in red the matching done by the substitutions and in black the role played by the
context. Namely, we have:
\begin{itemize}
\item $\sigma(x)=12m$ and $\sigma(X)=22a \conc 5b \conc 10c$,
\item $\sigma'(x)=14m$ and $\sigma'(X)=16a \conc 5b \conc 10c$,
\item $C=\wr{\emptyseq}{5a \conc 5b \conc 10c
\conc \phole \conc \wr{14m}{17a \conc 6b \conc 10c}{\ell} } {\top}$,
\item $C'=\wr{\emptyseq}{5a \conc 5b \conc 10c
\conc \wr{12m}{23a \conc 6b \conc 10c}{\ell} \conc \phole}{\top}$.
\end{itemize}
In blue we depicted the actual elements triggering the match.
\end{example}

\begin{figure*}
\centering
\subfigure[] {
\includegraphics[height=30mm]{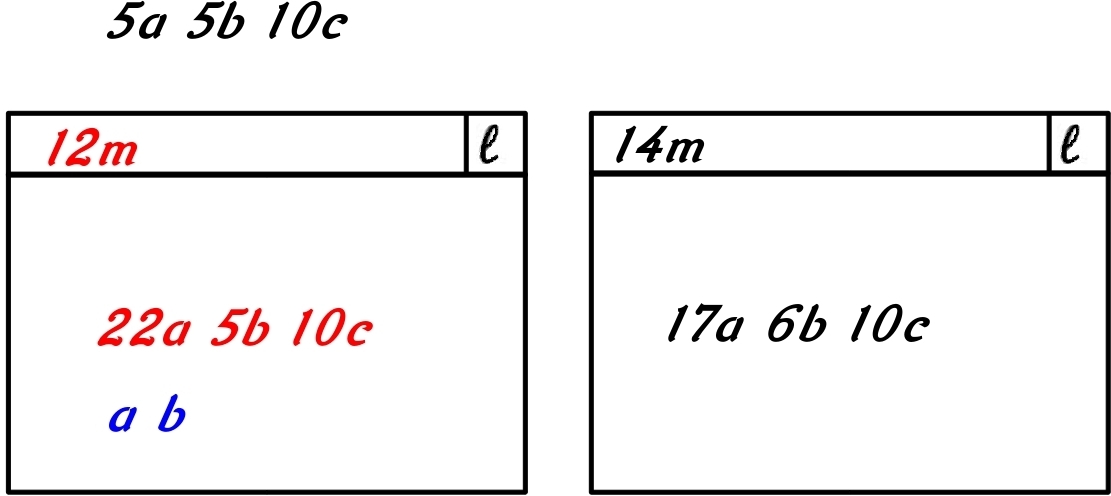}
}

\subfigure[] {
\includegraphics[height=30mm]{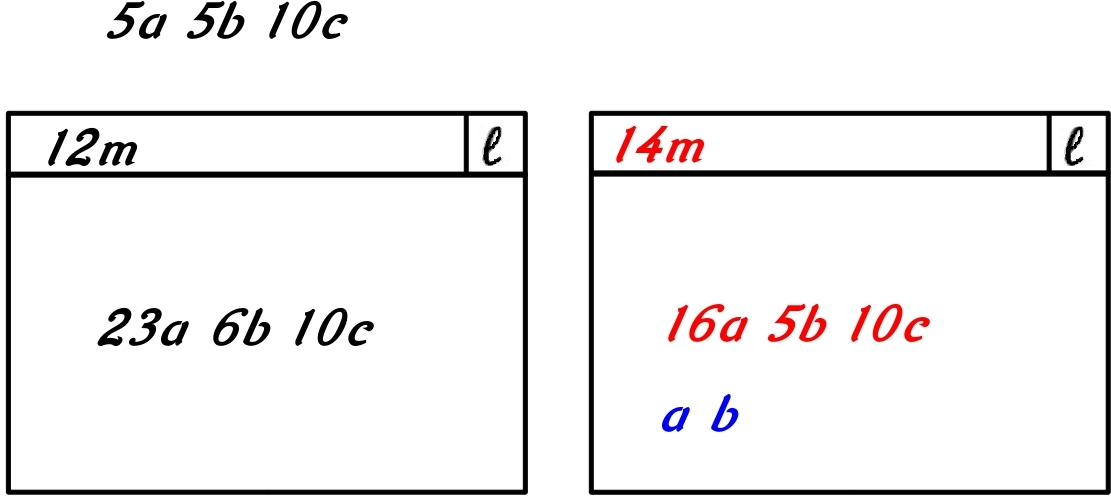}
}

\caption{\textbf{(a)} $\ov{t}=C[p\sigma]$ ; \textbf{(b)} $\ov{t}=C'[p\sigma']$ }
\label{fig:ex_cont}
\end{figure*}

\begin{remark} \label{DifferentOutcomes}   
For some rewrite rules $\ell : \LeftPat \srewrites{}\RightPat$ there may be, in general, different substitutions $\sigma$  such that
$\LeftPat\sigma \equiv t$ (for some term $t$)  but the results $\RightPat\sigma$ produced by them are different. Consider, for instance,
 the rewrite rule $(y \into a \conc (b \conc x \into X)^\ell\conc Y)^\top  \srewrites{} \wr{y}{(a\conc
b \conc x \into X)^\ell\conc Y)}\top$ modelling a catalyzed membrane joining at top level. In this case, a term $t= \wr{\emptyseq}{a \conc (b \conc b
\into c)^\ell \conc (b \into c)^\ell}{\top}$ can make a transition in two different terms, depending on which membrane will be joined by the element $a$.
Namely, $\wr{\emptyseq}{ (a \conc b \conc b \into c)^\ell \conc (b \into c)^\ell}{\top}$, given an instantiation $\sigma$ such that $\sigma(x)=b$ and
$\wr{\emptyseq}{(b \conc b \into c)^\ell \conc (a \conc b \into c)^\ell}{\top}$, given an instantiation $\sigma'$ such that $\sigma'(x)=\emptyseq$.
We remark that this can happen only when compartments are involved in the rewriting. We will need to take it into account in the stochastic approach.
\end{remark}

\begin{notation}[Rules that involve only content or wrap]\label{nota:ShortForRules}
Usually, rules involve only the content or the wrap of a compartment. Moreover, in a rule $\wr{\ov{a} \conc
x}{\overline{b} \conc \overline{\Pat}\conc X}{\ell} \! \srewrites{} \! \wr{\ov{q}}{ \ov{o}}{\ell} $ very often $X$ has single occurrence, at top level, in $\ov{o}$ and $x$ in
$\ov{q}$. We therefore introduce the following notations:
\begin{itemize}
\item
$\ell: \overline{a} \conc \overline{\LeftPat}   \csrew{}   \ov{o}$
(or simply $\ell: \overline{a} \conc \overline{\LeftPat}   \srew{}   \ov{o}$) as a short for
 $\wr{x}{\overline{a} \conc \overline{\LeftPat}\conc X}{\ell}   \srew{}  \wr{x}{ \ov{o} \conc X}{\ell} $, and
 \item
  $\ell: \overline{a}  \wsrew{}   \ov{b}$ as a short for
 $\wr{\overline{a} \conc x}{X}{\ell}   \srew{}   \wr{\ov{b} \conc x}{ X}{\ell} $
\end{itemize}
where $x$ and $X$ are canonically chosen variables not occurring in $\overline{a}$, $\overline{\LeftPat}$, $\overline{o}$ or $\overline{b}$.
Note that, according to Notation~\ref{nota:WholeTerm}, rules of the shape $\top: \overline{a}  \wsrew{}   \ov{b}$ are forbidden
(since the top level compartment must always have an empty wrap).
\end{notation}

\section{Modelling Guidelines}\label{sec:modelguide}
In this section we  give some explanations and general hints about how \short\ could be used to represent the behaviour of various biological systems.
Here, entities are
represented by terms of the rewrite system, and events by
rewrite rules.

First of all, we should select the biomolecular entities of
interest. Since we want to describe cells, we consider molecular
populations and membranes. Molecular populations are groups of
molecules that are in the same compartment of the cells and inside
them. As we have said before, molecules can be of many types: we
classify them as proteins, chemical moieties and other molecules.

Membranes are considered as elementary objects: we do not describe
them at the level of the phospholipids they are made of. The only
interesting properties of a membrane are that it may have a
content (hence, create a compartment) and that in its phospholipid
bilayer various proteins are embedded, which act for example as
transporters and receptors. Since membranes are represented as
multisets of the embedded structures, we are modeling a fluid
mosaic in which the membranes become similar to a two-dimensional
liquid where molecules can diffuse more or less
freely~\cite{SN72}.

Compartment labels are useful to identify the kind of a compartment. For example, we may use compartment labels to denote a nucleus within a cell, the different organelles, etc..

Table~\ref{tab:guidelines-events} lists the guidelines for the abstraction into CWC rules of some basic biomolecular
events, some of which will be used in our applications.\footnote{The prefixes specifying the labels associated to each rule are omitted
for simplicity, just notice that each rule shown in the table can be specified to apply only within a given type of compartment.} Entities
are associated with CWC terms: elementary objects (genes, domains, etc...) are modelled as atoms, molecular populations as CWC terms, and membranes as
atom multisets. Biomolecular events are associated with CWC rewrite rules.

\begin{table}[t]
\begin{center}
\begin{footnotesize}
\begin{tabular}{|l|l|}
\hline
{\bf Biomolecular Event} & {\bf \short\ Rewrite Rules} \\
\hline \hline State change (in content) &
    $a \csrew{} b$ \\
\hline State change (on membrane) &
    $ a \wsrew{} b $ \\
    \hline Complexation (in content) &
    $a \conc b \csrew{} c$ \\
    \hline Complexation (on membrane) &
    $a \conc b \wsrew{} c$ \\
      &
        $ a \conc (b \conc x \into X)^{\ell} \csrew{}  (c \conc x \into X)^{\ell} $ \\
    &  $ (b \conc x \into a \conc X)^{\ell} \csrew{}  (c \conc x \into X)^{\ell} $ \\
    \hline
    Decomplexation (in content) &
    $c \csrew{} a \conc b$
    \\
\hline Decomplexation (on membrane) &
       $c \wsrew{} a \conc b$
      \\
     & $ (c \conc x \into X)^{\ell} \csrew{}  a \conc (b \conc x \into X)^{\ell} $
    \\
    &  $ (c \conc x \into X)^{\ell} \csrew{}  (b \conc x \into a \conc X)^{\ell} $
    \\
\hline Membrane crossing &
    $ a \conc (x \into X)^{\ell} \csrew{}  (x \into a \conc X)^{\ell} $\\
    & $ (x \into a \conc X)^{\ell} \csrew{}  a \conc (x \into X)^{\ell} $\\
\hline
Catalyzed membrane crossing &  $ a \conc (b \conc x \into X)^{\ell} \csrew{}  (b \conc x \into a \conc X)^{\ell} $\\
 & $ (b \conc x \into a \conc X)^{\ell} \csrew{}  a \conc (b \conc x \into X)^{\ell} $\\
\hline Membrane joining &
    $ a \conc (x \into  X)^{\ell} \csrew{}  (a \conc x \into X)^{\ell} $\\
    & $ (x \into a \conc X)^{\ell} \csrew{} (a \conc x \into X)^{\ell}$\\
\hline Catalyzed membrane joining &
    $ a \conc (b\conc x \into X)^{\ell} \csrew{} (a \conc b \conc x \into X)^{\ell} $ \\
 &
    $ (b\conc x \into a \conc X)^{\ell} \csrew{} (a \conc b \conc x \into X)^{\ell} $ \\
    & $ (x \into a \conc b \conc X)^{\ell} \csrew{}(a \conc x \into b \conc X)^{\ell} $\\
\hline Compartment state change &
    $  ( x \into X)^{\ell} \csrew{} ( x \into X)^{\ell'} $ \\
\hline
\end{tabular}
\end{footnotesize}
\end{center}
\caption{Guidelines for modelling biomolecular events in \short}\label{tab:guidelines-events}
\end{table}

The simplest kind of event is the change of state of an elementary object. Then, there are interactions between molecules: in particular complexation,
decomplexation and catalysis. Interactions could take place between simple molecules, depicted as single symbols, or between membranes and molecules: for
example a molecule may cross or join a membrane. There are also interactions between membranes: in this case there may be many kinds of interactions
(fusion, vesicle dynamics, etc\ldots). Finally, we can model a state change of a compartment (for example a cell moving onto a new phase during the cell
cycle), by updating its label.\footnote{Note that in this case it is not the label of the pattern external compartment that is changing (which is omitted
here for simplicity).} Changing a label of a compartment implies changing the set of rules applied to it. This can be used, e.g., to model the different
activities of a cell during the different phases of its cycle.

\section{Turing Completeness}\label{sect:Turing}

The  CWC is Turing Complete. In the following we sketch how
Turing Machines can be simulated by  CWC models.

\begin{theorem}[Turing Completeness] The class of  CWC models is Turing complete.
\end{theorem}
\begin{proof} (Sketch)
A Turing machine $T$ over an alphabet $\Sigma\cup\set{\blank}$
(where $\blank$ represents the blank) with a set $S$ of states can be
simulated by a CWC system $\QQ_T$ in the following way. \\
Take $\AS = S \cup\Sigma \cup \set{\blank, l, r}$, where $l$ and $r$ are special symbols to represent the left and right ends of the tape. The tape of the
Turing machine can be represented by a sequence of nested compartments, with the same label $\ell$, whose wraps consist of a single atom (representing a symbol of the tape). The
content of each compartment defined in this way represents a right suffix of the written portion of the tape, while the atom on the wrap represents its
initial (w.r.t. the suffix) symbol. In each term representing a tape there is exactly one compartment which contains a state (the present state). For
example, the term $(l \into ( a \into s ~(b \into (r\into \emptyseq))))$ represents the configuration
in which the tape is $...\blank,a,b,\blank....$ (the remaining positions are blank), the machine is in state $s$ and the head is
positioned on $b$.\footnote{For simplicity we omit compartment labels.}
Rewrite rules are then used to model the machine evolution. We can define rules creating new blanks when needed (at the tape ends) to mimic a possibly
infinite tape. For instance a transition $(s, b) \rightarrow (s',c,right)$ (meaning that in state $s$ with a $b$ on the head the machine goes in state
$s'$ writing $c$ on the tape and moving the head right) is represented by the rules:\footnote{In these rules, in a Turing Machine simulation, the variable
$Y$ is always instantiated with the empty multiset.}
    $$\begin{array}{ll}
   (1) & \ell: s~(b~y\into (x \into Z)~ Y ) \red  (c~y \into s' ~ (x \into Z) ~ Y) \\
    (2)&   \ell: s~(b~y\into (r\into Z)~Y)  \red  (c~y \into
    s'~(\blank\into (r \into Z ))~Y)
    \end{array}$$
The second rule represents the case that b is the rightmost non blank symbol and so a new blank must be introduced in the simulation.\footnote{Note that
rule (1) could be applied also in this case but it would lead the system in a configuration (containing a subterm $s'~ (r\into t)$ for some t) from which
no further move could be possible. However situations of this kind can be easily avoided with a little complication of the encoding.} By construction, the
system $\QQ_T$, in which computations are defined by the transitive closure of the transition relation, correctly represents the behaviour of $T$.\end{proof}

\section{Quantitative Semantics for \short}\label{STO_SEM}

In order to make the formal framework suitable to model quantitative
aspects of biological systems we must associate to each rewriting step a numerical
parameter (the step \emph{rate}) which determines
the time spent between two successive system states and, in order to represent faithfully the
system evolution, the frequency with
which each interaction will take place.

In this section we introduce two quantitative simulation methods for CWC based respectively on a stochastic simulation method
and on the deterministic solution of ordinary differential equations (ODE).

In the stochastic framework the rate of the reduction steps are used as the parameters to determine, stocastically,  the next system configutation and the
time spent for reaching it. The system is then descibed as a Continuous Time Markov Chain (CTMC) \cite{Pra62}. This allows to simulate the system
evolution by means of standard simulation algorithms (see e.g.\cite{G77}). Stochastic simulation techniques can be applied to all CWC systems but, in
several cases, at a high computational cost. The deterministic method based on ODE, is computationally more efficient, but can be applied, in general,
only to systems in which compartments are absent or have a fixed, time-independent, structure. These two approaches, presented separately in this section,
will be integrated in the next section defining an hybrid simulation algorithm for CWC that keeps the generality of the stochastic approach but can reduce
its computational cost exploiting, when possible, the efficiency of the ODE simulation method.

In our calculus we will represent the speed of a reaction as a
\emph{rate function} having a parameter depending on the overall content of the compartment in which the reaction takes place. This allows to
tailor the reaction rates on the specific characteristics of the system, as for
instance when representing nonlinear reactions as it happens for Michaelis--Menten kinetics, or to describe more complex interactions involving compartments that may not follow the standard mass action
rate. This latter, more classical, collision based stochastic semantics (see~\cite{G77}) can be encoded as a particular choice of the rate
function (see Section~\ref{sect:LMA}). A similar approach is used in~\cite{DGT09b} to model reactions with inhibitors and catalysers in a single rule.

Obviously some care must be taken in the choice of the rate function: for
instance it must be complete (defined on the domain of the
application) and nonnegative. This properties are also enjoyed by the function
representing the law of mass action.

\begin{definition}
  A \emph{quantitative} rewrite rule is a triple $(p, o, f)$, denoted
$p\srewrites{f} o$, where $(p,o)$ is a rewrite rule and $f:\Sigma
\rightarrow \mathbb{R}^{\geq 0}$ is the \emph{rate function} associated to
the rule.\footnote{The value $0$ in the codomain of $f$ models the situations in which
the given rule cannot be applied, for example when the particular
environment conditions forbid the application of the rule.}
\end{definition}

The rate function takes an instantiation $\sigma$ as parameter. Such an instantiation models the actual compartment content determining the structure of the
environment in which the l.h.s. of a rule matches and that may actively influence the rule application. Notice that, by Remark \ref{DifferentOutcomes},
different instantiations that allow the l.h.s. $p$ of a rule to match a term $\ov{t}$ can produce different outcomes which could determine different rates in the associated transitions.

In the following we will use the function $\OO:\TS \times \TS
\rightarrow \mathbb{N}$ to count the occurrences of a term within another. Namely, $\OO(\ov{t},\ov{u})$ returns the number of occurrences of the term $\ov{t}$ within the term $\ov{u}$.

\begin{example}\label{ex:multiterm2}
Consider again the term given in Remark~\ref{DifferentOutcomes}. If the rate function of the rewrite rule is defined as $f(\sigma)=0.0002 \cdot (
\OO(b,\sigma(x)+1)  )$, the initial term $\ov{t}$ results in $ (a \conc b \conc b \into c)^\ell \conc (b \into c)^\ell$ with a rate $0.0004$ and in the
term $(b \conc b \into c)^\ell \conc (a \conc b \into c)^\ell$ with rate $0.0002$.
\end{example}



We already mentioned that equipping rewrite rules with a function leads the definition of a stochastic semantics that can abstract from the classical one based on collision analysis (practical for very low level simulations, for example chemical interactions), and allows defining more complex rules (for higher simulation levels, for example cellular or tissue interactions) which might follow different probability distributions. An intuitive example could be a simple membrane interaction: In the presence of compartments, a system could not be considered, in general, as well stirred. In such a case, the classical collision based analysis could not always produce faithful simulations and more factors (encapsulated within the context in which a rule is applied) should be taken into account.

A \emph{quantitative \short\ system} over a set $\AT$ of atoms and a set $\LT$ of labels is represented by a set $\QQ_{\AT,\LT}$ ($\QQ$ for short when $\AT$ and $\LT$
are understood) of quantitative rewrite rules over $\AT$ and $\LT$.

\subsection{Stochastic Evolution}

In the stochastic framework, the rate of a transition used as the parameter of an exponential distribution modeling the time spent to complete the
transition. A quantitative \short\ system $\QQ$  defines a Continuous Time Markov Chain (TCMC) in which the rate of a transition $C[\LeftPat\sigma]
\ltrans{}  C[\RightPat\sigma]$ is given by $f(\sigma)$ where $\srewRule{\LeftPat}{\RightPat}{f} \in \QQ$ is the quantitative rule which determines the
transition. The so defined CTMC determines the \emph{stochastic reduction semantics} of \short.


When applying a simulation algorithm to a \short\ system we must take into account, at a given time, all the system transitions (with their associate
rates) that are possible at that point. They are identified by:
\begin{itemize}
\item the rewrite rule applied;
\item context which selects the comparment in which the rule is applied;
\item the outcome of the transition (see Remark~\ref{DifferentOutcomes}).
\end{itemize}

\begin{remark}\label{rm:eq_comp}
We must take some care in identifying transitions involving compartments. For instance, in the term  $\ov{t}=25m \conc 8a \conc \wr{10c}{24a \conc
20b}{\ell} \conc \wr{10c}{24a \conc 20b}{\ell}$ shown in Figure~\ref{fig:eq_comp} (a)
there are two compartments that are exactly the same. If we apply to $\ov{t}$ the rule $\ell: a \conc b   \srew{}   c$ we obtain the term $\ov{t'}$ shown
in Figure~\ref{fig:eq_comp} (b). Actually, starting from $\ov{t}$ there are two compartments on which the rule can be applied, producing the same term
$\ov{t'}$ (up to structural congruence). Although the transition is considered as one (up to structural congruence), the quantitative evolution must take
this possibility into account by counting two transitions.
\end{remark}

\begin{figure*}
\centering \subfigure[] {
\includegraphics[height=30mm]{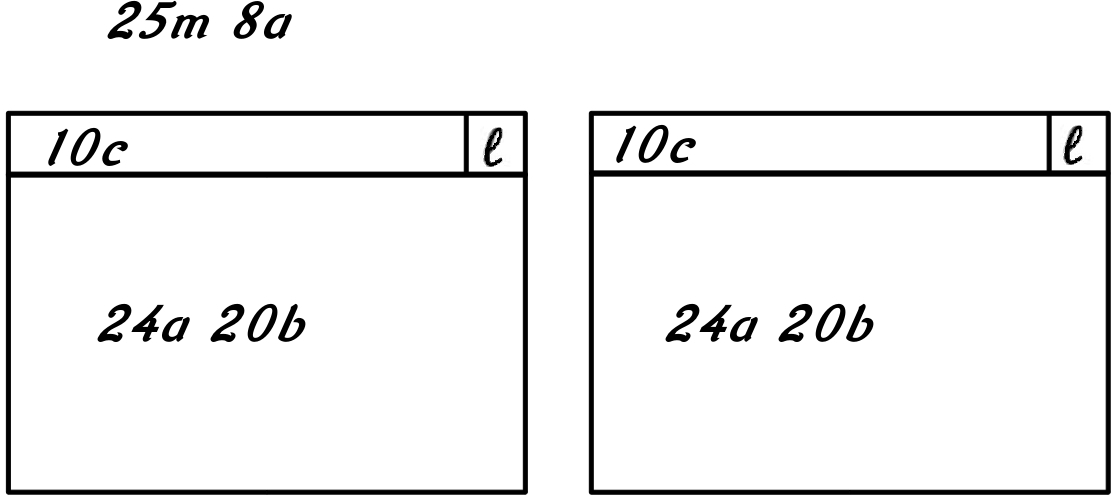}
}

\subfigure[] {
\includegraphics[height=30mm]{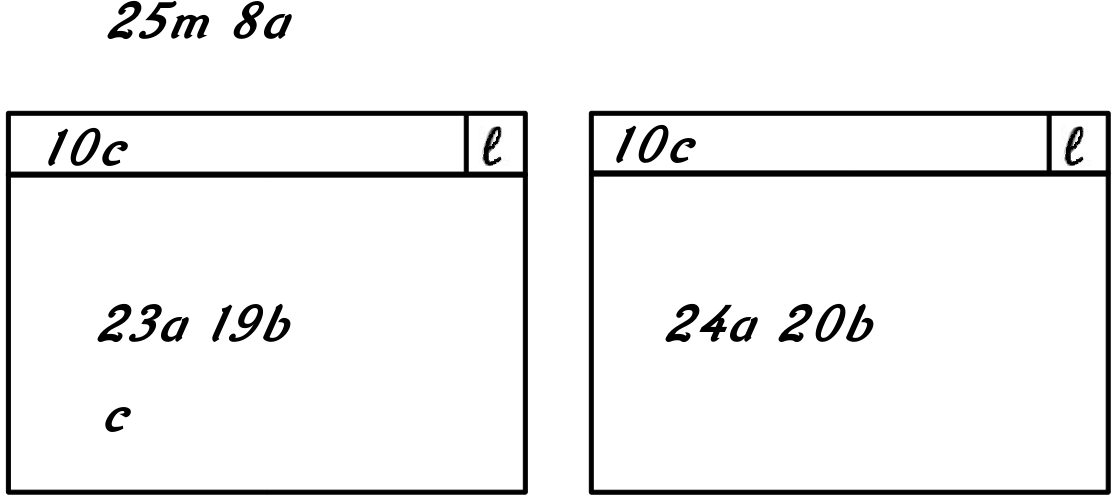}
}

\caption{\textbf{(a)} $\ov{t}$ ; \textbf{(b)} $\ov{t'}$ } \label{fig:eq_comp}
\end{figure*}

From the transition rates we can define, following a standard simulation procedure \cite{G77}, the exponential probability distribution of the moment in
which the next reaction will take place and the probability distribution of the next transition that will take place.

In particular, given a term $\ov{t}$, a global time $\delta$ and all the transitions $e_1,\ldots,e_M$ (counted as mentioned in Remark~\ref{rm:eq_comp})
that can be applied to $\ov{t}$, with rates $\pi_1,\ldots,\pi_M$ such that $\pi=\sum_{i=1}^M \pi_i$, the simulation procedure allows to determine, following
Gillespie's method:
\begin{enumerate}
\item The time $\delta+\tau$ at which the next stochastic transition will occur, randomly
chosen with $\tau$ exponentially distributed with parameter $\pi$;
\item The transition $e_i$ that will occur at time $\delta+\tau$, randomly
chosen with probability $\frac{\pi_i}{\pi}$.
\end{enumerate}

\subsection{Mass Action Law}\label{sect:LMA}

Gillespie's approach (see~\cite{G77}) simulates the time evolution of a
chemically reacting system by determining when the next reaction will occur and
what kind of reaction it will be. Kind and time of the next
reaction are computed on the basis of a stochastic reaction constant.

Gillespie's stochastic simulation algorithm is defined essentially for well stirred systems, confined to a constant volume and in thermal equilibrium at some constant temperature. In these conditions we can describe the system's state by specifying only the molecular populations, ignoring the positions and velocities of the individual molecules.
Different approaches such as Molecular Dynamics, Partial Differential Equations or Lattice-based methods are required in case of molecular crowding, anisotropy of the medium or canalization.

We might restrict CWC in order to match Gillespie's framework. Since we just need to deal with simple molecular populations, we restrict our calculus to multisets of atoms. We denote with $\TSG$ the infinite set of CWC terms representing Gillespie's molecular populations.

The usual notation for chemical reactions can be expressed by:\footnote{Here, as in~\cite{G77}, we define general chemical reactions with an arbitrary number of reagents. Notice that reactions involving more than two reagents are quite uncommon and could be expressed as a chain of reactions involving just two reagents.}
\begin{equation}\label{eq:cr}
n_1 a_1 + \ldots + n_\rho a_\rho
\mathop{\rightharpoonup}^{k}
n_1' b_1 + \ldots + n_{\gamma}' b_{\gamma}
\end{equation}
where $a_i$ and $b_i$ are the reagents and product molecules, respectively,
$n_i, n_i'$ are the stoichiometric coefficients
and $k$ is the kinetic constant.

We restrict to rewrite rules modelling chemical reactions as in reaction~\ref{eq:cr}.  A chemical reaction of the form~\ref{eq:cr} can be expressed by the
following CWC rewrite rule:

\begin{equation}\label{eq:crr}
\ell: n_1  a_1 \conc \ldots \conc n_\rho  a_\rho
\srewrites{f}
n_1'  b_1 \conc \ldots \conc n_{\gamma}' b_{\gamma}
\end{equation}
which is short for $\wr{x}{n_1  a_1 \conc \ldots \conc n_\rho  a_\rho \conc X}{\ell} \srewrites{f} \wr{x}{n_1'  b_1 \conc \ldots \conc n_{\gamma}'
b_{\gamma} \conc X}{\ell}$, where the rate function $f$ of rule~\ref{eq:crr} is suitably defined to model Gillespie's collision based stochastic
simulation algorithm. In particular, the collision based framework defined by Gillespie leads to binomial distributions of the reagents involved. Namely,
we define the rate function $f$ as:

\begin{equation}\label{eq:rf}
f(\sigma)=  \binom{\OO(a_1,\sigma(X))+n_1}{n_1} \cdot \ldots \cdot \binom{\OO(a_\rho,\sigma(X))+n_\rho}{n_\rho} \cdot k
\end{equation}
where $k$ is the kinetic constant of the modelled chemical reaction.

In many practical situations, this is approximated as:
\begin{equation}\label{eq:approx}
\frac{(\OO(a_1,\sigma(X))+n_1)\cdot \ldots \cdot (\OO(a_\rho,\sigma(X)+n_\rho))}{n_1\cdot \ldots \cdot n_\rho}\cdot k
\end{equation}

By construction, the following holds.

\begin{fact}
Molecular populations defined as $\TSG$ terms with a fixed set of rules of the form given by reaction~\ref{eq:cr}, represented by  rule~\ref{eq:crr},
interpret into the stochastic semantics of CWC the law of mass action within  Gillespie's framework for the evolution of chemically reacting systems.
\end{fact}

\begin{notation}
We will denote \emph{biochemical rewrite rules} as defined in rule~\ref{eq:crr} with the simplified notation:
$$
\ell:\overline{a} \srewrites{k} \overline{b}
$$
where $\overline{a}$ and $\overline{b}$ are multisets of atomic elements, and the rate function is represented by just the kinetic constant of the
chemical reaction.
\end{notation}

When the counting is done with the law of mass action, we will extend the simplified notation for biochemical rewrite rules (using a constant rate instead
of a function) also for rules involving compartments:
$$
\ell: \overline{a} \conc \overline{\LeftPat}   \srewrites{k}   \ov{o}
$$

\begin{example}
Given a term $\ov{t}= 15 g \conc 7 m \conc 1928 n$ and the biochemical rewrite rule $\TOP: g \conc m \srewrites{0.002} n$, the
following transitions generates from the stochastic semantics interpreted under Gillespie's assumptions: $\ov{t} \ltrans{0.21} 14  g \conc 6  m \conc 1929 n$, where $ 0.21 = \frac{(14+1)\cdot (6+1)}{1\cdot 1}\cdot 0.002 $.
\end{example}

\chapter{Modelling Biological Systems}
\label{case}
In this chapter we present an application of CWC modelling Ammonium Transporters in the Arbuscular Mycorrhiza Symbiosis.

\section{Ammonium Transporters in AM Symbiosis}

Given the central role of agriculture in worldwide economy, several
ways to optimize the use of costly artificial fertilizers are now
being actively pursued. One approach is to find methods to nurture
plants in more ``natural'' manners, avoiding the complex chemical
production processes used today. In the last decade the Arbuscular
Mycorrhiza (AM), the most widespread symbiosis between plants and
fungi, got into the focus of research because of its potential as a
natural plant fertilizer. Briefly, fungi help plants to acquire
nutrients as phosphorus (P) and nitrogen (N) from the soil whereas
the plant supplies the fungus with energy in form of
carbohydrates~\cite{Par08}. The exchange of these nutrients is
supposed to occur mainly at the eponymous arbuscules, a specialized
fungal structure formed inside the cells of the plant root. The
arbuscules are characterized by a juxtaposition of a fungal and a
plant cell membrane where a very active interchange of nutrients is
facilitated by several membrane transporters. These transporters are
surface proteins that facilitate membrane crossing of molecules
which, because of their inherent chemical nature, are not freely
diffusible.

Since almost all cells in the majority of multicellular
organisms share the same genome, modern theories point out that
morphological and functional differences between them are mainly
driven by different genes expression~\cite{MBOC1}. Thanks to the
latest experimental novelties~\cite{TB99,NCB01} a precise analysis of
which genes are expressed in a single tissue is attainable; therefore it
is possible to identify genes that are pivotal in specific
compartments and then study their biological function. Following
this route a new membrane transporter has been discovered by
expression analysis and further characterized~\cite{GUE09}. This
transporter is situated on the plant cell membrane which is directly
opposite to the fungal membrane, located in the arbuscules. Various
experimental evidence points out that this transporter binds to an
$NH_{4}^+$ moiety outside the plant cell, deprotonates it, and
mediates inner transfer of $NH_{3}$, which is then used as a
nitrogen source, leaving an $H^+$ ion outside. The AM symbiosis is
far from being unraveled: the majority of fungal transporters and
many of the chemical gradients and energetic drives of the symbiotic
interchanges are unknown. Therefore, a valuable task would be to
model \emph{in silico} these conditions and run simulations against
the experimental evidence available so far about this transporter.
Conceivably, this approach will provide biologists with working
hypotheses and conceptual frameworks for future biological
validation.

\begin{figure}
\begin{center}
\begin{minipage}{0.98\textwidth}
\begin{center}
\includegraphics[height=130mm]{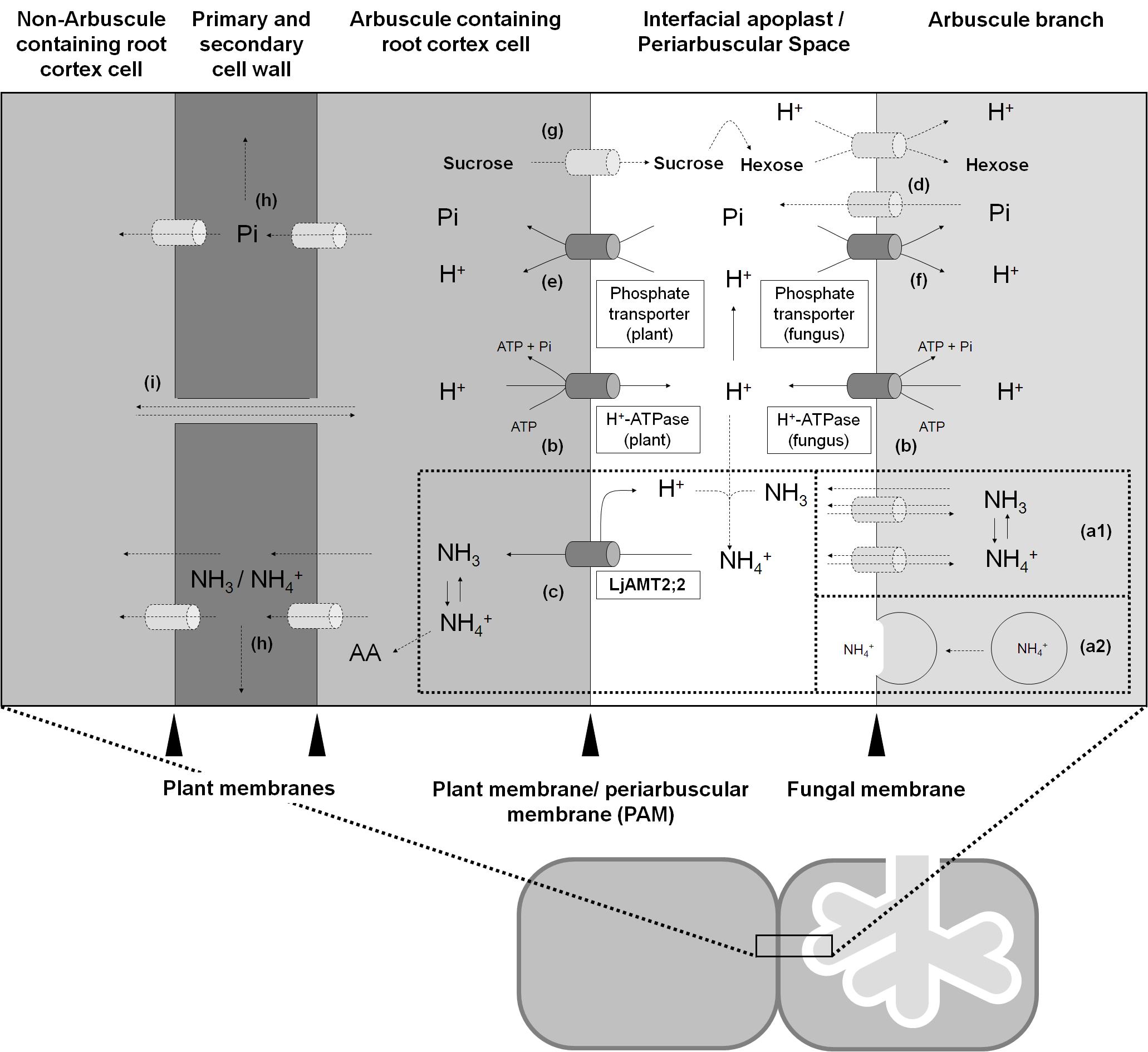}
\end{center}
\caption{Nitrogen, phosphorus and carbohydrate exchanges at the
mycorrhizal interface} \label{fig:symbschem}
\end{minipage}
\end{center}
\end{figure}

The scheme in Figure~\ref{fig:symbschem} (taken from~\cite{GUE09})
illustrates nitrogen, phosphorus and carbohydrate exchanges at the
mycorrhizal interface according to previous works and the results
of~\cite{GUE09}.
\textbf{(a1-2)} $NH_3/NH_4^+$ is released in the arbuscules from
arginine which is transported from the extra- to the intraradical
fungal structures~\cite{Gov05}. $NH_3$/$NH_4^+$ is released by so
far unknown mechanisms (transporter, diffusion \textbf{(a1)} or
vesicle-mediated \textbf{(a2)}) into the periarbuscular space
(PAS) where, due to the acidic environment, its ratio shifts
towards $NH_4^+$ ($>99.99\%$). \textbf{(b)} The acidity of the
interfacial apoplast is established by plant and fungal
$H^+$-ATPases~\cite{Hau05,Bal07} thus providing the energy for
$H^+$-dependent transport processes. \textbf{(c)} The $NH_4^+$ ion
is deprotonated prior to its transport across the plant membrane
via the LjAMT2;2 protein and released in its uncharged $NH_3$ form
into the plant cytoplasm. The $NH_3$/$NH_4^+$ acquired by the
plant is either transported into adjacent cells or immediately
incorporated into amino acids (AA). \textbf{(d)} Phosphate is
released by so far unknown transporters into the interfacial
apoplast. \textbf{(e)} The uptake of phosphate on the plant side
then is mediated by mycorrhiza-specific
Pi-transporters~\cite{JPH07,GUE09}.\footnote{Where Pi stands for inorganic phosphate.} \textbf{(f)} AM fungi might
control the net Pi-release by their own Pi-transporters which may
reacquire phosphate from the periarbuscular space~\cite{Bal07}.
\textbf{(g)} Plant derived carbon is released into the PAS
probably as sucrose and then cleaved into hexoses by sucrose
synthases~\cite{Hon03} or invertases~\cite{SRH06}. AM fungi then
acquire hexoses~\cite{Sha95,Sol97} and transport them over their
membrane by so far unknown hexose transporters. It is likely that
these transporters are proton co-transporter as the GpMST1
described for the glomeromycotan fungus Geosiphon
pyriformis~\cite{NT00}. Exchange of nutrients between arbusculated
cells and non-colonized cortical cells can occur by apoplastic
\textbf{(h)} or symplastic \textbf{(i)} ways.

\section{CWC Model}\label{SecModel}

We focus our investigation on the sectors labelled
with \textbf{(c)}, \textbf{(a1)} and \textbf{(a2)}. Namely, we
will present CWC models for the equilibrium between $NH_4^+$ and
$NH_3$ and the uptake by the LjAMT2;2 transporter \textbf{(c)},
and the exchange of $NH_4^+$ from the fungus to the interspatial
level \textbf{(a1-2)}. We will also analyze LjAMT2;2 role in
the AM symbiosis by comparing it with another known ammonium transporter,
LjAMT1;1.
The choice of CWC is motivated by the fact
that membranes, membrane elements (like LjAMT2;2) and the involved
reactions can be represented in it in a quite natural way.

The simulations illustrated in this section are done with the CWC
prototype simulator~\cite{HCWC_SIM}. In the following we will use
a more compact notation to represent multisets of the same atom,
namely, we will write $a\times n$ to denote the multiset of $n$ atomic elements $a$.

\subsection{$NH_3$/$NH_4^+$ Equilibrium}
We decided to
start modelling a simplified pH equilibrium, at the interspatial level (right part of section \textbf{(c)} in Figure~\ref{fig:symbschem}),
without
considering $H_2O$, $H^+$ and $OH^-$; therefore we tuned the reaction rates in
order to reach the correct percentages of $NH_3$
over total $NH_3/NH_4^+$ in the different compartments.
Like these all the rates and initial terms used in this work are obtained by manual adjustments made
looking at the simulations results and trying to keep simulations times acceptable - we plan to refine these rates
and numbers in future work to reflect biological data when they become available.
Following~\cite{GUE09}, we consider an extracellular pH of 4.5~\cite{Gut00}.
In such conditions, the percentage of molecules of $NH_3$ over the sum $NH_3 + NH_4^+$ should be around $0.002$.
The reaction we considered is the following:
$$
NH_3 \overset{k_1}{\underset{k_2}{\rightleftharpoons}}
NH_4^+
$$
with $k_1=0.018 \times 10^{-3}$ and $k_2=0.562 \times 10^{-9}$. One can
translate this reaction with the CWC rules:
\begin{equation*}
\tag*{(R1)} \top: NH_3 \srewrites{k_1} NH_4^+    
\end{equation*}
\begin{equation*}
\tag*{(R2)} \top: NH_4^+ \srewrites{k_2} NH_3
\end{equation*}
In Figure~\ref{fig:NH3eqNH4} we show the results of this first
simulation given the initial term $\ov{t}=NH_3\times
138238 \conc NH_4^+\times 138238$.

\begin{figure}[t]
\begin{center}
\begin{minipage}{0.98\textwidth}
\begin{center}
\includegraphics[width=80mm]{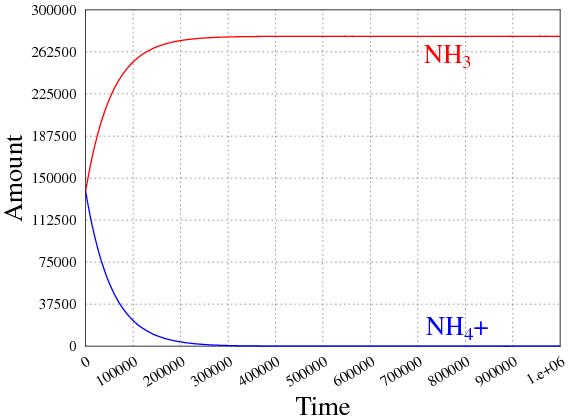}
\end{center}
\vspace{-0.5cm}
\caption{Extracellular equilibrium between $NH_3$ and $NH_4^+$.}
\label{fig:NH3eqNH4}
\end{minipage}
\end{center}
\end{figure}

This equilibrium is different at the intracellular level (pH
around 7 and 8)~\cite{GBD92}, so we use two new rules to model the
transformations of $NH_3$ and $NH_4^+$ inside the cell (labelled $\ell$), namely:
\begin{equation*}\tag*{(R3)} \ell: NH_3 \srewrites{k_1} NH_4^+
\end{equation*}
\begin{equation*}\tag*{(R4)} \ell: NH_4^+\srewrites{k_2'} NH_3
\end{equation*}
where $k_2'=0.562 \times 10^{-6}$.

\subsection{LjAMT2;2 Uptake}

We can now present the CWC model of the uptake of the LjAMT2;2
transporter (left part of section \textbf{(c)} in
Figure~\ref{fig:symbschem}). We add a compartment modelling
an arbusculated plant cell. Since we are only interested in the
work done by the LjAMT2;2 transporter, we consider a membrane
containing this single element. The work of the transporter is
modelled by the rule:
\begin{equation*}\tag*{(R5)} \top:
NH_4^+ \conc (LjAMT2 \conc x \into X)^\ell \srewrites{k_t} H^+ \conc (LjAMT2 \conc x \into X \conc NH_3)^\ell
\end{equation*}
where $k_t=0.1 \times 10^{-5}$.

We can investigate the uptake rate of the transporter at different
initial concentrations of $NH_3$ and $NH_4^+$.
Figure~\ref{fig:nh4_ljamt2_noNH3} and
Figure~\ref{fig:nh4_ljamt2_noNH4} show the results for the initial
terms:
$$
\begin{array}{ll}
\ov{t}= & NH_3\times 776 \conc NH_4^+\times 276400
\conc (LjAMT2 \into \emptyseq)^\ell\\
\ov{u}= & NH_3\times276400 \conc
NH_4^+\times776 \conc (LjAMT2 \into \emptyseq)^\ell
\end{array}
$$
where the graphs above represent the whole simulations, while the ones below
are a magnification of their initial segment.

\begin{figure}[t]
\begin{center}
\begin{minipage}{0.98\textwidth}
\begin{center}
\includegraphics[width=77mm]{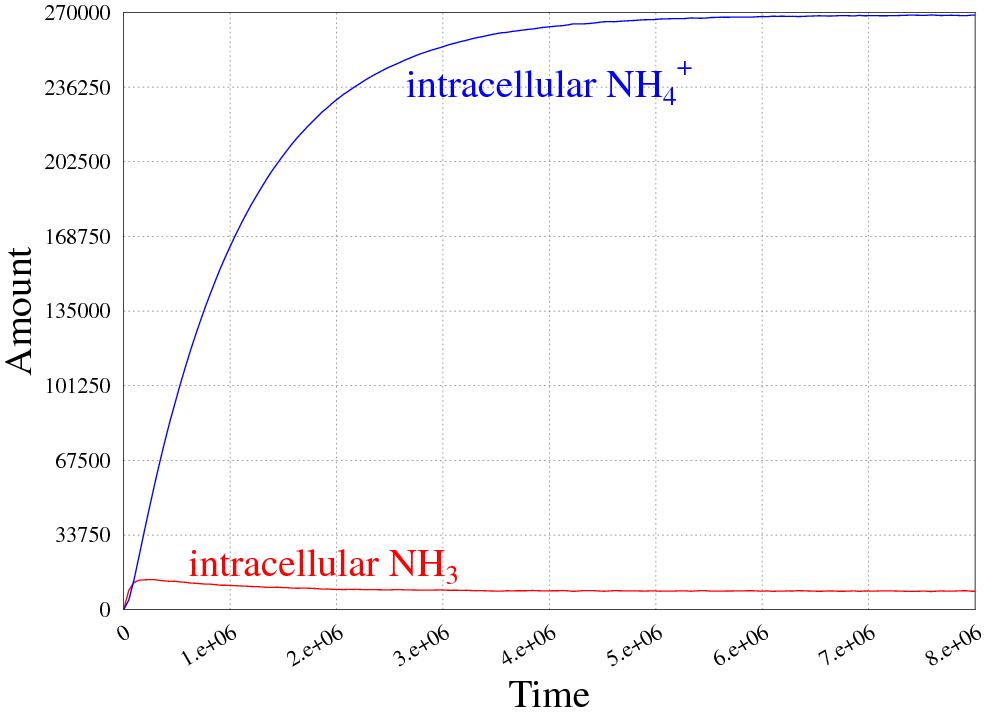}
\includegraphics[width=77mm]{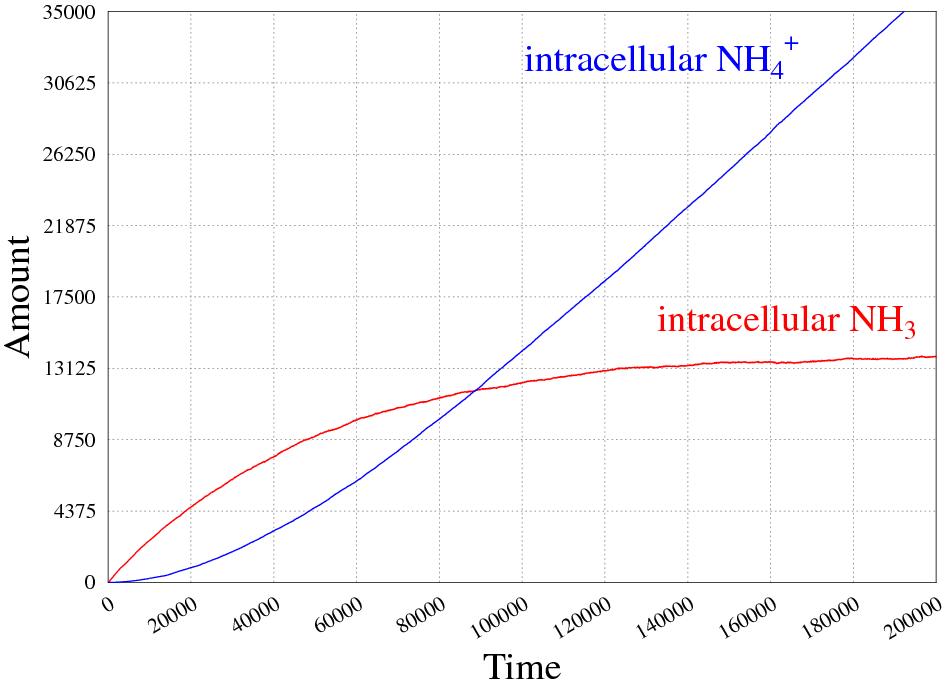}
\end{center}
\caption{At high $NH_4^+$ concentration.}
\label{fig:nh4_ljamt2_noNH3}
\end{minipage}
\end{center}
\end{figure}
\begin{figure}[t]
\begin{center}
\begin{minipage}{0.98\textwidth}
\begin{center}
\includegraphics[width=77mm]{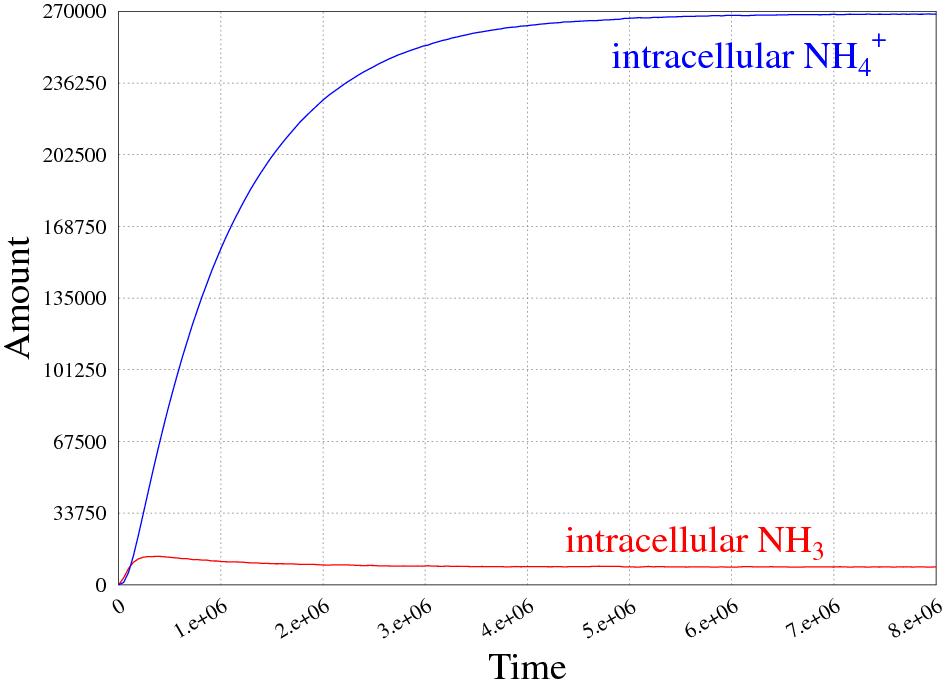}
\includegraphics[width=77mm]{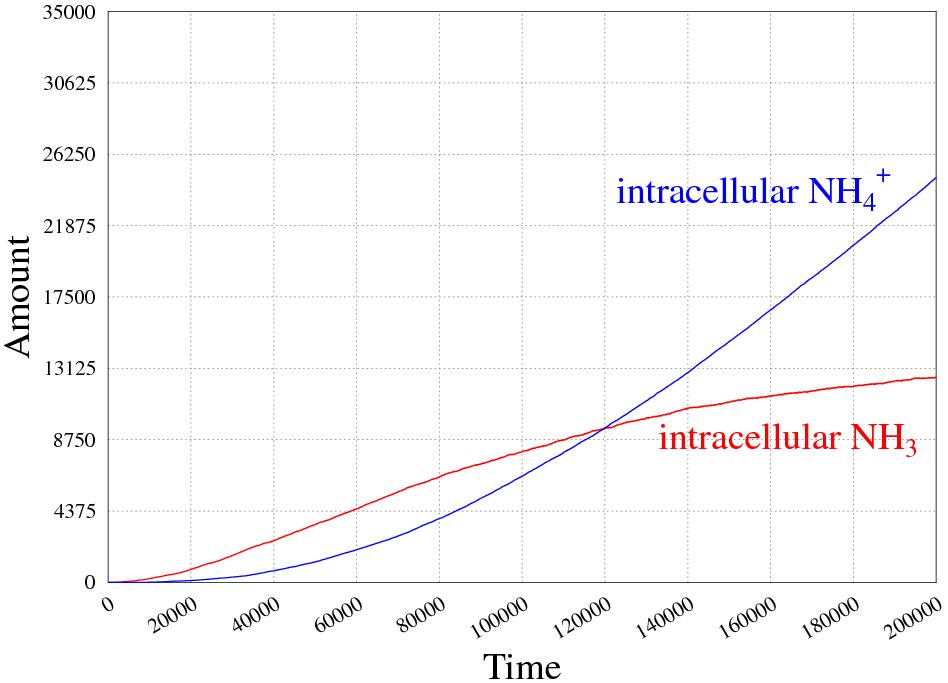}
\end{center}
\vspace{-0.5cm}
\caption{At low $NH_4^+$ concentration.}
\label{fig:nh4_ljamt2_noNH4}
\end{minipage}
\end{center}
\end{figure}

We can also investigate the uptake rate of the transporter at
different extracellular pH. Namely, we consider an extracellular
pH equal to the intracellular one (pH around 7 and 8), obtained by
imposing $R1$ and $R2$ equal to $R3$ and $R4$, respectively, i.e.
$k_2 = k_2'$. Figure~\ref{fig:nh4_ljamt2_pH7} shows the results
for the initial term $\ov{v}=NH_3\times 138238 \conc
NH_4^+\times 138238 \conc (LjAMT2 \into \emptyseq)^\ell$.

\begin{figure}[t]
\begin{center}
\begin{minipage}{0.98\textwidth}
\begin{center}
\includegraphics[width=77mm]{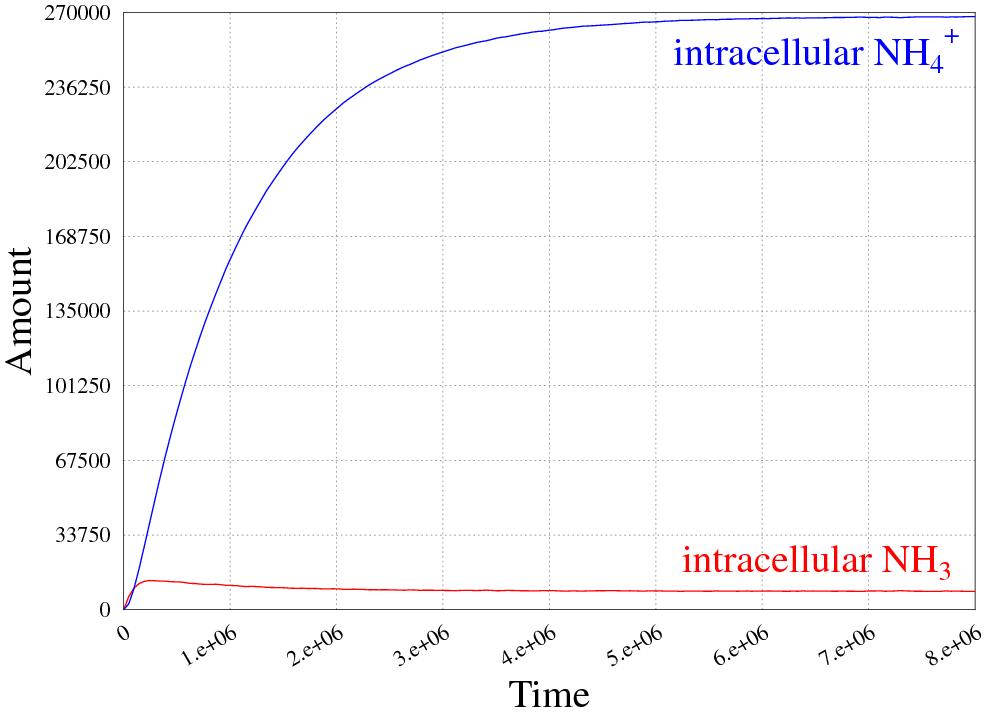}
\includegraphics[width=77mm]{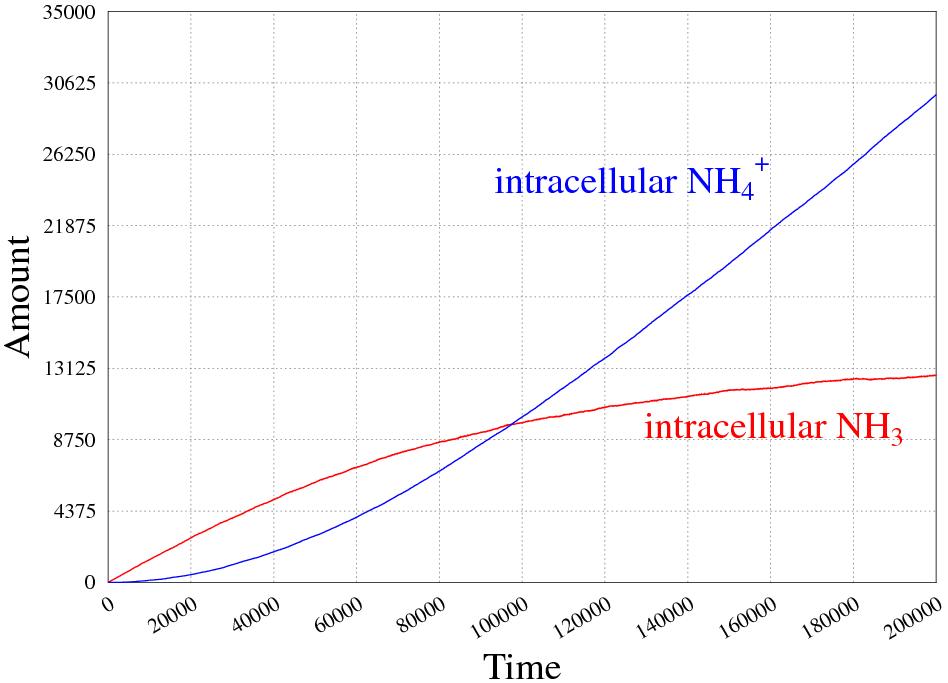}
\end{center}
\vspace{-0.5cm}
\caption{At extracellular pH=7.}
\label{fig:nh4_ljamt2_pH7}
\end{minipage}
\end{center}
\end{figure}

Since now we modeled the transporter supposing that no active form
of energy is required to do the actual work - which means that the
$NH_4^+$ gradient between the cell and the extracellular ambient
is sufficient to determine a net uptake. The predicted
tridimensional structure of LjAMT2;2 suggests that it does not use
ATP~\footnote{ATP is the ``molecular unit of currency'' of
intracellular energy transfer~\cite{KNO80} and is used by many
transporters that work against chemical gradients.} as an energy
source~\cite{GUE09}, nevertheless trying to model an ``energy
consumption'' scenario is interesting to make some comparisons.
Since this is only a proof of concept there is no need to specify
here in which form this energy is going to be provided.
Furthermore, as long as we are only interested in comparing the
initial rates of uptake, we can avoid defining rules that
regenerate energy in the cell. Therefore, rule R5 modelling the
transporter role can be modified as follows:
\begin{align*}\tag*{(R5')}
\top: NH_4^+ \conc (LjAMT2 \conc x \into & \mathit{ENERGY} \conc X)^\ell \srewrites{k_t'} \\
& H^+ \conc (LjAMT2 \conc x \into NH_3 \conc X)^\ell
\end{align*}
which consumes an element of energy within the cell. We also make
this reaction slower, since it is now catalysed by the
concentration of the $\mathit{ENERGY}$ element, actually, we set
$k'_{t}=0.1 \times 10^{-10}$. Given the initial term
$\ov{w}=NH_3\times 138238 \conc NH_4^+ \times 138238
\conc (LjAMT2 \into \mathit{ENERGY}\times 100000)^\ell$ we obtain the
simulation result in Figure~\ref{fig:nh4_ljamt2_energy}. Note that
the uptake work of the transporter terminates when the $\mathit{ENERGY}$
inside the cell is completely exhausted.

\begin{figure}[t]
\begin{center}
\begin{minipage}{0.98\textwidth}
\begin{center}
\includegraphics[width=77mm]{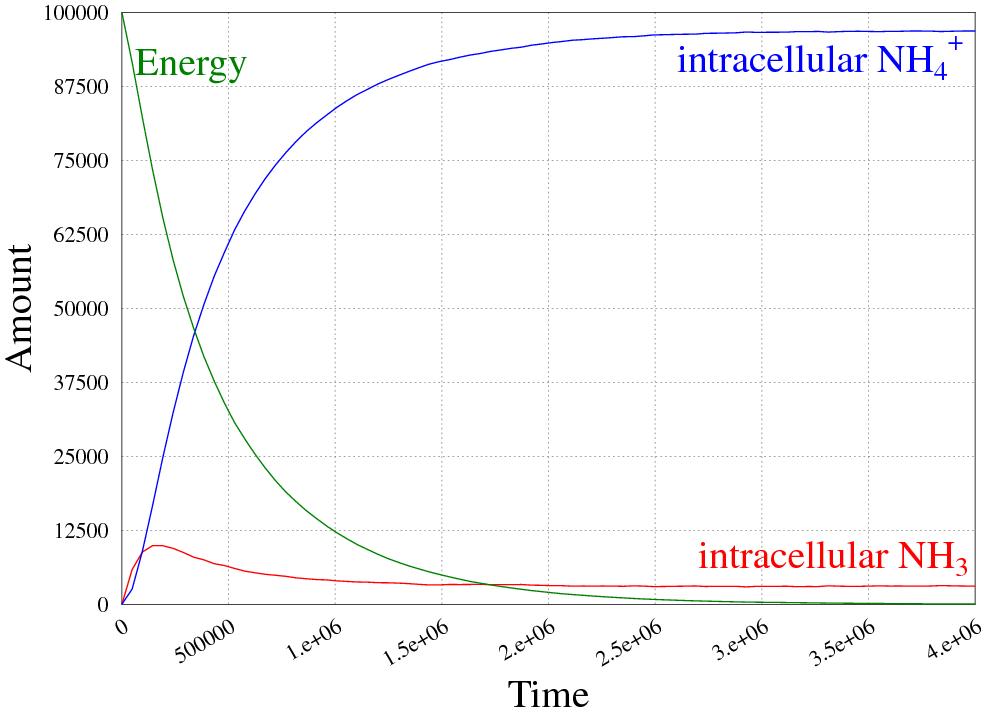}
\includegraphics[width=77mm]{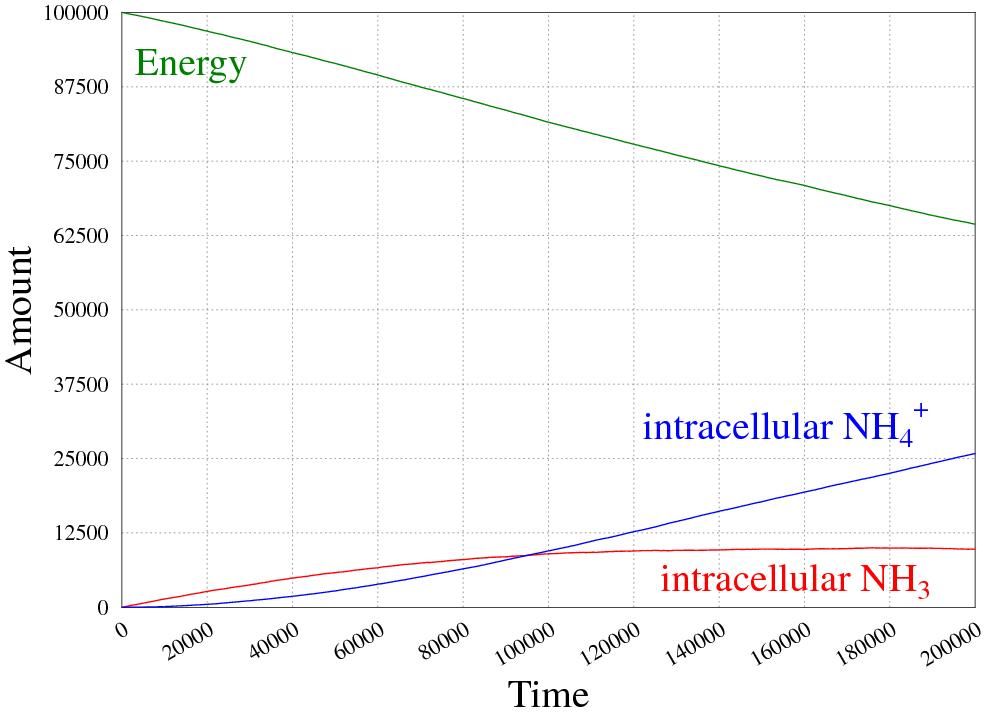}
\end{center}
\vspace{-0.5cm}
\caption{LjAMT2;2 with active energy.}
\label{fig:nh4_ljamt2_energy}
\end{minipage}
\end{center}
\end{figure}

\subsection{$NH_4^+$ Diffusing from the Fungus}

We now model the diffusion of $NH_4^+$ from the fungus to the
extracellular level (sections \textbf{(a1)}, and \textbf{(a2)} of
Figure~\ref{fig:symbschem}). In section \textbf{(a1)} of the
figure, the passage of $NH_4^+$ to the interfacial periarbuscular
space happens by diffusion. We can model this phenomenon by adding
a new compartment (labelled with ${f}$), representing the fungus, from which $NH_4^+$
flows towards the fungus-plant interface. This could be modelled
through the rule: \begin{equation*}\tag*{(R6)}\top: 
(FungMembr \conc x \into NH_4^+ \conc X)^{f}\srewrites{k_f} NH_4^+ \conc (FungMembr \conc x\into X)^{f}
\end{equation*}

By varying the value of the rate $k_f$ one might model different
externalization speeds and thus test different hypotheses about
the underlying mechanism. In Figure~\ref{fig:nh4_ljamt2_fungus} we
give the simulation result, with three different magnification
levels, going from the whole simulation to the initial parts,
obtained from the initial term
$\ov{t}_f = (FungMembr \;\into\; NH_4^+\times 2764677)^{f}
\conc (LjAMT2 \into \emptyseq)^\ell$ with $k_f=1$. In the initial
part, one can see how fast, in this case, $NH_4^+$ diffuses into
the extracellular space. In
Figure~\ref{fig:nh4_ljamt2_fungus_slow} we give the simulation
result obtained from the same initial term $\ov{t}_f$ with a
slower diffusion rate, namely $k_f=0.01 \times 10^{-3}$.

\begin{figure}[t]
\begin{center}
\begin{minipage}{0.98\textwidth}
\begin{center}
\includegraphics[width=51mm]{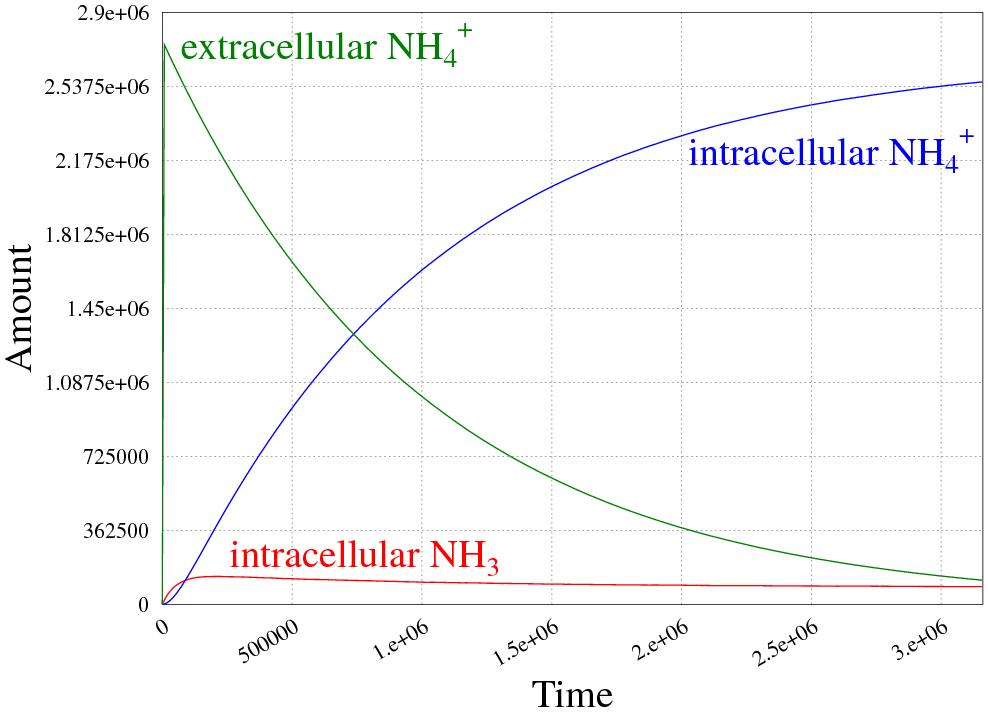}
\includegraphics[width=51mm]{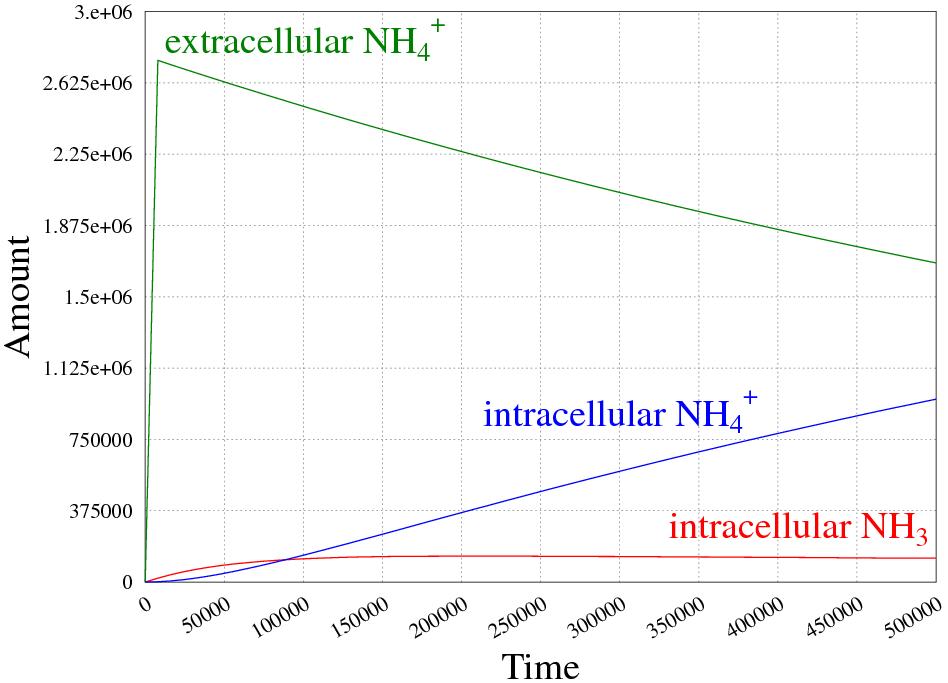}
\includegraphics[width=51mm]{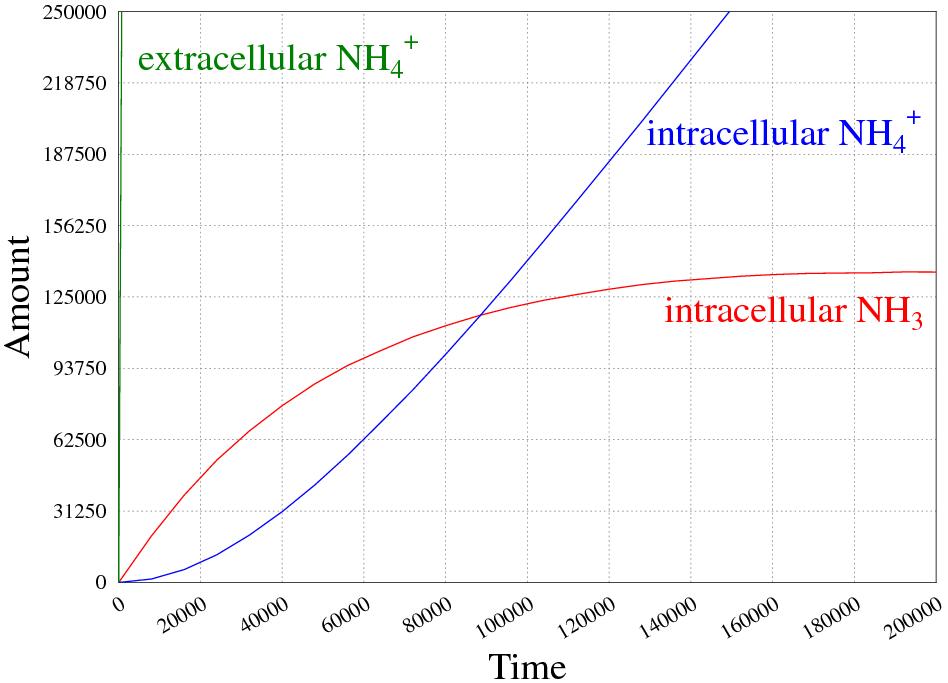}
\end{center}
\vspace{-0.5cm}
\caption{Diffusing $NH_4^+$ from the fungus, $k_f=1$.}
\label{fig:nh4_ljamt2_fungus}
\end{minipage}
\end{center}
\end{figure}

\begin{figure}[t]
\begin{center}
\begin{minipage}{0.98\textwidth}
\begin{center}
\includegraphics[width=51mm]{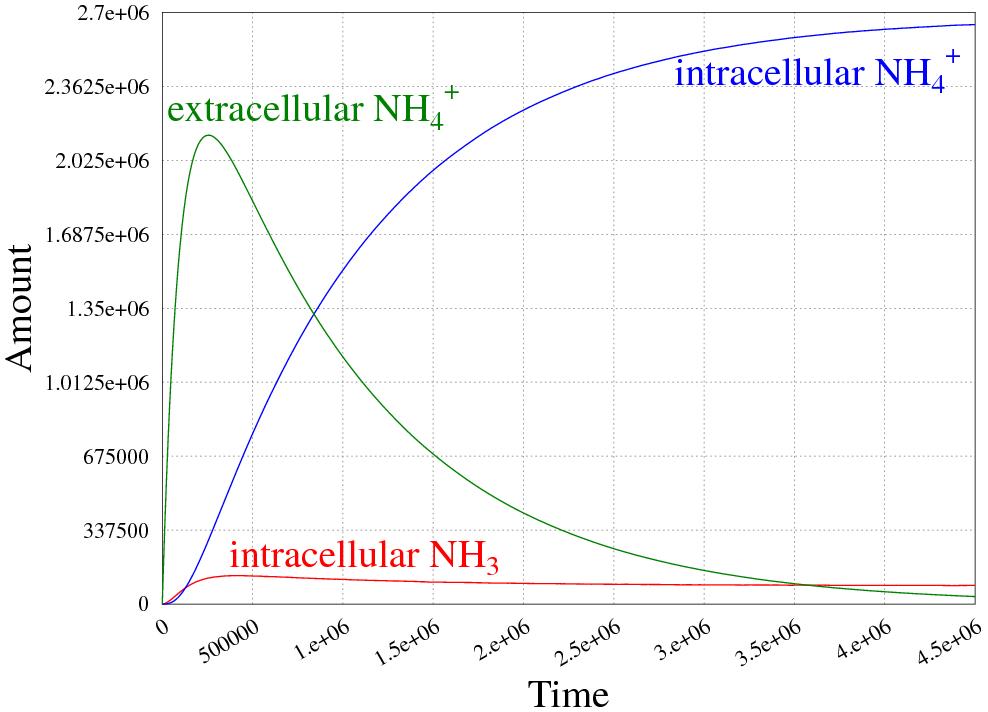}
\includegraphics[width=51mm]{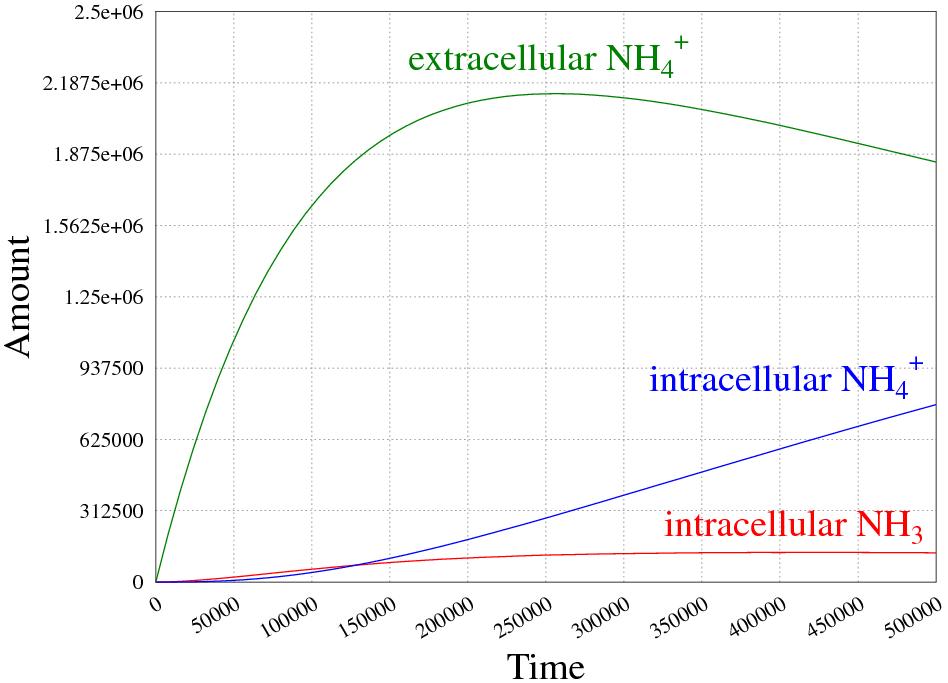}
\includegraphics[width=51mm]{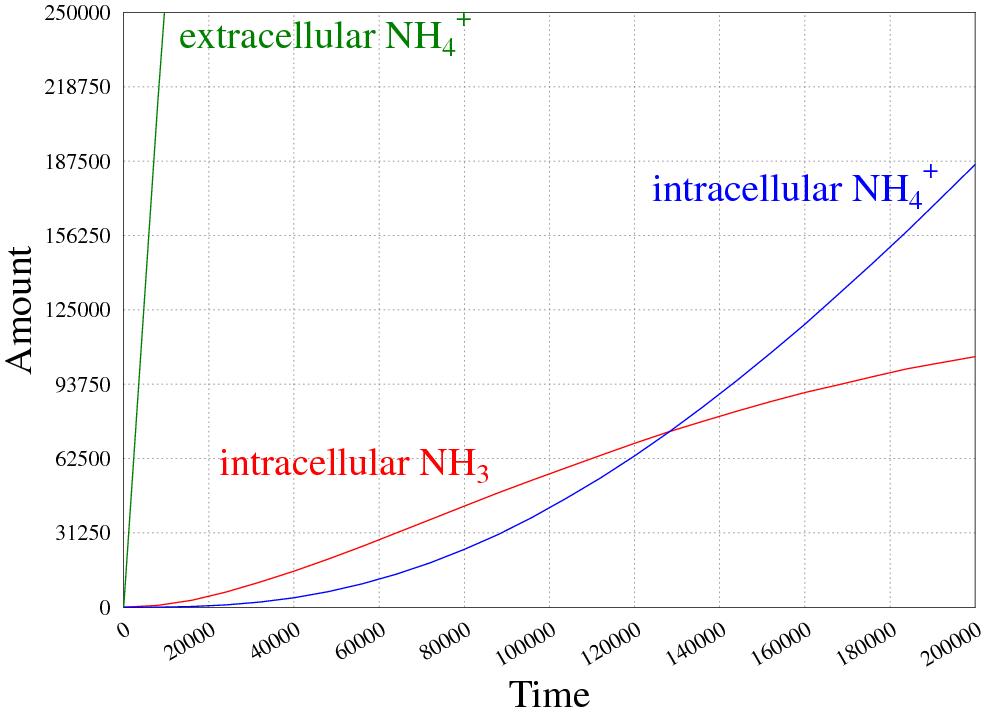}
\end{center}
\vspace{-0.5cm}
\caption{Diffusing $NH_4^+$ from the fungus, $k_f=0.01*10^{-3}$.}
\label{fig:nh4_ljamt2_fungus_slow}
\end{minipage}
\end{center}
\end{figure}

Additionally, we would like to remark, without going into the
simulation details, how we can model in a rather natural way the
portion \textbf{(a2)} of Figure~\ref{fig:symbschem} in CWC.
Namely, we need some rules to produce vesicles (labelled with ${v}$) containing $NH_4^+$
molecules within the fungal cell. Once the vesicle is formed,
another rule drives its exocytosis towards the interfacial space,
and thus the diffusion of the previously encapsulated $NH_4^+$
molecules. The necessary rules are given in the following:
\begin{equation*}\tag*{(R7)} \top: (FungMembr \conc x \;\into\; X)^{f} \;\srewrites{k_c}\; (FungMembr \conc x \;\into\; X \conc
(Vesicle \into \emptyseq)^{v})^{f}
\end{equation*}
\begin{align*}\tag*{(R8)} \top: 
(FungMembr \conc x \;\into\; & NH_4^+ \conc Y \conc (Vesicle \conc y \into X)^{v})^{f}
 \;\srewrites{k_a}\; \\ & ( FungMembr \conc x \;\into\; Y \conc (Vesicle \conc y \into
NH_4^+\conc X)^{v} )^{f}
\end{align*}
\begin{equation*}\tag*{(R9)} \top: 
( FungMembr \conc x \,\into\, Y \conc (Vesicle \conc y \into X)^{v} )^{f}
\;\srewrites{k_e}\; X \conc (FungMembr \conc x \conc y \into Y)^{f}
\end{equation*}
Where rule R7 models the creation of a vesicle, rule R8 models the
encapsulation of an $NH_4^+$ molecule within the vesicle and rule
R9 models the exocytosis of the vesicle content.

\subsection{Emplicit pH representation}

To be able to further investigate the delicate equilibria that are established in
the AM symbiosis we needed a model that deals with $H^+$ and $OH^-$
and therefore could offer more precise analysis opportunities on the system -
to correctly model chemical reactions
at different pHs it is important to find a set of rules that is capable
of reaching and keeping the right ratio between $H^+$ and $OH^-$, even
when used with other rules that comprises these ions, adding or removing them.

To avoid the need of huge quantities of water molecules only hydrogen ions and hydroxide are considered and not
the process of water dissociation (as long as water is 55.5 M while, for example, at pH 4.5 hydrogen is $3.2 \times 10^{-5}$ M and hydroxyl is $3.2 \times 10^{-10}$ M and we have to use natural numbers to represent the quantities of molecules in the simulations it would be cumbersome to consider water).

Thus the rules simply have to ``create'' and ``destroy'' the ions: $H^+$ has to be destroyed considering its quantity and
generated considering $OH^-$ quantity and the same should be done for $OH^-$;\footnote{In such a way we abstract the reaction $H_2O \rightleftharpoons OH^- + H^+$ without considering explicitly the amount of water molecules.} different pH will be obtained changing
the rates of these rules.\\
The rules are easy to define: $ H^+ \srewrites{k1} \emptyseq $ and $ OH^- \srewrites{k2} \emptyseq$ for the rules
that destroy ions and $ H^+ \srewrites{k3} H^+ \conc OH^- $ and $ OH^- \srewrites{k4} OH^- \conc H^+$ to create them,
the stochastic simulation correctly applies them with an application rate that depends on the defined $k$s
and the given ion numbers in the term to which they are applied.
In this way two couples of rules are capable of maintaining
a proper ratio between $H^+$ and $OH^-$: to obtain different pH it is enough to tune their rates
in order to reflect the desiderate ratio.

\begin{figure}[t]
\begin{center}
\begin{minipage}{0.98\textwidth}
\begin{center}
\includegraphics[height=45mm]{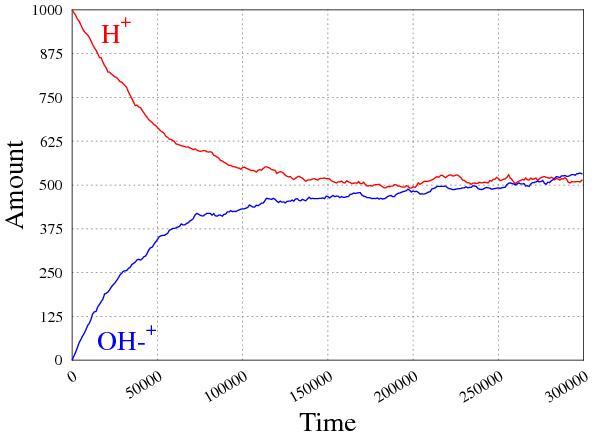}
\includegraphics[height=45mm]{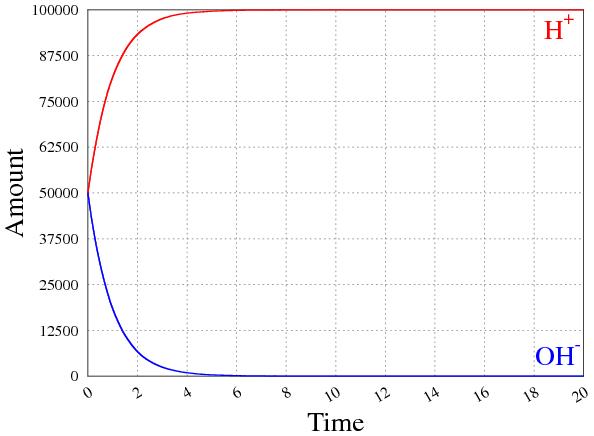}
\end{center}
\vspace{-0.5cm}
\caption{pH 7 and 4.5 with correct ratios reached}
\label{fig:ph_eq}
\end{minipage}
\end{center}
\end{figure}

We defined the correct rates for different pHs,
for example pH 4.5, which is necessary to model the periarbuscular space
and is characterized by a ratio between $H^+$ and $OH^-$ of $10^5$.

\subsection{$NH_3$/$NH_4^+$ equilibrium and LjAMT2;2}
After the definition of proper rules for pH, we had to drive the exchange
between ammonium cations and ammonia, with rules that should be capable of reaching and
keeping the right ratio for these molecules at a given pH.
Due to the explicit pH model the rules were changed with respect to the previous ones:
\begin{equation}
\label{eq1}\tag*{(R1')}
\top: NH_3 \conc H^+ \srewrites{k1} NH_4^+
\end{equation}
\begin{equation}
\label{eq2}\tag*{(R2')}
\top: NH_4^+ \srewrites{k2} NH_3 \conc H^+
\end{equation}

Rule~\ref{eq1} represents ammonia which becomes protonated binding a free
hydrogen ion, while rule~\ref{eq2} is the converse reaction. We did not consider other reactions
involving water, such as $NH_3 \conc H_2O \srewrites{k3} NH_4^+ \conc OH^-$,
for the previously explained reasons.

Several simulation were needed to tune the right rates, which were
defined for pH 7, pH 4.5 and pH 8.6 - more
extreme pHs are unlikely in our biological domain and they are difficult to model due to the ratios
that have to be reached ($1.7 \times 10^9$ between $H^+$ and $OH^-$ at pH 2.6),
which force to run simulations with huge molecule numbers.

\begin{figure}[t]
\begin{center}
\includegraphics[width=77mm]{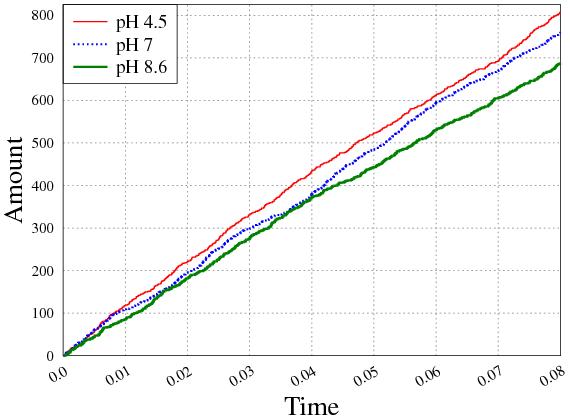}
\caption{\footnotesize{LjAMT2;2 uptake (internalized NH3), comparison between different extracellular pHs.}}
\label{fig:pHdep}
\end{center}
\end{figure}

Having defined all these rules the next step was to add LjAMT2;2
and try to confirm some of the results obtained with the first model,
for example the pH dependent uptake rate.
Figure~\ref{fig:pHdep} shows a plot that represents the internalized
$NH_3$ versus time with different periarbuscular pHs.
The three simulations, started with ``steady state'' quantities of $H^+$ and
$OH^-$ and of $NH_3$ and $NH_4^+$ outside the cell (according to the chosen pH),
while the cell started with no ammonia; they still have only a cell
with a single transporter on the membrane - to be able to compare the internalization
rate with sufficient numbers we changed the rate for the transport rule, namely $k_t=0.1 \times 10^{-2}$.

\subsection{Comparison with LjAMT1;1}

To further investigate the role of LjAMT2;2 in the context of AM and its peculiar
mechanism of transport, which does not depend on the $H^+$ gradient and seems otherwise
to have a role in its maintenance (by expelling the $H^+$ gotten from the $NH_4^+$
molecule), we compared it with another ammonium transporter which exists in plants
but is not so selectively expressed in arbusculated cells: LjAMT1;1~\cite{SMRPC01};
this is a transporter for $NH_4^+$ that does not expel $H^+$ in the periarbuscular space, therefore
it internalizes directly an $NH_4^+$ molecule.

It is interesting to try to understand if LjAMT2;2 has a role which is synergic with other transporters
in the AM symbiosis which relies on the $H^+$ gradient, such as those for the phosphates on the plant side
or for the carbohydrates in the fungi and if other ammonia transporters, like LjAMT1;1, would instead
``compete'' with other transporters consuming hydrogen.


Modelling what would happen if in the arbusculated cell LjAMT1;1 will be the principal ammonia
transporter, instead of LjAMT2;2, requires only to change the transporter rule with respect to the simulation
with the periarbuscular pH at 4.5 (which is the normal one) shown in the previous section:
\begin{equation*}\tag*{(R7')} \top: 
NH_4^+ \conc (LjAMT1 \conc x \into X)^\ell \srewrites{k_t} (LjAMT1 \conc x \into X \conc NH_4+)^\ell
\end{equation*}
Note that the rate is the same for both rules.

\begin{figure}[t]
\begin{center}
\begin{minipage}{0.98\textwidth}
\begin{center}
\subfigure[Periarbuscular $H^+$]{
\includegraphics[height=45mm]{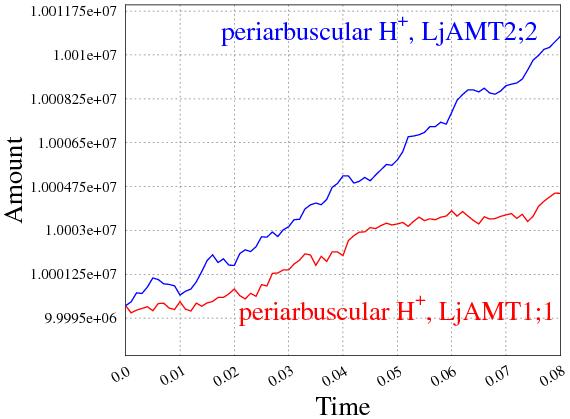}
\label{fig:LjAMT_comp1}
}
\subfigure[Periarbuscular $OH^-$]{
\includegraphics[height=45mm]{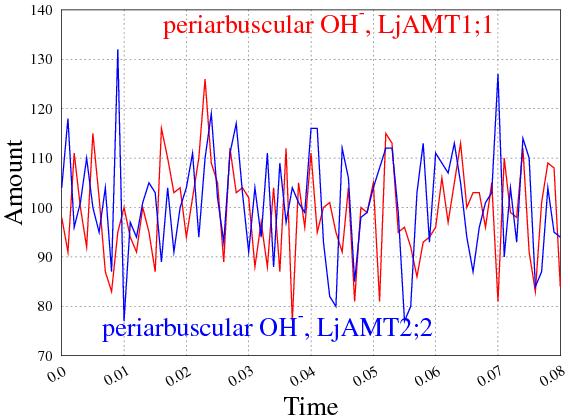}
\label{fig:LjAMT_comp2}
}
\caption{\footnotesize{Comparison between LjAMT1;1 and LjAMT2;2.}}
\end{center}
\label{fig:LjAMT_comp}
\end{minipage}
\end{center}
\end{figure}

Figure~\ref{fig:LjAMT_comp} represents the periarbuscular quantities of $H^+$ and $OH^-$ for two simulations
which had all the same rules except for the transporter one and started from the terms:
\begin{equation*}
NH_4^+ \times 10^7 \conc NH_3 \times 200 \conc OH^- \times 100 \conc H^+ \times 10^7 \conc
(LjAMT2 \into \emptyseq)^\ell
\end{equation*}
\begin{equation*}
NH_4^+ \times 10^7 \conc NH_3 \times 200 \conc OH^- \times 100 \conc H^+ \times 10^7 \conc
(LjAMT1 \into \emptyseq)^\ell
\end{equation*}

\section{Discussion}\label{SecDiscussion}
We dissected the route for the passage of $NH_3$ / $NH_4^+$ from the fungus to the plant in known and hypothetical mechanisms which were transformed in rules. Further, also the properties of the different compartments and their influence on the transported molecules were included, thus giving a first model for the simulation of the nutrients transfer.
With the model so far we can simulate the behaviour of the system when varying parameters as the different compartments pH, the initial substrate concentrations, the transport/diffusion speeds and the energy supply.

We can start comparing the two simulations with the plant cell
with the LjAMT2;2 transporter placed in different extracellular
situations: low $NH_4^+$ concentration
(Figure~\ref{fig:nh4_ljamt2_noNH4}) or high $NH_4^+$ concentration
(Figure~\ref{fig:nh4_ljamt2_noNH3}). As a natural
consequence of the greater concentration, the ammonium
uptake is faster when the simulation starts with more $NH_4^+$, as
long as the LjAMT2;2 can readily import it. The real situation
should be similar to this simulation, assuming that the level of
extracellular $NH_4^+$ / $NH_3$ is stable, meaning an active
symbiosis.

The simulation which represents an extracellular pH around 7 (Figure~\ref{fig:nh4_ljamt2_pH7}) shows a decreased internalisation
speed with respect to the simulation in Figure~\ref{fig:nh4_ljamt2_noNH3}, as could be inferred from the
concentrations of NH4\_inside and NH3\_inside in the plots on the right (focusing on the initial activity): this supports experimental data about
the pH-dependent activity of the transporter and suggests that the extracellular pH is fundamental
to achieve a sufficient ammonium uptake for the plants.
It is noticeable how the initial uptake rate in this case
is higher, despite the neutral pH, than the rate obtained considering an ``energy quantum''
used by the transporter (which has the same starting term),
as could be seen in the right panels of Figure~\ref{fig:nh4_ljamt2_energy}
and Figure~\ref{fig:nh4_ljamt2_pH7}.
These results could enforce the biological hypothesis
that, instead of ATP, a $NH_4^+$
concentration gradient (possibly created by the fungus) is used as
energy source by the LjAMT2;2 protein.

The simulations which also consider the fungal counterpart are
interesting because they provide an initial investigation of this
rather
 poorly characterized side of the symbiosis and confirm that plants can efficiently gain
ammonium if $NH_4^+$ is released from the fungi. This evidence
supports the latest biological hypothesis about how fungi supply
nitrogen to plants~\cite{Gov05,Cha06}, and could lead to further
models which could suggest which is the needed rate for $NH_4^+$
transport from fungi to the interfacial apoplast; thus driving
biologists toward one (or some) of the nowadays considered
hypotheses (active transport of $NH_4^+$, vesicle formation,
etc.).

The explicit representation of pH, while still yielding results comparable with the first
simulations, like the pH-dependency of the uptake rate which is clearly represented
in Figure~\ref{fig:pHdep}, offers a scenario where it is possible to analyze
more deeply LjAMT2;2 characteristics and the delicate interactions between
different transporters.

In the LjAMT2;2 and 1;1 comparison simulations the periarbuscolar pH is mantained at around 4.5
by the previously discussed rules, but by examining $H^+$ and $OH^-$ one could note
that LjAMT2;2 determines an increment of hydrogen which is higher than the one determined by LjAMT1;1,
while hydroxide shows oscillations around a mean value for both simulations.
This result could suggest that LjAMT2;2 indeed has a role in maintaining the $H^+$ gradient which is pivotal for other
nutrient exchanges that take place in the symbiosis and that its overexpression in the arbusculated cells has a functional
meaning.

\chapter{Modelling Ecological Systems}
\label{ecolog}
Computational Ecology is a field devoted to the quantitative description and analysis of ecological systems using empirical data, mathematical models (including statistical models), and computational technology. While the different components of this interdisciplinary field of research are not new, there is a new emphasis on the integrated treatment of the area. This emphasis is amplified by the expansion of our local, national, and international computational infrastructure, coupled with the heightened social awareness of ecological and environmental issues and its effects on research funding.

We advocate a convergence between computer and life sciences. This emerging paradigm moves to a system level understanding of life, where unpredictable,
complex behaviour show up. We claim that computer science will greatly contribute to a
better understanding of the behaviour of ecological systems. We plan to develop models, languages and tools for describing, analysing and
implementing \emph{in silico} ecological systems, as an additional contribution of Information Technology to
those typical research areas in current Computational Ecology, such as (i) storing, organising and
retrieving large amounts of ecological data or (ii) visual modelling techniques for scientific visualisation of multi--dimensional, computer--generated scenes that can be used to express empirical data.

More in detail, we use our formal framework for modelling and studying the
behaviour of living systems. Our starting point is that ecological systems are conveniently described as
entities that change their state because of the occurrence of biotic and abiotic interactions, giving
rise to some observable behaviour. We thus adhere to the view of living
systems as biological computing units.

In Section~\ref{sec:pop} we present some of the characteristic features leading the evolution of ecological systems, and we show how to encode them within CWC.

In section~\ref{Sec:Croton}, we take into consideration the species \emph{Croton wagneri}, a shrub in the dry ecosystem of southern Ecuador, and investigate how it could adapt to global climate change.

\section{Population Dynamics}\label{sec:pop}

Models of population dynamics describe the changes in the size and composition of populations.

A \emph{metapopulation}\footnote{The term metapopulation was coined by Richard Levins in 1970. In Levins' own words, it consists of ``a population of populations''~\cite{Lev69}.} is a group of populations of the same species distributed in different patches\footnote{A patch is a relatively homogeneous area differing from its surroundings.} and interacting at some level.
Thus, a metapopulation consists of several distinct populations and areas of suitable habitat.

Individual populations may tend to reach extinction as a consequence of demographic stochasticity (fluctuations in population size due to random demographic events); the smaller the population, the more prone it is to extinction. A metapopulation, as a whole, is often more stable: immigrants from one population (experiencing, e.g., a population boom) are likely to re-colonize the patches left open by the extinction of other populations. Also, by the \emph{rescue effect}, individuals of more dense populations may emigrate towards small populations, rescuing them from extinction.

\subsection{Exponential Growth Model}
The \emph{exponential growth model} is a common mathematical model for population dynamics, where, using $r$ to represent the pro-capita growth rate of a population of size $N$, the change of the population is proportional to the size of the already existing population:
$$
\frac{dN}{dt} = r \cdot N
$$

We can encode within CWC the exponential growth model with rate $r$ using a stochastic rewrite rule describing a reproduction event for a single individual at the given rate. Namely, given a population of species $a$ living in an environment modelled by a compartment with label $\ell$, the following CWC rule encodes the exponential growth model:
$$
\ell: a  \srewrites{r} a \conc a
$$
Counting the number of possible reactants, the growth rate of the overall population is automatically obtained by the stochastic semantics underlying CWC.

\subsection{BIDE model}

Populations are affected by births and deaths, by immigrations and emigrations (BIDE model~\cite{Cas01}). The number of individuals at time $t+1$ is given by:
$$N_{t+1}=N_{t}+B+I-D-E$$
where $N_t$ is the number of individuals at time $t$ and, between time $t$ and $t+1$, $B$ is the number of births, $I$ is the number of immigrations, $D$ is the number of deaths and $E$ is the number of emigrations.
Conditions triggering migration could be: climate, food availability or mating~\cite{DD07}.

We can encode within CWC the BIDE model for a compartment of type $\ell$ using stochastic rewrite rules describing the given events with their respective rates $r$, $i$, $d$, $e$:
$$
\begin{array}{lr}
\ell: a  \srewrites{r} a \conc a & \quad \text{(birth)}\\
\top: a \conc (x \into X)^\ell \srewrites{i} (x \into a \conc X)^\ell  & \quad \text{(immigration)}\\
\ell: a \srewrites{d} \emptyseq & \quad \text{(death)}\\
\top: (x \into a \conc X)^\ell  \srewrites{e} a \conc (x \into X)^\ell & \quad \text{(emigration)}
\end{array}
$$
Starting from a population of $N_t$ individuals at time $t$, the number $N_{t+1}$ of individuals at time  $t+1$ is computed by successive simulation steps of the stochastic algorithm. The race conditions computed according to the propensities of the given rules assure that all of the BIDE events are correctly taken into account.

\begin{example}\label{immigration} Immigration and extinction are key components of island biogeography. We model a metapopulation of species $a$ in a context of 5 different patches: 4 of which are relatively close, e.g. different ecological regions within a small continent, the last one is far away and difficult to reach, e.g. an island. The continental patches are modelled as CWC compartments of type $\ell_c$, the island is modelled as a compartment of type $\ell_i$. Births, deaths and migrations in the continental patches are modelled by the following CWC rules:
$$
\begin{array}{c}
\ell_c: a  \srewrites{0.005} a \conc a \quad \quad
\ell_c: a \srewrites{0.005} \emptyseq\\
\top: (x \into a \conc X)^{\ell_c}  \srewrites{0.01} a \conc (x \into X)^{\ell_c} \quad \quad
\top: a \conc (x \into X)^{\ell_c} \srewrites{0.5} (x \into a \conc X)^{\ell_c}
\end{array}
$$
These rates are drawn considering days as time unites and an average of life expectancy and reproduction time for the individuals of the species $a$ of 200 days ($\frac{1}{0.005}$). For the modelling of real case studies, these rates could be estimated from data collected \emph{in situ} by tagging individuals.\footnote{In the remaining examples we will omit a detailed time description.} In this model, when an individual emigrates from its previous patch it moves to the top-level compartment from where it may reach one of the close continental patches (might also be the old one) or start a journey through the sea (modelled as a rewrite rule putting the individual on the wrapping of the island compartment):
$$
\top: a \conc (x \into X)^{\ell_i} \srewrites{0.2} (x\conc a \into  X)^{\ell_i}\\
$$
Crossing the ocean is a long and difficult task and individuals trying it will probably die during the cruise; the luckiest ones, however, might actually reach the island, where they could eventually benefit of a better life expectancy for them and their descendants:
$$
\begin{array}{c}
\top: \conc (x\conc a \into X)^{\ell_i} \srewrites{0.333} (x \into  X)^{\ell_i} \quad \quad
\top: \conc (x\conc a \into X)^{\ell_i} \srewrites{0.0005} (x \into a\conc X)^{\ell_i}\\
\ell_i: a  \srewrites{0.007} a \conc a \quad \quad
\ell_i: a \srewrites{0.003} \emptyseq
\end{array}
$$
Considering the initial system modelled by the CWC term:
$$
\ov{t}= (\emptyseq \into 30*a)^{\ell_c} \conc (\emptyseq \into 30*a)^{\ell_c} \conc (\emptyseq \into 30*a)^{\ell_c} \conc (\emptyseq \into 30*a)^{\ell_c} \conc (\emptyseq \into \emptyseq)^{\ell_i}
$$
we can simulate the possible evolutions of the overall diffusion of individuals of species $a$ in the different patches. Notice that, on average, one over $\frac{0.333}{0.0005}$ individuals that try the ocean journey, actually reach the island.
In Figure~\ref{FigMP} we show the result of a simulation plotting the number of individuals in the different patches in a time range of approximatively 10 years. Note how, in the final part of the simulation, empty patches get recolonised. In this particular simulation, also, an exponential growth begins after the colonisation of the island.
The full CWC model describing this example can be found at: \url{http://www.di.unito.it/~troina/cmc13/metapopulation.cwc}.
\begin{figure}
\centering
\includegraphics[height=60mm]{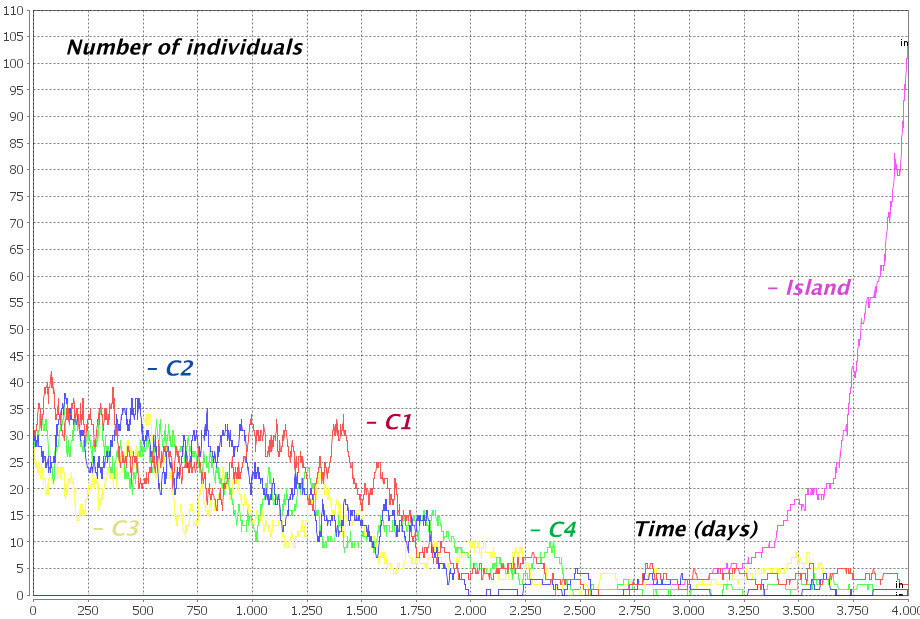}
\caption{\label{FigMP} Metapopulation dynamics.}
\end{figure}

\end{example}

\subsection{Logistic Model}

In ecology, using $r$ to represent the pro-capita growth rate of a population and $K$ the \emph{carrying capacity} of the hosting environment,\footnote{I.e., the population size at equilibrium.} $r/K$ selection theory~\cite{Pia70} describes a selective pressure driving populations evolution through the \emph{logistic model}~\cite{Verhulst}:
$$
    \frac{dN}{dt} = r\cdot N \cdot \left(1 - \frac{N}{K}\right)
$$
where $N$ represents the number of individuals in the population.

The logistic model with growth rate $r$ and carrying capacity $K$, for an environment modelled by a compartment with label $\ell$, can be encoded within CWC using two stochastic rewrite rules describing (i) a reproduction event for a single individual at the given rate and (ii) a death event modelled by a fight between two individuals at a rate that is inversely proportional to the carrying capacity:
$$
\begin{array}{l}
\ell: a  \srewrites{r} a \conc a\\
\ell: a \conc a \srewrites{\frac{2\cdot r}{K-1}} a
\end{array}
$$
If $N$ is the number of individuals of species $a$, the number of possible reactants for the first rule is $N$ and the number of possible reactants for the second rule is, in the exact stochastic model, ${N \choose 2}=\frac{N\cdot(N-1)}{2}$, i.e. the number of distinct pairs of individuals of species $a$. Multiplying this values by the respective rates we get the propensities of the two rules and can compute the value of $N$ when the equilibrium is reached (i.e., when the propensities of the two rules are equal): $r\cdot N = \frac{2 \cdot r}{K-1}\cdot \frac{N\cdot (N-1)}{2}$, that is when $N=0$ or $N=K$.

For a given species, this model allows to describe different growth rates and carrying capacities in different ecological regions. Identifying a CWC compartment type (through its label) with an ecological region, we can define rules describing the growth rate and carrying capacity for each region of interest.

Species showing a high growth rate are selected by the $r$ factor, they usually exploit low-crowded environments and produce many offspring, each of which has a relatively low probability of surviving to adulthood. By contrast, $K$-selected species adapt to densities close to the carrying capacity, tend to strongly compete in high-crowded environments and produce fewer offspring, each of which has a relatively high probability of surviving to adulthood.

\begin{example}
There is little, or no advantage at all, in evolving traits that permit successful competition with other organisms in an environment that is very likely to change rapidly, often in disruptive ways. Unstable environments thus favour species that reproduce quickly ($r$-selected species).
Characteristic traits of $r$-selected species include: high fecundity, small body, early reproduction and short generation time.
Stable environments, by contrast, favour the ability to compete successfully for limited resources ($K$-selected species). Characteristic traits of $K$-selected species include: large body size, long life expectancy, production of fewer offspring (usually requiring extensive parental care until maturity).
We consider individuals of two species, $a$ and $b$.
Individuals of species $a$ are modelled with an higher growth rate with respect to individuals of species $b$ ($r_a>r_b$). Carrying capacity for species $a$ is, instead, lower than the carrying capacity for species $b$ ($K_a<K_b$).
The following CWC rules describe the $r/K$ selection model for $r_a=5$, $r_b=0.00125$, $K_a=100$ and $K_b=1000$:
$$
\begin{array}{c}
\ell: a  \srewrites{5} a \conc a \quad \quad
\ell: b  \srewrites{0.00125} b \conc b\\
\ell: a \conc a \srewrites{0.1} a \quad \quad
\ell: b \conc b \srewrites{0.0000025} b
\end{array}
$$
We might consider a disruptive event occurring on average every 4000 years with the rule:
$$
\top: (x \into X)^{\ell} \srewrites{0.00025} (x \into a \conc b)^{\ell}
$$
devastating the whole content of the compartment (modelled with the variable $X$) and just leaving one individual of each species.
In Figure~\ref{FigrK} we show a 10000 years simulation for an initial system containing just one individual for each species. Notice how individuals of species $b$ are disadvantaged with respect to individuals of species $a$ who reach the carrying capacity very soon. A curve showing the growth of individuals of species $b$ in a stable (non disruptive) environment is also shown. The full CWC model describing this example can be found at: \url{http://www.di.unito.it/~troina/cmc13/rK.cwc}.
\begin{figure}
\centering
\includegraphics[height=50mm]{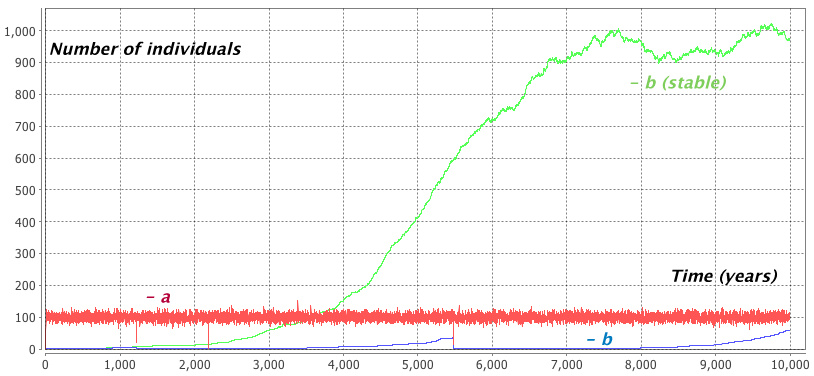}
\caption{\label{FigrK} r/K selection in a disruptive environment.}
\end{figure}
\end{example}

\subsection{Competition and Mutualism}

In ecology, \emph{competition} is a contest for resources between organisms: animals, e.g., compete for water supplies, food, mates, and other biological resources. In the long term period, competition among individuals of the same species (\emph{intraspecific competition}) and among individuals of different species (\emph{interspecific competition}) operates as a driving force of adaptation, and, eventually, by natural selection, of evolution. Competition, reducing the fitness of the individuals involved,\footnote{By fitness it is intended the ability of surviving and reproducing. A reduction in the fitness of an individual implies a reduction in the reproductive output. On the opposite side, a fitness benefit implies an improvement in the reproductive output.} has a great potential in altering the structure of populations, communities and the evolution of interacting species. It results in the ultimate survival, and dominance, of the best suited variants of species: species less suited to compete for resources either adapt or die out. We already depicted a form of competition in the context of the logistic model, where individuals of the same species compete for vital space (limited by the carrying capacity $K$).

Quite an apposite force is \emph{mutualism}, contest in which organisms of different species biologically interact in a relationship where each of the individuals involved obtain a fitness benefit. Similar interactions between individuals of the same species are known as \emph{co-operation}. Mutualism belongs to the category of symbiotic relationships, including also \emph{commensalism} (in which one species benefits and the other is neutral, i.e. has no harm nor benefits) and \emph{parasitism} (in which one species benefits at the expense of the other).

The general model for
competition and mutualism between two species $a$ and $b$ is defined by the following equations~\cite{Tak89}:
$$
\begin{array}{l}
   \frac{dN_a}{dt} = \frac{r_a \cdot N_a}{K_a} \cdot \left( K_a - N_a + \alpha_{ab}\cdot N_b \right) \\
   \frac{dN_b}{dt} = \frac{r_b \cdot N_b}{K_b} \cdot \left( K_b - N_b + \alpha_{ba}\cdot N_a \right)
\end{array}
$$
where the $r$ and $K$ factors model the growth rates and the carrying capacities for the two species, and
the $\alpha$ coefficients describe the nature of the relationship between the two species: if $\alpha_{ij}$ is negative, species $N_j$ has negative effects on species $N_i$ (i.e., by competing or preying it), if $\alpha_{ij}$ is positive, species $N_j$ has positive effects on species $N_i$ (i.e., through some kind of mutualistic interaction).

The logistic model, already discussed, is included in the differential equations above. Here we abstract away from it and just focus on the components which describe the effects of competition and mutualism we are now interested in.

\begin{model}[Competition and Mutualism] For a compartment of type $\ell$, we can encode within CWC the model about competition and mutualism for individuals of two species $a$ and $b$ using the following stochastic rewrite rules:
$$
\ell: a \conc b \srewrites{f_a \cdot |\alpha_{ab}|}
\left\{
\begin{array}{lr}
 a\conc a \conc b & \quad \text{if } \alpha_{ab}>0 \\
 b & \quad \text{if } \alpha_{ab}<0
\end{array}
\right.
\quad\quad
\ell: a \conc b \srewrites{f_b \cdot |\alpha_{ba}|}
\left\{
\begin{array}{lr}
 a\conc b \conc b & \quad \text{if } \alpha_{ba}>0 \\
 a & \quad \text{if } \alpha_{ba}<0
\end{array}
\right.
$$
where $f_i=\frac{r_i}{K_i}$ is obtained from the usual growth rate and carrying capacity. The $\alpha$ coefficients are put in absolute value to compute the rate of the rule, their signs affect the right hand part of the rewrite rule.
\end{model}

\begin{example}
Mutualism has driven the evolution of much of the biological diversity we see today, such as flower forms (important to attract mutualistic pollinators) and co-evolution between groups of species~\cite{Tho05}.
We consider two different species of pollinators, $a$ and $b$, and two different species of angiosperms (flowering plants), $c$ and $d$. The two pollinators compete between each other, and so do the angiosperms. Both species of pollinators have a mutualistic relation with both angiosperms, even if $a$ slightly prefers $c$ and $b$ slightly prefers $d$. For each of the species involved we consider the rules for the logistic model and for each pair of species we consider the rules for competition and mutualism. The parameters used for this model are in Table~\ref{TabMC}. So, for example, the mutualistic relations between $a$ and $c$ are expressed by the following CWC rules
$$
\top: a \conc c \srewrites{\frac{r_a}{K_a}\cdot \alpha_{ac}} a \conc a \conc c
\quad \quad \quad
\top: a \conc c \srewrites{\frac{r_c}{K_c}\cdot \alpha_{ca}} a \conc c \conc c
$$
Figure~\ref{FigCM} shows a simulation obtained starting from a system with 100 individuals of species $a$ and $b$ and 20 individuals of species $c$ and $d$. Note the initially balanced competition between pollinators $a$ and $b$. This random fluctuations are resolved by the ``long run'' competition between the angiosperms $c$ and $d$: when $d$ predominates over $c$ it starts favouring the pollinator $b$ that now can win its own competition with pollinator $a$. The model is completely symmetrical: in other runs, a faster casual predominance of a pollinator may lead the evolution of its preferred angiosperm. The CWC model describing this example can be found at: \url{http://www.di.unito.it/~troina/cmc13/compmutu.cwc}.

\begin{table}
\footnotesize
\centering
\begin{tabular}{|c|c|c|c|c|c|c|}
\hline
\textbf{Species} ($i$)	&	$r_i$		&	$K_i$	&	$\alpha_{ai}$	& $\alpha_{bi}$	& $\alpha_{ci}$& $\alpha_{di}$	\\
\hline
$a$								&	0.2			& 1000		& $\bullet$ 			& -1						& +0.03			 & +0.01				\\
\hline
$b$								&  0.2			& 1000		& -1						& $\bullet$			& +0.01			& +0.03				\\
\hline
$c$								&  0.0002		& 200		& +0.25				& +0.1					& $\bullet$ 		& -6						\\
\hline
$d$								&  0.0002		& 200		& +0.1					& +0.25				& -6					 & $\bullet$			\\
\hline
\end{tabular}
\normalsize
\caption{\label{TabMC} Parameters for the model of competition and mutualism.}
\end{table}

\begin{figure}
\centering
\includegraphics[height=50mm]{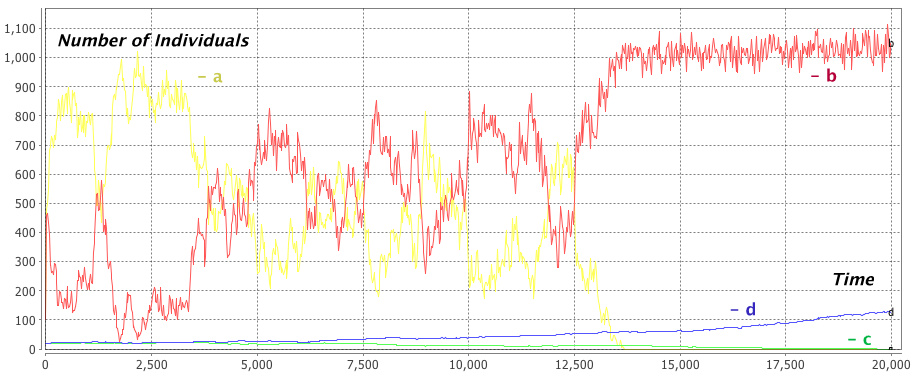}
\caption{\label{FigCM} Competition and Mutualism.}
\end{figure}

\end{example}

\subsection{Trophic Networks}

A \emph{food web} is a network mapping different species according to their alimentary habits. The edges of the network, called \emph{trophic links}, depict the feeding pathways (``who eats who'') in an ecological community~\cite{Elt27}. At the base of the food web there are autotroph species\footnote{Self-feeding: able to produce complex organic compounds from simple inorganic molecules and light (by \emph{photosynthesis}) or inorganic chemical reactions (\emph{chemosynthesis}).}, also called basal species. A \emph{food chain} is a linear feeding pathway that links monophagous consumers (with only one exiting trophic link) from a top consumer, usually a larger predator, to a basal species. The length of a chain is given by the number of links between the top consumer and the base of the web. The influence that the elements of a food web have on each other determine important features of an ecosystem like the presence of strong interactors (or \emph{keystone species}), the total number of species, and the structure, functionality and stability of the ecological community.

To model quantitatively a trophic link between species $a$ and $b$ (i.e., a particular kind of competition) we might use Lotka-Volterra equations~\cite{Vol26}:
$$
\begin{array}{l}
    \frac{dN_b}{dt} = N_b \cdot (r_b-\alpha \cdot N_a) \\
    \frac{dN_a}{dt} = N_a\cdot (\beta \cdot N_b - d)
\end{array}
$$
where $N_a$ and $N_b$ are the numbers of predators and preys, respectively, $r_b$ is the rate for prey growth, $\alpha$ is the prey mortality rate for per-capita predation, $\beta$ models the efficiency of conversion from prey to predator and $d$ is the mortality rate for predators.

\paragraph{Trophic Links} Within a compartment of type $\ell$, given a predation mortality $\alpha$ and conversion from prey to predator $\beta$, we can encode in CWC a trophic link between individuals of species $a$ (predator) and $b$ (prey) by the following rules:
$$
\begin{array}{l}
\ell: a \conc b \srewrites{\alpha} a\\
\ell: a \conc b \srewrites{\beta} a \conc a \conc b
\end{array}
$$
Here we omitted the rules for the prey exponential growth (absent predators) and predators exponential death (absent preys). These factors are present in the Lotka-Volterra model between two species, but could be substituted by the effects of other trophic links within the food web. In a more general scenario, a trophic link between species $a$ and $b$ could be expressed condensing the two rules within the single rule:
$$
\ell: a \conc b \srewrites{\gamma} a \conc a
$$
with a rate $\gamma$ modelling both the prey mortality rate and the predator conversion factor.

\begin{example}
Trophic cascades occur when predators in a food web suppress the abundance of their prey, thus limiting the predation of the next lower trophic level. For example, an herbivore species could be considered in an intermediate trophic level between a basal species and an higher predator. Trophic cascades are important for understanding the effects of removing top predators from food webs, as humans have done in many ecosystems through hunting or fishing activities. We consider a three-level food chain between species $a$, $b$ and $c$. The basal species $a$ reproduces with the logistic model, the intermediate species $b$ feeds on $a$, species $c$ predates species $b$:
\footnotesize
$$
\ell: a \srewrites{0.4} a \conc a \qquad
\ell: a \conc a\srewrites{0.0002} a
\qquad \qquad
\ell: a \conc b \srewrites{0.0004} b \conc b \qquad
\ell: b \conc c \srewrites{0.0008} c \conc c
$$
\normalsize
Individuals of species $c$ die naturally, until an hunting species enters the ecosystem. At a rate lower than predation, $b$ may also die naturally (absent predator). An atom $h$ may enter the ecosystem and start hunting individuals of species $c$:
\footnotesize
$$
\ell: c \srewrites{0.52} \emptyseq \qquad
\ell: b \srewrites{0.03} \emptyseq
\qquad \qquad
\top: h \conc (x \into X)^\ell \srewrites{0.003} (x \into X \conc h)^\ell \qquad
\ell: h \conc c \srewrites{0.5} h
$$
\normalsize
Figure~\ref{FigTC} shows a simulation for the initial term $h \conc (\emptyseq \into 1000*a \conc 100*b \conc 10*c )^\ell$. When the hunting activity starts, by removing the top predator, a top-down cascade destroys the whole community. The CWC model describing this example can be found at: \url{http://www.di.unito.it/~troina/cmc13/trophic.cwc}.
\begin{figure}
\centering
\includegraphics[height=50mm]{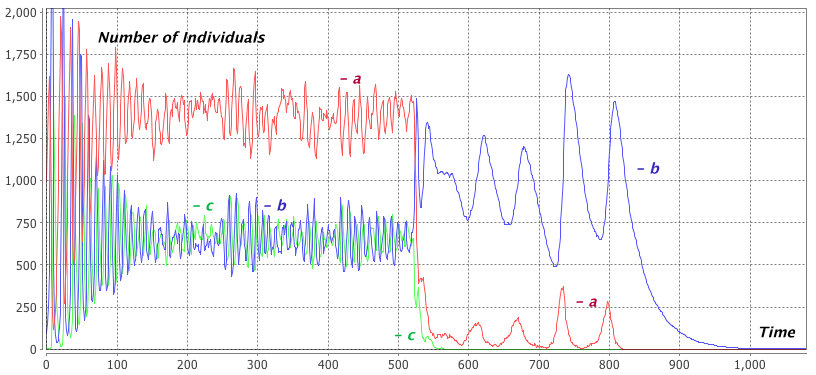}
\caption{\label{FigTC} A Throphic Cascade.}
\end{figure}
\end{example}

\section{\emph{Croton wagneri} and Climate Change}\label{Sec:Croton}

The knowledge of the relationships between the various attributes that define natural ecosystems (composition, structure, biotic and abiotic interactions, etc.) is crucial to improve our understanding of their functioning and dynamics. Much in the same way, this knowledge could allow us to evaluate and predict the ecological impacts of global change and to establish the appropriate measures to preserve the intrinsic characteristics of the ecological systems minimising the environmental impact.

Dry ecosystems are characterised by the presence of discontinuous vegetation that may reflect less than 60\% of the available landscape. The main pattern in arid ecosystems is a vegetation mosaic composed of patches and clear sites.
Dry ecosystems constitute an ideal model to analyse the relationships between the attributes of the ecosystem and its functioning at different levels of organisation. This becomes particularly important for the crucial need, for this kind of ecosystem, to deal with environmental problems like the loss of biodiversity and climate change \cite{ReySS02}. Morphological and functional changes shown by each species of these ecosystems are generally very different. As a consequence, the continuously adapting structure of the community will depend on the plasticity of each species in response to the major environmental changes. In the last decades, changes in the morphology of a species have been investigated from a functional perspective. This kind of analysis becomes particularly relevant because, for example, the existence of plant species in environments with extreme climatic conditions may depend not only on the ability of the plant to structurally specialise but also on its capacity to adapt metabolically. Another key factor of dry ecosystems is the possibility to study special patterns of the vegetation distribution to understand whether the plant communities show competitive or mutualistic relations. In particular, in dry (infertile) ecosystems, when competitive relations predominate, the spatial distribution of the vegetation tends to be uniform (with a quite regular distance between the plants); when mutualistic relations predominate, the vegetation tends to form clusters of plants; if none of these relations is observed, the spatial distribution of the vegetation is more random \cite{WM04}.

The study site is located in a dry scrub in the south of Ecuador ($03^\circ 58' 29''$ S, $01^\circ 25' 22''$ W) near the Catamayo Valley, with altitude ranging from 1400m to 1900m over the sea level. Floristically, in this site we can find typical species of xerophytic areas (about 107 different species and 41 botanical families). The seasonality of the area directly affects the species richness: about the 50\% of the species reported in the study site emerge only in the rainy season. Most species are shrubs although there are at least 12 species of trees with widely scattered individuals, at least 50\% of the species are herbs.
The average temperature is $20^\circ$ C with an annual rainfall around 600 mm, the most of the precipitation occurs between December and March, the remaining months present a water deficit with values between 200 and 600 mm of rain annually.
Generally, this area is composed by clay, rocky and sandy soils~\cite{AKS05}.

In~\cite{Gen93} about 1300 different species belonging to the dry ecosystems in Northwest South America have been identified. For this study we will focus on data available for the species \emph{Croton wagneri M\"{u}ll. Arg.}, belonging to the Euphorbiaceae family. This species, particularly widespread in tropical regions, can be identified by the combination of latex, alternate simple leaves, a pair of glands at the apex of the petiole, and the presence of stipules. \emph{C. wagneri} is the dominant endemic shrub in the dry scrub of Ecuador and has been listed as Near Threatened (NT) in the Red Book of Endemic Plants of Ecuador~\cite{VPLJ00}. This kind of shrub could be considered as a nurse species\footnote{A nurse plant is one with an established canopy, beneath which germination and survival are more likely due to increased shade, soil moisture, and nutrients.} and is particularly important for its ability to maintain the physical structure of the landscape and for its contribution to the functioning of the ecosystem (observing a marked mosaic pattern of patches having a relatively high biomass dispersed in a matrix of poor soil vegetation)~\cite{Gut01}.

In the study area, 16 plots have been installed along four levels of altitude gradients (1400m, 1550m, 1700m and 1900m): two 30mx30m plots per gradient in plane terrain and two 30mx30m plots per gradient in a slope surface (with slope greater than $10^\circ$). The data collection survey consisted in enumerating all of the \emph{C. wagneri} shrubs in the 16 plots: the spatial location of each individual was registered using a digital laser hypsometer. Additionally, plant heights were measured directly for each individual and the crown areas were calculated according to the method in~\cite{SACH11}. Weather stations collect data about temperatures and rainfall for each altitude gradient. An extract of data collected from the field can be found at: \url{http://www.di.unito.it/~troina/croton_data_extract.xlsx}.
This data show a morphological response of the shrub to two factors: temperature and terrain slope. A decrease of the plant height is observed at lower temperatures (corresponding to higher altitude gradients), or at higher slopes.

\subsection{The CWC model}

A simulation plot is modelled by a compartment with label $P$. Atoms $g$, representing the plot gradient (one $g$ for each metre of altitude over the level of the sea), describe an abiotic factor put in the compartment wrap.

According to the temperature data collected by the weather stations we correlate the mean temperatures in the different plots with their respective gradients. In the content of a simulation plot, atoms $t$, representing 1$^\circ$C each, model its temperature. Remember that, in this case, the higher the gradient, the lower the temperature. Thus, we model a constant increase of temperature within the simulation plot compartment, controlled by the gradient elements $g$ on its wrap:
$$
\top: (x \into X)^P  \srewrites{1}   (x \into t \conc X)^P \qquad
\top: (g \conc x \into t \conc X)^P  \srewrites{0.000024}   (g \conc x \into  X)^P
$$


Atoms $i$ are also contained within compartments of type $P$, representing the complementary angle of the plot's slope (e.g., $90*i$ for a plane plot or $66*i$ for a 24$^\circ$ slope).

We model \emph{C. wagneri} as a CWC compartment with label $c$. Its observed trait, namely the plant height, is specified by atomic elements $h$ (representing one mm each) on the compartment wrap.

To model the shrub heights distribution within a parcel, we consider the plant in two different states: a ``young'' and an ``adult'' state. Atomic elements $y$ and $a$ are exclusively, and uniquely, present within the plant compartment in such a way that the shrub height increases only when the shrub is in the young state ($y$ in its content). The following rules describe (i) the passage of the plant from $y$ to $a$ state with a rate corresponding to a 1 year average value, and (ii) the growth of the plant, affected by temperature and slope, with a rate estimated to fit the field collected data:
$$
c: y  \srewrites{0.00274}  a \qquad \qquad
P: t \conc i \conc (x \into y \conc X)^{c}  \srewrites{0.000718}   t \conc i \conc (x \conc h \into y \conc X)^c
$$

\subsection{Simulation results}
Now we have a model to describe the distribution of  \emph{C. wagneri} height using as parameters the plot's gradient ($n*g$) and slope ($m*i$). Since we do not model explicitly interactions that might occur between \emph{C. wagneri} individuals, we consider plots containing a single shrub. Carrying on multiple simulations, through the two phase model of the plant growth, after 1500 time units (here represented as days), we get a snapshot of the distribution of the shrubs heights within a parcel. The CWC model describing this application can be found at: \url{http://www.di.unito.it/~troina/cmc13/croton.cwc}.

Each of the graphs in Figure~\ref{FigC} is obtained by plotting the height deviation of 100 simulations with initial term $(n*g \into m*i \conc (\emptyseq \into y)^c)^P$.
The simulations in Figures~\ref{FigC} (a) and (c) reflect the conditions of real plots and the results give a good approximation of the real distribution of plant heights. Figures~\ref{FigC} (b) and (d) are produced considering an higher slope than the ones on the real plots from were the data has been collected. These simulation results can be used for further validation of the model by collecting data on new plots corresponding to the parameters of the simulation.

\begin{figure}
\centering
\subfigure[$1400*g$ and $90*i$] {
\includegraphics[height=45mm]{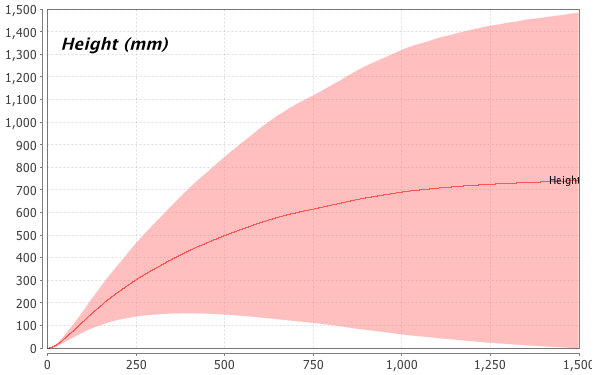}
}
\centering
\subfigure[$1550*g$ and $60*i$] {
\includegraphics[height=45mm]{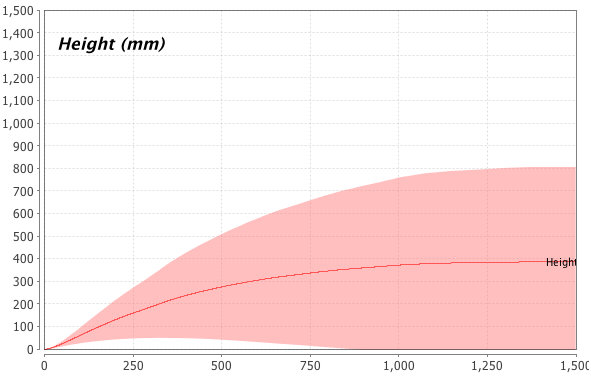}
}
\subfigure[$1700*g$ and $85*i$] {
\includegraphics[height=45mm]{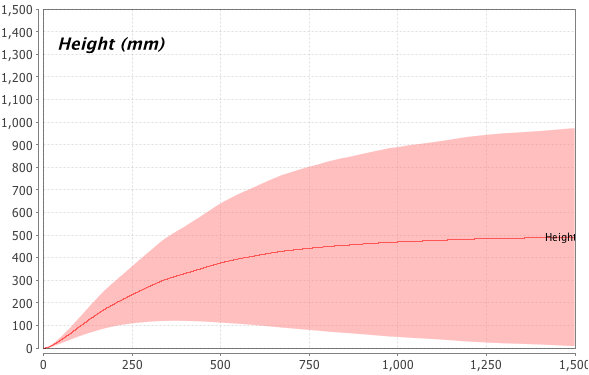}
}
\subfigure[$1900*g$ and $75*i$] {
\includegraphics[height=45mm]{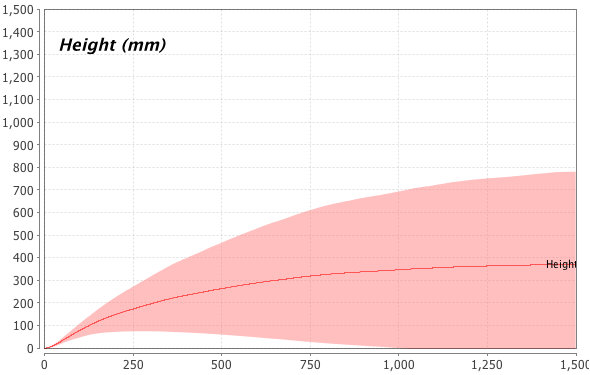}
}

\caption{Deviation of the height of \emph{Croton wagneri} for 100 simulations.}
\label{FigC}
\end{figure}

If we already trust the validity of our model, we can remove the correlation between the gradient and the temperature, and directly express the latter. Predictions can thus be made about the shrub height at different temperatures, and how it could adapt to global climate change. Figure~\ref{FigCT} shows two possible distributions of the shrub height at lower temperatures (given it will actually survive these more extreme conditions and follow the same trend).

\begin{figure}
\centering
\subfigure[$12^\circ$C, plain terrain] {
\includegraphics[height=45mm]{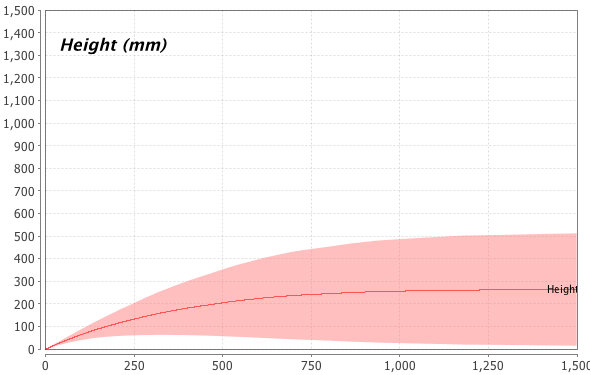}
}
\centering
\subfigure[$10^\circ$C, plain terrain] {
\includegraphics[height=45mm]{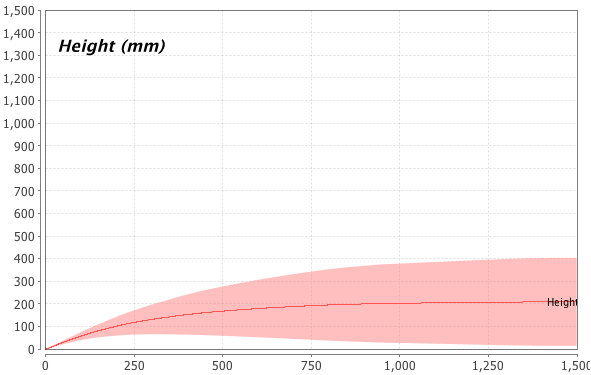}
}

\caption{Deviation of the height of \emph{Croton wagneri} for 100 simulations.}
\label{FigCT}
\end{figure}

\chapter{Spatial Models with CWC}
\label{spatial}
For the well-mixed chemical systems (even divided into nested
compartments) often found in cellular biology, interaction and
distribution analysis are sufficient to study the system's
behaviour. However, several complex biological phenomena include aspects in which space plays an essential role, key examples are the growth of tissues and organisms, embryogenesis and morphogenesis processes or cell proliferation. Thus, a realistic
modelling of biological processes requires space to be taken into account~\cite{K06}.
This has encouraged, in recent years, the development of formal models for the description of biological systems in which spatial properties can be modelled~\cite{CG10,barbuti2011SCLS,MV10}, as required by the emerging field of \emph{spatial systems biology}~\cite{SMG11} which aims at integrating the analysis of biological systems with spatial properties.

With CWC one might exploit the notion of compartment to represent spatial regions (given a fixed topology) in which the labels play a key role in defining the spatial properties. As we have done in Example~\ref{immigration} the movement and growth of system elements are described via specific rules involving adjacent compartments and the functionalities of biological or ecological components are affected by the spatial constraints given by the sector or region in which they interact with other elements.

\section{Case Studies with Topological Properties}

In this section we provide some hint about spatial modelling and analysis within the CWC framework by means of some paradigmatic examples. For simplicity,
we will consider the compartments in our system to be \emph{well-stirred}. Thus, also the compartments representing the spatial sectors of our topology
will consist of well-mixed components. Note that the topological analysis described in the following examples could be expressed also with other calculi
with compartmentalisation such as, for example, BioAmbients~\cite{RPSCS04}, Brane Calculi~\cite{Car04} and Beta-Binders~\cite{DPPQ06}.

\subsection{Developmental Biology: Retinal Layers}

The vertebrate retina contains seven major cell types, six
neuronal and one glial. These cells are derived, during the embryonal stage, from
multipotent retinal stem (progenitor) cells. The retina could be considered a stack of several neuronal layers (see Figure~\ref{FIG_RET}), where interactions specific to a certain developmental
stage are mediated by local conditions, which
affect the course in which the cells acquire their definite
phenotype~\cite{Gri01,ADT99}.

\begin{figure}[t]
\begin{center}
\begin{minipage}{0.98\textwidth}
\begin{center}
\includegraphics[width=50mm]{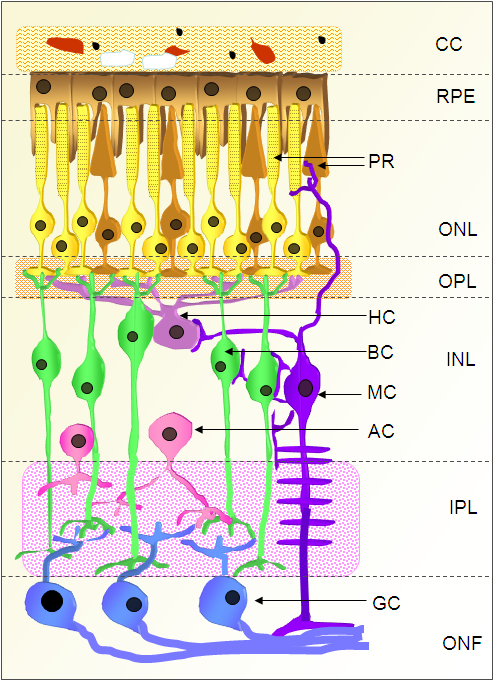}
\end{center}
\caption{\label{FIG_RET}Structural scheme of the vertebrate retina: (CC) choroid coat, (RPE) retinal pigment epithelium, (ONL) outer nuclear layer, (OPL) outer plexiform layer, (INL) inner nuclear layer, (IPL) inner plexiform layer, (ONF) optic nerve fibers. Cell types: (PR) photoreceptor cells: cones and rods, (HC) horizontal cells, (BC) bipolar cells, (MC) M\"uller cells, (AC) amacrine cells, (GC) ganglion cells.}
\end{minipage}
\end{center}
\end{figure}

The analysis of retinal layers appears to be a solid case study in a spatial context.
As a simple modelling exercise, we can describe in CWC a topological abstraction (we use the $\ldots$ notation to avoid details) of the retinal layers with the term:

$$
\begin{array}{ll}
\texttt{RETINA}= & ( \ldots \into \ldots )^{CC} \conc ( \ldots \into \ldots )^{RPE} \conc\\
 & ( \ldots \into (\ldots \into \ldots)^{cone} \conc (\ldots \into \ldots)^{rod} \conc \ldots )^{ONL} \conc\\
 & ( \ldots \into \ldots )^{OPL} \conc\\
 & ( \ldots \into (\ldots \into \ldots)^{hc} \conc (\ldots \into \ldots)^{bc} \conc (\ldots \into \ldots)^{mc} \conc (\ldots \into \ldots)^{ac} \conc \ldots )^{INL} \conc\\
 & ( \ldots \into \ldots )^{IPL} \conc\\
 & ( \ldots \into (\ldots \into \ldots)^{gc} \conc \ldots )^{ONF}
\end{array}
$$
where labels with capital characters denote the topological layers and labels with lower-case characters denote the cell types.

In the embryonal stage, all layers may contain only retinal stem (progenitor) cells with different signals, on the different layers, leading different cellular differentiation. Namely, the initial condition of the $ONL$ and $INL$ layer could be expressed by:

$$
(  \into (\ldots \into \ldots)^{stem} \conc (\ldots \into \ldots)^{stem} \conc \ldots )^{ONL} \conc (  \into (\ldots \into \ldots)^{stem} \conc (\ldots \into \ldots)^{stem} \conc \ldots )^{INL}
$$

and rules modelling cellular differentiation in the above layers could be:

$$
\begin{array}{lcrl}
&ONL & : & (x \into X)^{stem} \srewrites{k_{cone}} (x \into X)^{cone}\\
&ONL & : & (x \into X)^{stem} \srewrites{k_{rod}} (x \into X)^{rod}\\
&INL & : & (x \into X)^{stem} \srewrites{k_{hc}} (x \into X)^{hc}\\
&INL & : & (x \into X)^{stem} \srewrites{k_{ac}} (x \into X)^{ac}\\
& \ldots & &
\end{array}
$$

The reader might have noticed that, according to Figure~\ref{FIG_RET}, some of the retinal cells may expand through different layers (see, for example,
M\"uller cells). In the representation above we located these cells in the layer where the cell lays with the main part of its body (and its nucleus).
However, a more complex structure could be built if we slightly enrich the CWC framework by defining the compartment labels as (non recursive) algebraic
data types \`{a} la ML~\cite{SMLrev}, and adding the feature to create fresh names. In such a way, complexes with the same, unique, name may lay on
different topological sectors. For the moment, we avoid to implement this extension in the CWC framework, since it is not explicitly needed in the AM
symbiosis interactions we are considering. However, we leave this feature as a very close future work.

\subsection{Cell Growth and Proliferation: A CWC Grid}\label{SECT_GRID}

Eukaryotic cells reproduce themselves by cell cycle which can be divided in two brief periods: the interphase during which the cell grows, accumulating nutrients and duplicating its DNA and the mitosis ($M$) phase, during which the cell splits itself into two distinct cells, often called ``daughter cells''.
Interphase proceeds in three stages, $G1$, $S$, and $G2$. $G1$ (Gap1) is also called the growth phase, and is marked by synthesis of various enzymes that are required in $S$ phase, mainly those needed for DNA replication. In $S$ phase (Synthesis) DNA replication occurs. The cell then enters the $G2$ (Gap2) phase, which lasts until the cell enters mitosis. Again, significant biosynthesis occurs during this phase, mainly involving the production of microtubules, which are required during the process of mitosis.


Spatial representation of the growth of cell populations could be useful in many situations (in the study of bacteria colonies, in developmental biology, and in the analysis of cancer cells proliferation).

Intuitively, we may express cell proliferation by describing a finite space through a CWC grid such that on each  \emph{grid cell} we can insert a cell of the growing population under analysis. An $k\times n$ grid could be modelled in CWC with $k*n$ compartment labels modelling the \emph{grid cells}. A cell may grow in one of the adjacent \emph{grid cells} (with a given scheme of adjacency modelled by the set of rewrite rules).

Intuitively, using the atomic element $e$ to denote an empty \emph{grid cell}, an empty grid could be modelled by the following CWC term:

$$
\begin{array}{ll}
\texttt{Prol\_Grid}= & (  \into e )^{1,1} \conc \ldots \conc (  \into e )^{1,n} \conc\\
& (  \into e )^{2,1} \conc \ldots \conc (  \into e )^{2,n} \conc\\
& \ldots \\
& (  \into e )^{m,1} \conc \ldots \conc (  \into e )^{m,n} \conc\\
\end{array}
$$

If we represent a cell of our population as a compartment with membrane $m$ and with a single atom in its content representing the phase of the cell cycle ($M$, $G1$, $S$, $G2$), a \emph{grid cell} $i,j$ may contain the cell $(m \into M)^{cell}$ that could grow, with kinetic $k_{mit}$, towards the empty \emph{grid cell} $i,j+1$ with the rewrite rule
\footnote{Note that under our assumptions the variables 
will always match only with the empty multiset}:

$$
\begin{array}{lcrl}
&\TOP & : & ( x \into  ( m \conc y \into M \conc Y)^{cell} \conc X )^{i,j} ( z \into  e \conc Z )^{i,j+1} \srewrites{k_{mit}}  \\
&     &   & ( x \into  ( m \conc y \into G1 \conc Y)^{cell} \conc X  )^{i,j} ( z \into  ( m \into G1)^{cell} \conc Z )^{i,j+1} \\
\end{array}
$$

\noindent  while cells may change their internal state (cell cycle phase) with the rules:

$$
\begin{array}{lcrl}
&i,j & : & ( m \conc x \into G1 \conc X)^{cell}  \srewrites{k_{phase1}}  ( m \conc x \into S \conc X)^{cell}\\
&i,j & : & ( m \conc x \into S \conc X)^{cell}  \srewrites{k_{phase2}}  ( m \conc x \into G2 \conc X)^{cell}\\
&i,j & : & ( m \conc x \into G2 \conc X)^{cell}  \srewrites{k_{phase3}}  ( m \conc x \into M \conc X)^{cell}
\end{array}
$$

Note that each element of the grid can be occupied only by one cell, so this model allows only a limited growth of the cells, corresponding to the dimension of the grid. However this could be enough to investigate the properties of the cell growth.

\subsection{Quorum Sensing and Molecular Diffusion}\label{SECT_QS}

Molecular diffusion (or simply diffusion) arises form the motion of particles as a function of temperature, viscosity of the medium and the size (or mass) of the particles. Usually, diffusion explains the flux of molecules from a region of higher concentration to one of lower concentration, but it may also occur when there is no concentration gradient. The result of the diffusion process is a gradual mixing of the involved particles. In the case of uniform diffusion (given a uniform temperature and absent external forces acting on the particles) the process will eventually result in a complete mixing.
Equilibrium is reached when the concentrations of the diffusing molecules between two compartments becomes equal.

\begin{figure}
\begin{center}
\begin{minipage}{0.98\textwidth}
\begin{center}
\includegraphics[width=80mm]{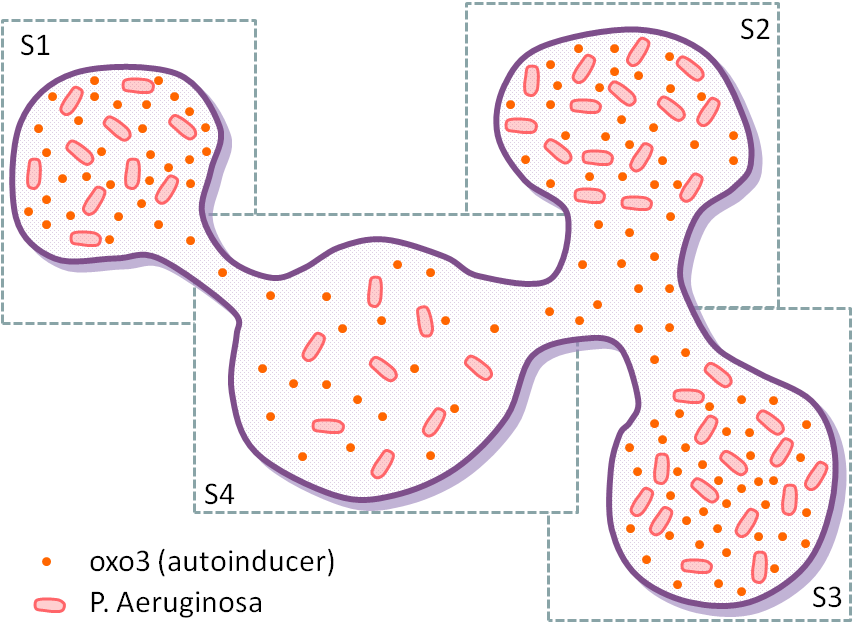}
\end{center}
\caption{\label{FIG_QS}Quorum Sensing in a distributed topology.}
\end{minipage}
\end{center}
\end{figure}
Diffusion is of fundamental importance in many disciplines of physics, chemistry, and biology.
In this subsection we will consider the diffusion of the auto-inducer communication signals in a quorum sensing process where the bacteria initiating the process are distributed in different sectors of an environment (see Figure~\ref{FIG_QS}). We consider the sectors of the environment to be connected through channels with different capacities and different speeds regulating the movement of the auto-inducer molecules between the different sectors.

\begin{figure}
\center
\includegraphics[width=.49\textwidth]{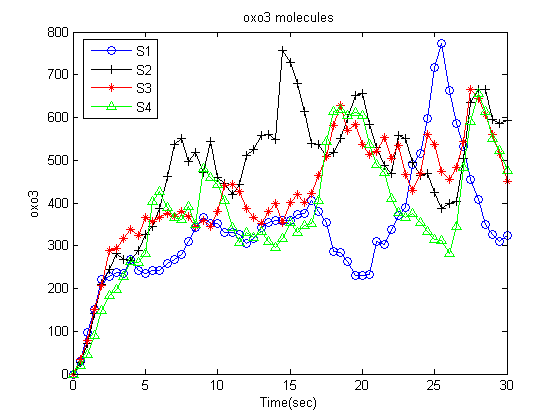}
\includegraphics[width=.49\textwidth]{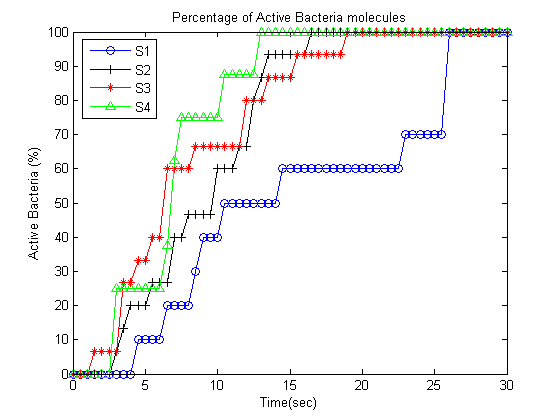}
\caption{\label{res_qs} Simulation results of oxo3 molecules (left figure) and active bacteria percentage with respect to the total number of bacteria (right figure) in the different sectors.}
\end{figure}
The CWC model describing the system can be found in Appendix A. We modelled the four sectors as four different compartments with special labels $Si$ allowing the flow of the oxo3 auto-inducer molecules with different speed. Bacteria (modelled as compartments with label $m$) were marked as ``active'' when they reached a sufficient high level of the oxo3 auto-inducer. The results of a simulation running for $30$ seconds is shown in Figure~\ref{res_qs}.

Although CWC does not have a direct mechanism to formulate the spatial position of its terms and a metric distance between them, it has the potential to describe topologically different spatial concepts. For example, the compartments are able, through the containment specification, to express discrete distance among the entities that are inside and those on the outside. Moreover, through the rule kinetics we are able to model space anisotropy (lack of uniformity in all orientations), which are useful to formulate complex concepts of reachability. The model discussed in this section is a clear example of these issues.

\section{Spatial Interactions in the Arbuscular Mycorrhiza Symbiosis}

Arbuscular mycorrhizae, which are formed between $80\%$ of land plants and arbuscular mycorrhizal (AM) fungi belonging to the Glomeromycota phylum, are the most common and widespread symbiosis on our planet \cite{harrison2005signaling}. AM fungi are obligate symbionts
which supply the host with essential nutrients such as phosphate, nitrate and other minerals from the soil. In return, AM fungi receive carbohydrates derived from photosynthesis in the host. AM symbiosis also confers resistance to the plant against pathogens and environmental stresses. Despite the central importance of AM symbiosis in both agriculture and natural ecosystems, the mechanisms for the formation of a functional symbiosis between plants and AM fungi are largely unknown.

\subsection{The Biological Interactions under Investigation}

The interaction begins with a molecular dialogue between the plant and the fungus \cite{harrison2005signaling}. Host roots release signalling molecules characterized as strigolactones \cite{akiyama2005plant} which are sesquiterpenes derived from the carotenoid pathway \cite{matusova2005strigolactone,lopez2008tomato} 
Within just a few hours, strigolactones at subnanomolar concentrations induce alterations in fungal physiology and mitochondrial activity and extensive hyphal branching \cite{besserer2006strigolactones}.
Root exudates from plants grown under phosphate-limited conditions are more active than those from plants with sufficient phosphate nutrition, suggesting that the production of these active exudates in roots is regulated by phosphate availability \cite{nagahashi2000partial,yoneyama2001production}.
Strigolactones also appear to act as chemo-attractants: fungal hyphae of \textit{Glomus mosseae} exhibited chemotropic growth towards roots at a distance of at least 910 $\mu$m in response to host-derived signals, possibly strigolactones \cite{sbrana2005chemotropism}.
Some solid evidence has been presented for AM fungal production of a long-hypothesized symbiotic signal, the ``Myc Factor'', in response to the plant symbiotic signal. Fungal hyphae of the genus \textit{Gigaspora} growing in the vicinity of host roots, but separated from the roots by a membrane, release a diffusible substance that induces the expression of a symbiosis-specific gene 
in \textit{Medicago truncatula} roots \cite{kosuta2003diffusible}. This expression was correlated both spatially and temporally with the appearance of hyphal branching, and was not observed when hyphal branching was absent.
The external signal released by these fungi is perceived by a receptor on the plant plasma membrane and is transduced into the cell with the activation of a symbiotic signalling pathway that lead to the colonization process.
One of the first response observed in epidermal cells of the root is a repeated oscillation of $Ca^{2+}$ concentration in the nucleus and perinuclear cytoplasm through the alternate activity of $Ca^{2+}$ channels and transporters \cite{oldroyd2008coordinating,kosuta2008differential,chabaud2011arbuscular}.
Kosuta and collegues \cite{kosuta2008differential} assessed mycorrhizal induced calcium changes in \textit{Medicago truncatula} plants transformed with the calcium reporter cameleon. Calcium oscillations were observed in a number of cells in close proximity to plant/fungal contact points. Very recently Chabaud and colleagues \cite{chabaud2011arbuscular} using root organ cultures of both \textit{Medicago truncatula} and \textit{Daucus carota} 
observed $Ca^{2+}$ spiking in AM-responsive zone of the root treated with AM spore exudate.

Once reached the root surface the AM fungus differentiates a hyphopodium via
which it enters the root. 
A reorganization of the plant cell occurs before the fungal penetration: a thick cytoplasmic bridge (called pre-penetration apparatus -PPA-) 
is formed and it constitutes a trans-cellular tunnel in which
the fungal hypha will grow \cite{genre2005arbuscular}. Following this event, a hyphal peg is produced which enters and crosses the epidermal cell, avoiding a direct contact between fungal
wall and host cytoplasm thanks to a surrounding membrane of host origin \cite{genre2007check}.
Than the fungus overcomes the epidermal layer and it grows inter-and intracellularly
all along the root in order to spread fungal structures.
Once inside the inner layers of the cortical cells the differentiation of specialized, highly branched intracellular hyphae called arbuscules occur. Arbuscules are considered the major site for nutrients exchange between the two organisms 
and they form by repeated dichotomous hyphal branching and grow
until they fill the cortical cell. They are ephemeral structures \cite{toth1984dynamics,javot2007medicago}: at some point after maturity, the arbuscules collapse, degenerate, and die and the plant cell regains its previous organization \cite{bonfante1984anatomy}. The mechanisms underlying arbuscule turnover are unknown
but do not involve plant cell death. In arbuscule-containing cells the two partners come in intimate physical contact and adapt their metabolisms to allow reciprocal benefits.

\subsection{The CWC Model}

We model the Arbuscular Mycorrhyzal symbiosis in a 2D space. The soil environment is partitioned into different concentric layers to account for the distance between the plant root cells and the fungal hyphae where the inner layer contains the plant root cells. The plant root tissue is also modelled as concentric layers each one composed by sectors.

For the sake of simplicity we modelled the soil into $5$ different concentric layers (modelled as different compartments labelled by $L_0$-$L_{4}$ identifying $L_4$ with the top level $\TOP$). Although the root has a complex layered structure we simplified it into $3$ concentric layers (one epidermal layer and $2$ cortical layers) each one composed by $5$ sectors mapped as different compartments labelled by $e_i$ ($i=1 \ldots 5$) for epidermal cells and $c_{jk}$ ($j=1, 2$ and $k=1 \ldots 5$) for cortical cells. Exploiting the flexibility of CWC, compartment labels are here used to characterize both their spatial position and their biological properties. All the compartments and the main atoms of the model are depicted in Figure~\ref{sym_schema}.

\begin{figure}
\center
\includegraphics[width=.5\textwidth]{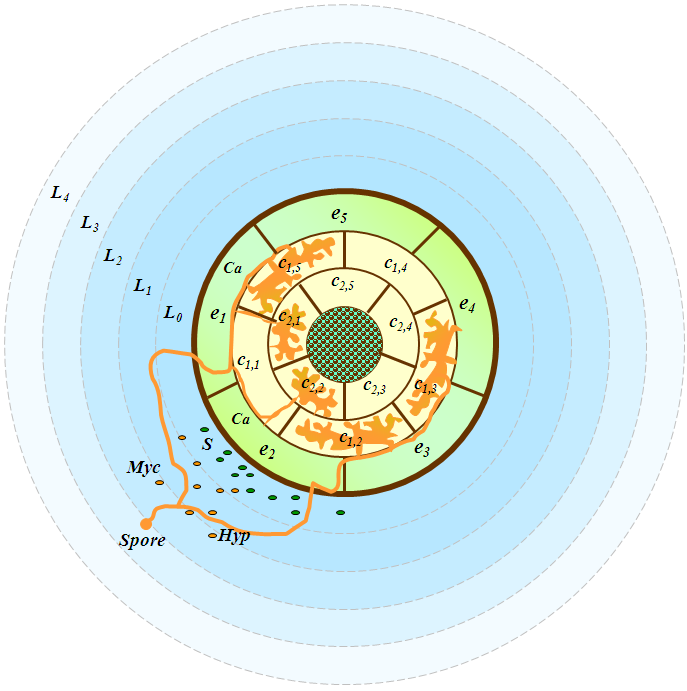}
\caption{\label{sym_schema}2D spatial model of the Arbuscular Mycorrhyzal symbiosis.}
\end{figure}

The stigolactones (atoms $S$) derived from the carotenoid pathway (atom $C_P$) are released by epidermal cells and can diffuse through the soil degrading their activity when farthest from the plant. This situation was modelled by the following rules:

\begin{equation}
\begin{array}{lcrlr}
(R_{CP_i}) & L_0 & : &  ( x \into C_P \conc X)^{e_i} \srewrites{K_{CP}} S ( x \into C_P \conc X)^{e_i} & i=1 \ldots 5\\
(R_{S_i})  & L_{i+1} & : &  ( x \into S \conc X)^{L_i} \srewrites{K_{S}} S ( x \into \conc X)^{L_i} &  i=0 \ldots 3\\
(R_{S_{4}})  & \TOP & : &  S \srewrites{K_{S}}  \emptyseq &
\end{array}
\end{equation}

Since phosphate (atoms $P$) conditions regulate strigolactones biosynthesis inhibiting the carotenoid pathway (atom $IC_P$), the following rules were considered:

\begin{equation}
\begin{array}{lcrlr}
(R_{ICP_i}) & L_0 & : & P ( x \into C_P \conc X)^{e_i} \srewrites{K_{ICP}} P ( x \into IC_P \conc X)^{e_i} & i=1 \ldots 5
\end{array}
\end{equation}

The fungal structure is described by atoms
where each atom represents a particular type of fungal component. The atom $Spore$ represents the fungal spore while atoms $Hyp$ represent fungal hyphae.
AM fungal hyphae branching towards the plant root is regulated by the strigolacones activity. If we locate the fungal spore at a fixed distance from the plant root (let us say at the layer $L_k$ where $0 < k < 4$) then the spore hyphal branching is generated stochastically in the relative inferior and superior layers. In our experiments the fungal spore was placed at the soil layer $L_2$. The hyphal branching is radial with respect to the spore position causing a branching away from the spore and promoted at the proximity of the root by the strigolactones. This branching process was represented by the following rules:

\begin{equation}
\begin{array}{lcrlr}
(R_{F_1}) & L_2 & : & Spore ( x \into  \conc X)^{L_{1}} \srewrites{K_{F}}  Spore ( x \into Hyp \conc X)^{L_{1}}  & \\
(R_{F_2}) & L_{3} & : & ( x \into Spore \conc X)^{L_{2}} \srewrites{K_{F}}  Hyp ( x \into Spore \conc X)^{L_{2}}  & \\
(R_{H_j}) & \TOP & : & ( x \into Hyp \conc X)^{L_{3}} \srewrites{K_{H}}  Hyp ( x \into Hyp \conc X)^{L_{3}}  & \\
(R_{H_1}) & L_{1} & : & S  \conc Hyp ( x \into \conc X)^{L_{0}} \srewrites{K_{H_1}}  S  \conc Hyp ( x \into Hyp \conc X)^{L_{0}}  &
\end{array}
\end{equation}

The AM fungal production of Myc factor (atom $Myc$) at the proximity of the root cells inducing the calcium spiking phenomenon (atom $Ca$) at a responsive state (atom $Resp$) which prepares the epidermal cells to include the fungal hyphae was modelled by the following rules:

\begin{equation}
\begin{array}{lcrlr}
(R_{Myc}) & L_0 & : & Hyp \srewrites{K_{M}} Myc \conc Hyp & \\
(R_{MycCa_i}) & L_0 & : & Myc ( Resp \conc x \into \conc X)^{e_i} \srewrites{K_{MC}}  (Resp \conc x \into \conc Ca \conc X)^{e_i} & i=1 \ldots 5\\
(R_{Ca_i}) & L_0 & : & Hyp ( Resp \conc x \into Ca \conc X)^{e_i} \srewrites{K_{Ca}} Hyp ( NResp \conc x \into Hyp \conc X)^{e_i} & i=1 \ldots 5
\end{array}
\end{equation}

where the epidermal cells change their state from responsive to not responsive (atom $NResp$) since they can be penetrated by only one hypha.

Once the fungal hyphae penetrate the epidermal cells they are allowed to branch in the two directions of the underneath cortical layer and from this  cortical layer to the inner one. The cortical cells change their state from responsive (atom $Resp$) to transitive (atom $Trans$) and then to mycorrhized (atom $Arb$) since an hypha can grow in the cortical intercellular spaces even if already mycorrhized but only one arbuscule can be generated. The hyphal branching from the epidermal layer to the cortical layers is modelled through the following rules:

\begin{equation}
\begin{array}{lcrlr}
(R_{CB_{1i}}) & L_0 & : & ( NResp \conc x \into Hyp \conc X)^{e_i} ( Resp \conc y \into \conc Y)^{c_{1j}} \srewrites{K_{CB}} \\ 
              &     &   & ( NResp \conc x \into Hyp \conc X)^{e_i} ( Trans \conc y \into  Hyp \conc Y)^{c_{1j}} & \\
& & & i=1 \ldots 5, j=i (mod 5), i+1 (mod 5); & \\
(R_{CB1_{2i}}) & L_0 & : & ( Trans \conc x \into Hyp \conc X)^{c_{1i}} ( Resp \conc y \into  \conc Y)^{c_{2j}} \srewrites{K_{CB}} \\
               &     &   &  ( Trans \conc x \into Hyp \conc X)^{c_{1i}} ( Trans \conc x \into  Hyp \conc Y)^{c_{2j}} & \\
 & & & i=1 \ldots 5, j=i (mod 5), i+1 (mod 5); & \\
(R_{CB2_{2i}}) & L_0 & : & ( Arb \conc x \into Hyp \conc X)^{c_{1i}} ( Resp \conc y \into  \conc Y)^{c_{2j}} \srewrites{K_{CB}}  \\
               &     &   & ( Arb \conc x \into Hyp \conc X)^{c_{1i}} ( Trans \conc y \into  Hyp \conc Y)^{c_{2j}} & \\
 & & & i=1 \ldots 5, j=i (mod 5), i+1 (mod 5); & \\
\end{array}
\end{equation}

Finally, the modelling of each fungal hyphae penetrated into a cortical cell to create an arbuscule (atom $A$) and start the symbiosis is performed by the following rules:

\begin{equation}
\begin{array}{lcrlr}
(R_{A_{ij}}) & L_0 & : & ( Trans \conc x \into Hyp \conc X)^{c_{ij}} \srewrites{K_{A}}   ( Arb \conc x \into Hyp \conc X)^{c_{ij}} & i=1, 2 ; j=1 \ldots 5\\
\end{array}
\end{equation}

\subsection{Simulation Results}

The initial term describing our system is given by $T$:
\begin{equation}
\begin{array}{lll}
T =& ( \into  \conc ( \into  \conc ( \into Spore \conc ( \into  \conc ( \into  \conc T')^{L_{0}})^{L_{1}})^{L_{2}})^{L_{3}})^{\TOP}\\
&\text{where} \\
&T'= V \conc U \\
& V = V_1 \conc  V_2 \conc  V_3 \conc  V_4 \conc  V_5\\
& V_i = ( Resp \into  \conc C_P)^{e_{i}}\\
& U = U_{11} \conc  U_{12} \conc  U_{13} \conc  U_{14} \conc  U_{15} \conc  U_{21} \conc   U_{22} \conc  U_{23} \conc  U_{24} \conc  U_{25}\\
& U_{jk} = ( Resp \into  \conc )^{c_{jk}}\\
\end{array}
\end{equation}
\begin{figure}
\center
\includegraphics[width=.49\textwidth]{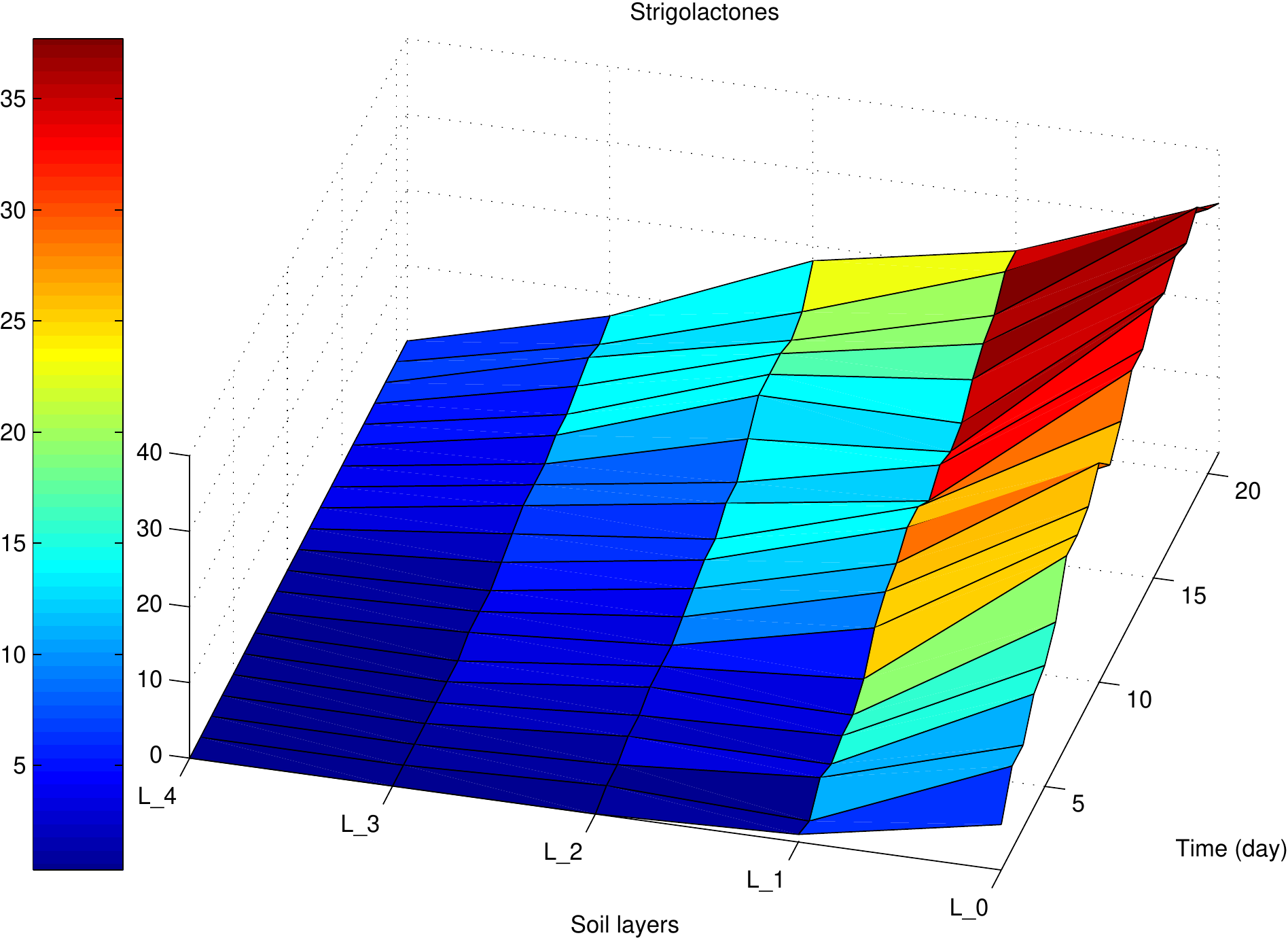}
\includegraphics[width=.49\textwidth]{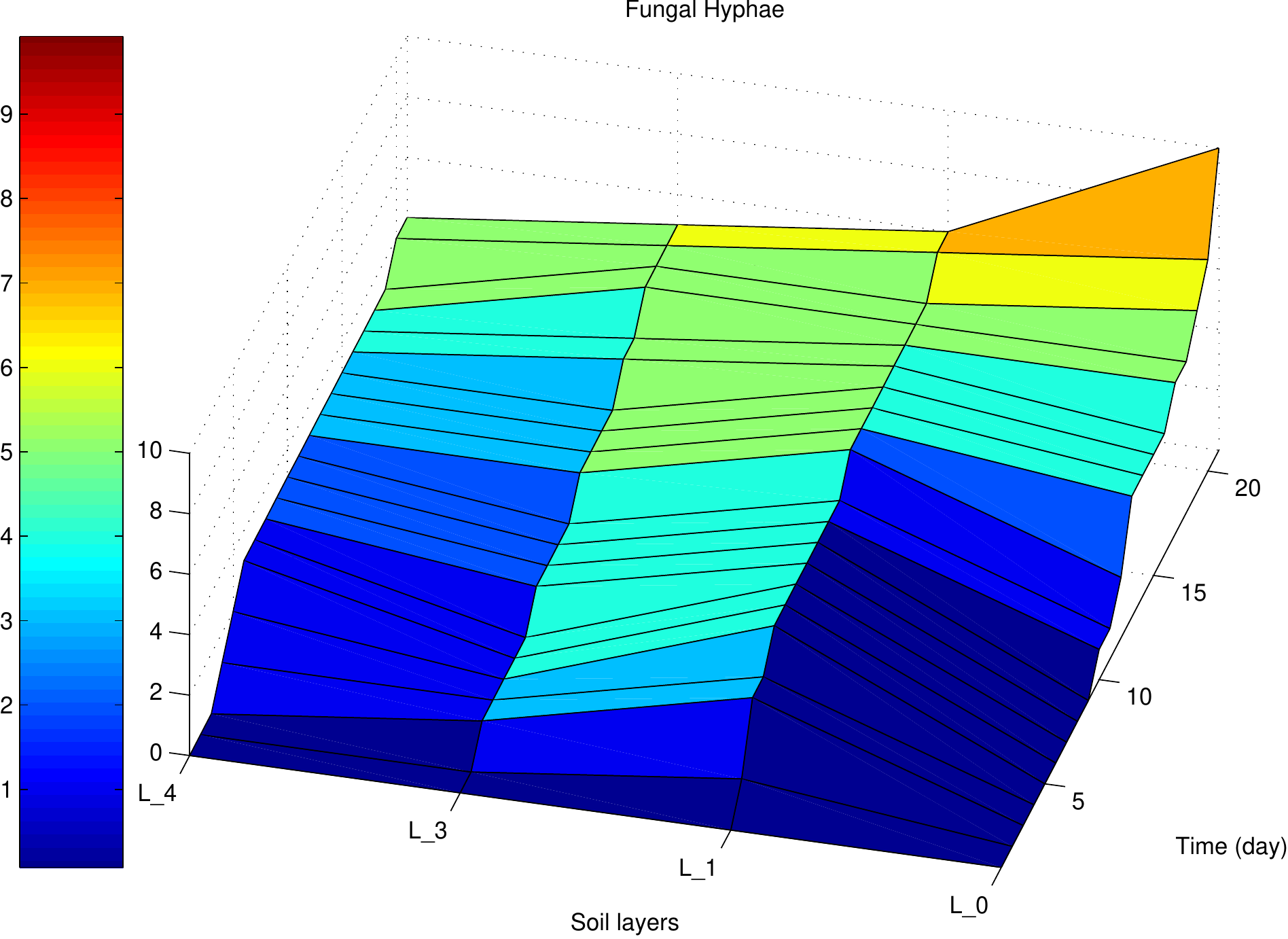}
\caption{\label{res_noP}Simulation results of strigolactones diffusion (left figure) and fungal hyphal branching (right figure) along the soil layers during time.}
\end{figure}

The model has been used to virtually simulate $21$ days of the fungus--plant interactions. The reaction rates have been chosen to resemble the experimental observations of ``in vivo'' fungal germination, hyphal branching and arbuscules creation under similar environmental conditions of the simulations. The rates were set as following: $K_{CP}= 0.8$, $K_{S}=0.1$, $K_{ICP}= 0.01$, $K_{F}=0.2$, $K_{H}=0.05$, $K_{H_1}
0.01$, $K_{M} = 0.2$, $K_{MC} = 0.2$, $K_{Ca} = 1.0$, $K_{CB} = 0.9$, $K_{A} = 1.0$.
In Figure \ref{res_noP} we show a surface plot representing an exemplification of the CWC simulation about strigolactones diffusion (on the left) and fungal hyphal branching (on the right) along the soil layers during time.

We compared our simulation results adding the phosphate ($100$ $P$ atoms) to the soil layers yielding to the results given by Figure \ref{res_highP}. In this case we have a lower strigolactones production by the root cells and an hyphal branching not promoted in the direction of the root.

We performed $60$ simulations in the three conditions of low phosphate ($10$ $P$ atoms), medium phosphate ($50$ $P$ atoms) and high phosphate ($100$ $P$ atoms) conditions leading to a $60\%$ of arbuscules in the first case, $30\%$ in the second case, and $10\%$ in the latter case. The first arbuscules occur around the $8^{th}-9^{th}$ day of simulation at low phosphate condition while they are formed from the $9^{th}$ to the $15^{th}$ day at high phosphate conditions.
These results reflect the in vivo experiments on several plants inoculated with different AM fungi \cite{mosse1973advances,asimi1980influence,hepper1983effect}.
\begin{figure}
\center
\includegraphics[width=.49\textwidth]{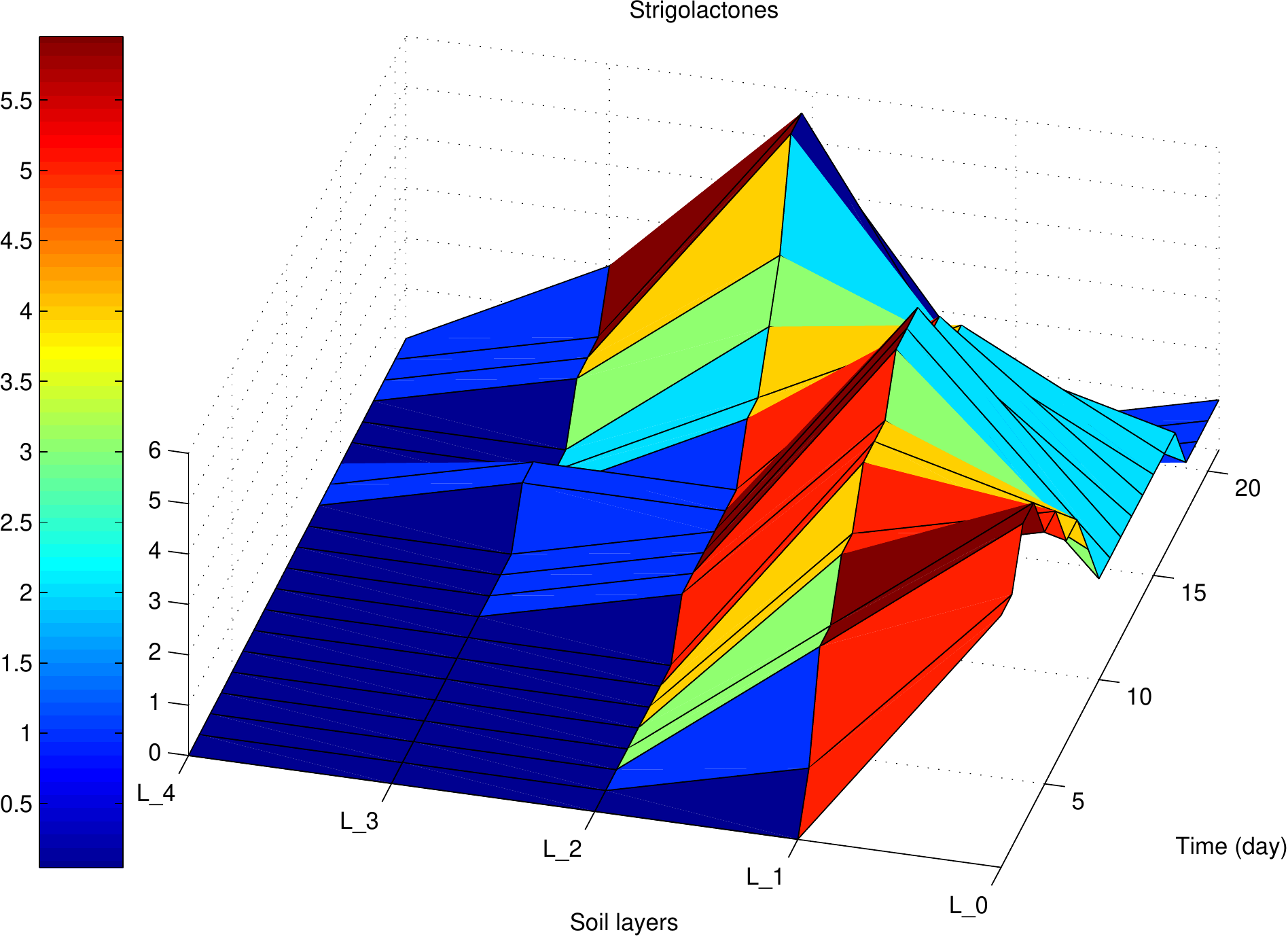}
\includegraphics[width=.49\textwidth]{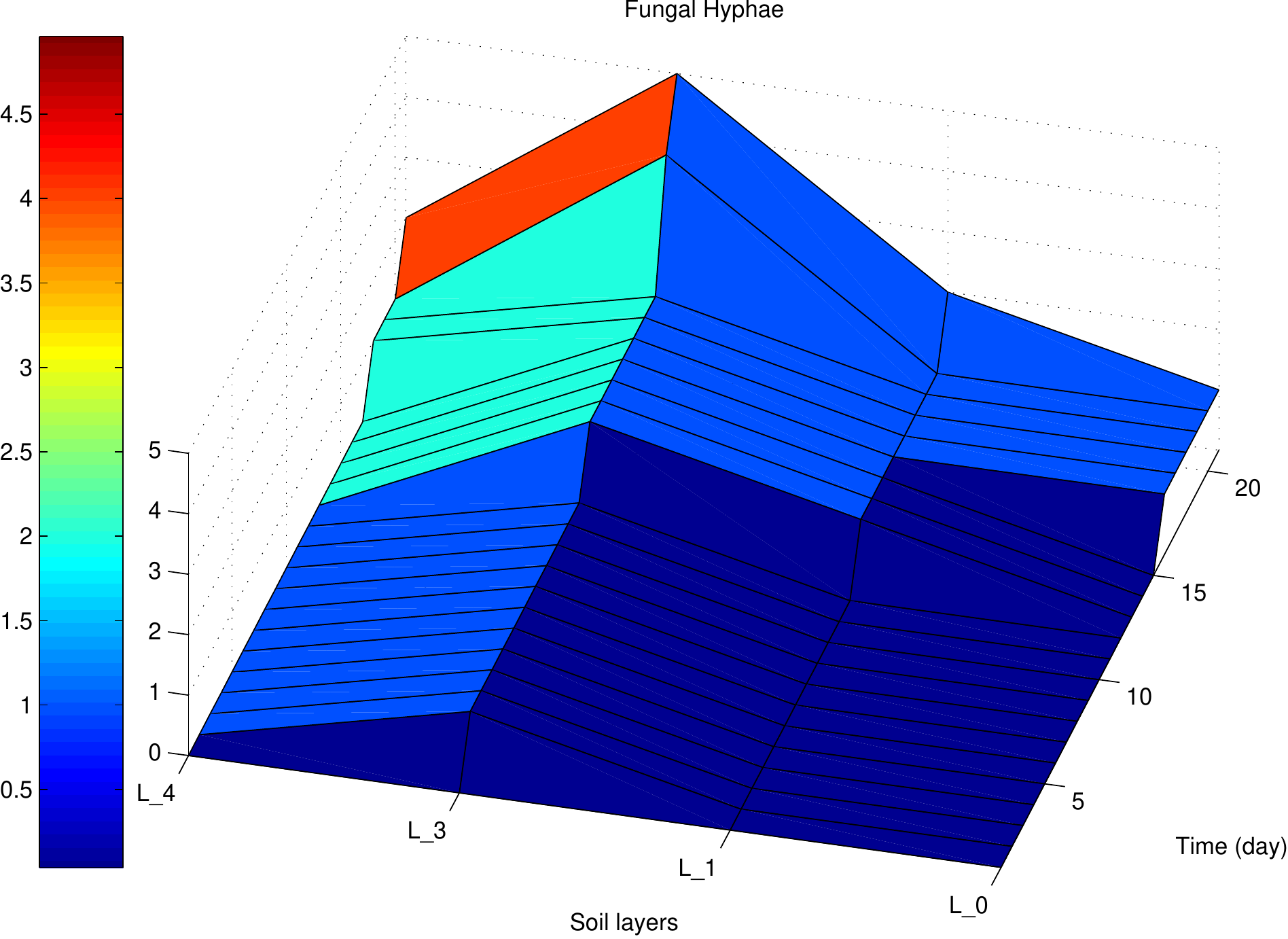}
\caption{\label{res_highP}Simulation results of strigolactones diffusion (left figure) and fungal hyphal branching (right figure) along the soil layers at high phosphate concentrations.}
\end{figure}

Our long term goal is to design a virtual spatial model that would enable us to test ``in silico'' various
hypotheses about the interaction between the physiological processes that drive the arbuscular mycorrizal symbiosis and the signalling processes that take place in this symbiosis.

\section{A Surface Language}

CWC allows to model several spatial interactions in a very natural way. However, when the complexity of simulation scenarios increases, the specification of large networks of compartments each one having its own peculiar behaviour and initial state may require a long and error prone modelling phase.

In this paper we introduce a surface language for CWC that defines a framework in which the notion of space is included as an essential component of the system. The space is structured as a square grid, whose dimension must be declared as part of the system specification.  The surface language provides basic constructs for modelling spatial interactions on the grid. These constructs can be compiled away to obtain a standard CWC model, thus exploiting the existing CWC simulation tool.

A similar approach can be found in \cite{MV10} where the topological structure of the components is expressed via explicit links which require ad-hoc rules to represent movements of biological entities and a logic-oriented language to flexibly specify complex simulation
scenarios is provided.

Thus, we embed CWC into a surface language able to express, in a synthetic form, both spatial (in a two-dimensional grid) and biochemical CWC transformations. The semantics of a surface language model is defined by translation into a standard CWC model.

We distinguish between two kind of compartments: 
 \begin{enumerate}
\item \emph{Standard} compartments (corresponding to the usual CWC compartments), used to represent entities (like bacteria or cells) that can move through space.
\item \emph{Spatial} compartments, used to represent portions of space. Each spatial compartment defines a location in a two dimensional grid
through a special atom, called \emph{coordinate}, that occurs on its wrap. A coordinate is denoted by \texttt{row.column}, where \texttt{row} and
\texttt{column} are intergers. Spatial compartments have distinguished labels, called \emph{spatial labels}, that can be used to provide a specific
characterisation of a portion of space.
\end{enumerate}
For simplicity we assume that the wraps of each spatial compartment contains only the coordinate. Therefore, spatial compartment differentiations can be
expressed only in terms of labels.\footnote{Allowing the wrap of spatial compartments to contain other atoms, thus providing an additional mean to express
spatial compartment differentiations, should not pose particular technical problems (extend the rules of the surface language to deal with a general wrap
content also for spatial compartments should be straightforward).}

For example, the spatial compartment $( \texttt{1.2}  \into 2*b)^{\textit{soil}}$ represents the cell of the grid located in the first row and the second
column, and has type \textit{soil}, the spatial compartment $( \texttt{2.3}  \into 3*b \conc c)^{\textit{water}}$ represents a \textit{water}-type spatial
compartment in position \texttt{2.3}. In our grid we assume that molecules can float only through neighbour cells: all the rules of interaction between
spatial compartments must obviously contain the indexes of their location. For example, the rule $\TOP : ( \texttt{1.2} \conc x \into a\conc
X)^{\textit{water}} ( \texttt{2.2} \conc y \into Y )^{\textit{soil}} \srewrites{k} ( \texttt{1.2} \conc x \into X)^{\textit{water}}  ( \texttt{2.2}\conc y
\into a\conc Y)^{\textit{soil}}$ moves the molecule $a$ from the \textit{water} compartment in position \texttt{1.2} to the \textit{soil} compartment in
position \texttt{2.2} with a rate $k$ representing in this case, the speed of the movement of $a$ in downwards direction from a cell of
\textit{water}-type to a cell of \textit{soil}-type.






Let \texttt{R} and \texttt{C} denote the dimensions of our
\texttt{R} $\times$ \texttt{C} grid defined by \texttt{R} rows and
\texttt{C} columns. To increase the expressivity of the language we define a few structures to denote portions (i.e. sets of cells) of
the grid. With $\Theta$ we denote a  set of coordinates of
the grid and we use the notion \texttt{r.c} $\in \Theta$
when the coordinate \texttt{r.c} is contained in the set
$\Theta$.  We define rectangles by
\texttt{rect}[\texttt{r.c},\texttt{r'.c'}] where
\texttt{r.c},\texttt{r'.c'} represent the edges of the rectangle. We
project rows and columns of our grid with the constructions
\texttt{row}[$i$] and \texttt{col}[$j$] respectively.

\begin{example}
The set $\Theta=\{\mathtt{6.6}\}\cup \mathtt{rect}[\mathtt{1.1},\mathtt{3.2}]
\cup \mathtt{col}[5]$ represents the set of coordinates
$$\Theta = \{ \mathtt{6.6}\} \cup \{ \mathtt{1.1}, \mathtt{2.1}, \mathtt{3.1}, \mathtt{1.2}, \mathtt{2.2} , \mathtt{3.2}\} \cup \{ \mathtt{i.5} \;| \; \forall i \in [1,\mathtt{R}]\}.$$
\end{example}

Note that \texttt{row}[$i$] is just a shorthand for
\texttt{rect}[\texttt{i.1},\texttt{i.\texttt{C}}]. Similarly for
columns.\\
  We use \texttt{[*]} as shorthand to indicate the whole grid
(i.e. \texttt{rect[1.1,R.C]}).

We also define four \emph{direction} operators,
\texttt{N,~W,~S,~E} that applied to a range of cells shift them,
respectively, up, left, down and right.
 For instance \texttt{E(1.1)} = \texttt{1.2}. In the intuitive way, we also define the four diagonal movements (namely, \texttt{NW},~\texttt{SW},~\texttt{NE},~\texttt{SE}). With $\Delta$ we denote a set of directions and we use the special symbol $\diamond$ to denote the set containing all eight possible directions.

We convene that when a coordinate, for effect of a shit, goes out of
the range of the grid the corresponding point is eliminated from the
set.

\subsection{Surface Terms}

We define the initial state of the system under analysis as a set of
compartments modelling the two-dimensional grid containing the
biological entities of interest.

Let $\Theta$ denote a set of coordinates and $\ell_s$ a spatial label. We use the notation:
    $$\Theta, \ell_s  \boxplus \; \ov{t} $$
to define a set of cells of the grid. Namely $\Theta, \ell_s  \boxplus \; \ov{t}$ denotes the top level CWC term:
 $$\wr{\emptyseq}{   \wr{ \mathtt{r_1.c_1}}{ \ov{t}   }{\ell_s} \conc \ldots \conc \wr{ \mathtt{r_n.c_n}}{ \ov{t}  }{\ell_s}   }{\top}  $$
 where  $\mathtt{r_i.c_j}$ range over all elements of $\Theta$.

A spatial CWC term is thus defined by the set of grid cells covering
the entire grid.

\begin{figure}
\centering
\subfigure[Initial state] {
\includegraphics[width=.45\textwidth]{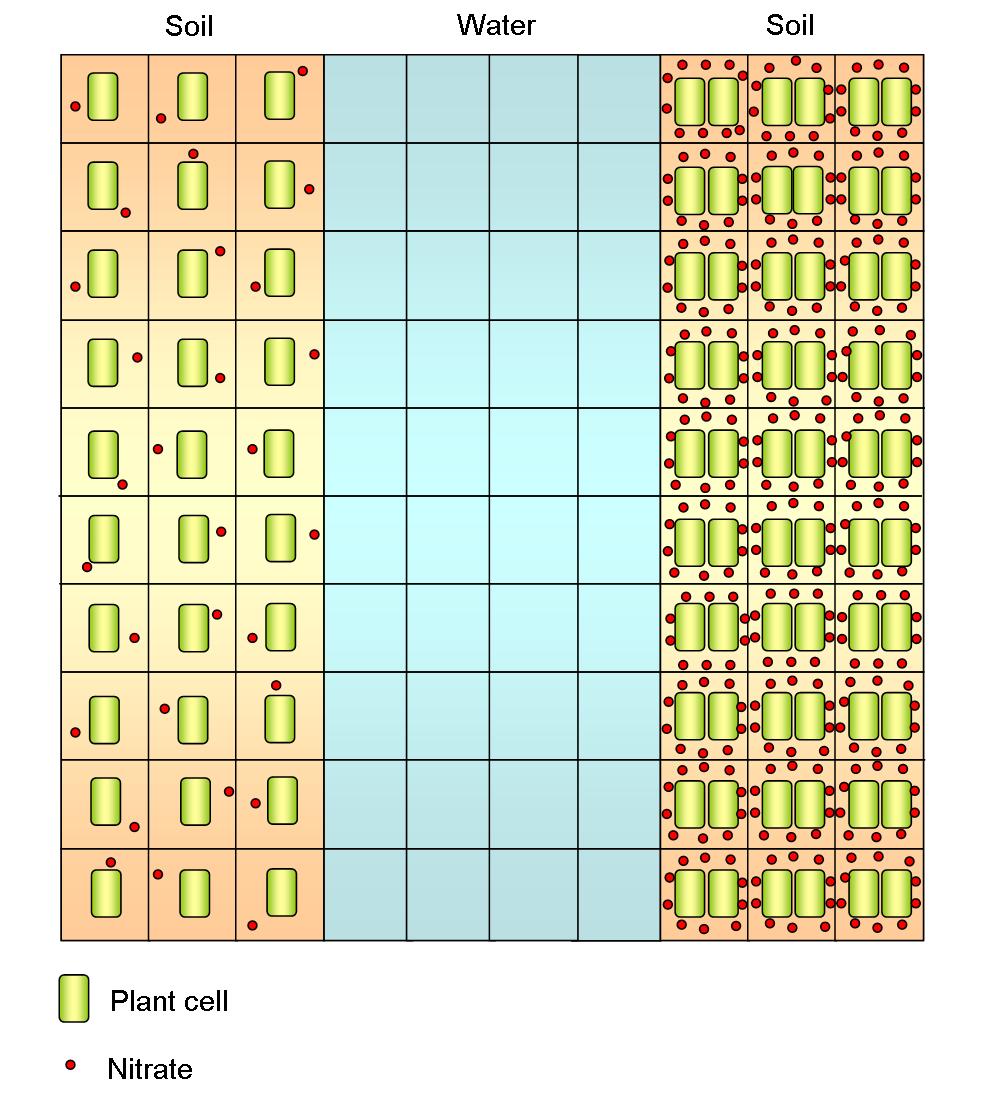}
\label{ex_grid_1}
}
\centering
\subfigure[Spatial events] {
\includegraphics[width=.45\textwidth]{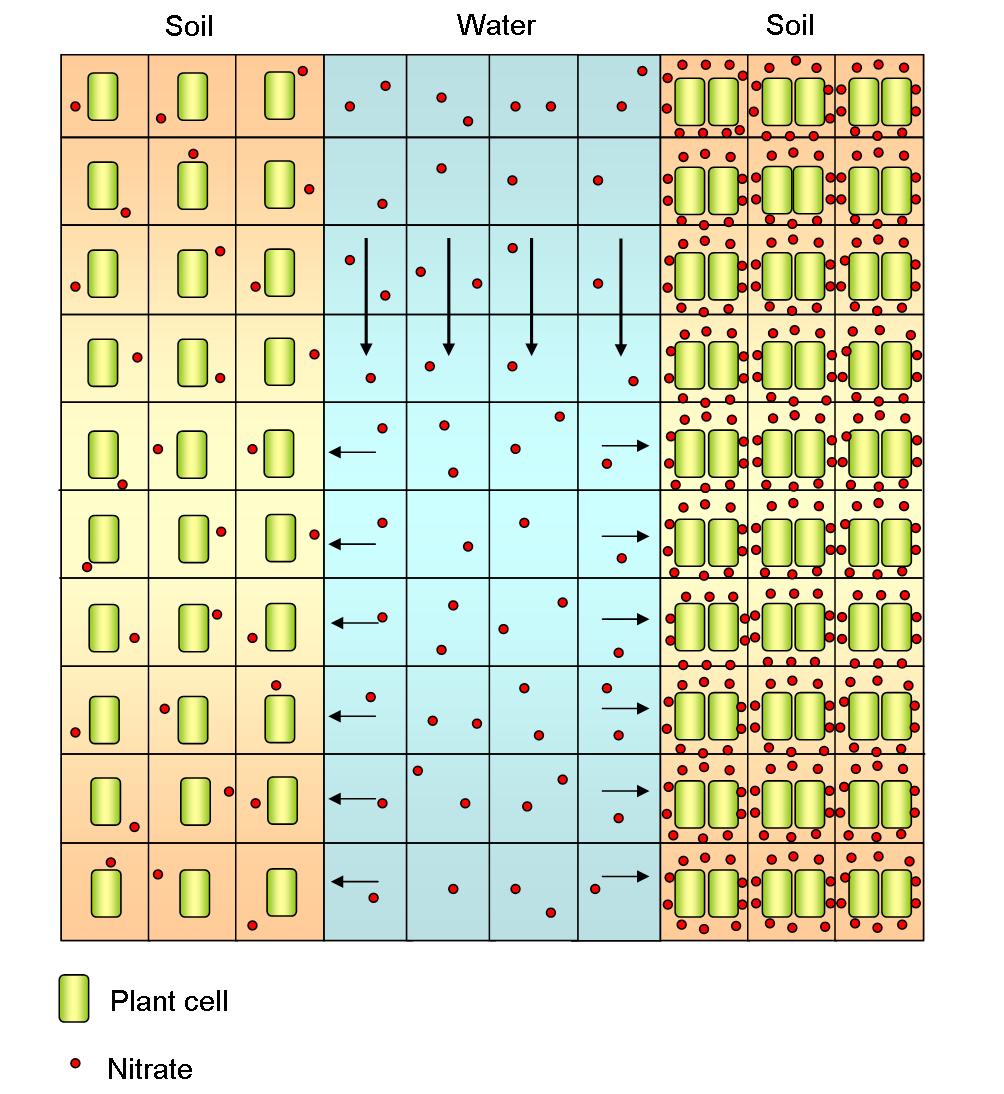}
\label{ex_grid_2}
}
\caption{Graphical representations of the grid described in the Example $3.2$.}\label{ex_grid}
\end{figure}

\begin{example}\label{EX_grid}
The CWC term obtained by the three grid cell constructions:\\
$
\mathtt{rect[1.1,10.3]}, soil \boxplus \; nitr \conc (receptors \into cytoplasm\conc nucleus)^{PlantCell}$\\
$\mathtt{rect[1.4,10.7]}, water \boxplus \; \emptyseq$ \\
$\mathtt{rect[1.8,10.10]}, soil \boxplus \; 10*nitr \conc 2 * (receptors \into cytoplasm \conc nucleus)^{PlantCell}$\\
\noindent builds a $10\times 10$ grid composed by two portions of soil (the right-most one reacher of nitrates and plant cells) divided by a river of water (see Figure \ref{ex_grid_1}).
\end{example}






\subsection{Surface Rewrite Rules}

We consider rules for modelling three kind of events.

\smallskip
\noindent
\emph{Non-Spatial Events:} are described by standard CWC rules, i.e. by rules of the shape:
 $$\ell:  \ov\LeftPat \srewrites{k} \ov\RightPat $$
\noindent Non-spatial rules can be applied to any compartment of type $\ell$ occurring in any
portion of the grid and do not depend on a particular location.

\begin{example}
A plant cell might perform its activity in any location of the grid. The following rules, describing some usual activities within a cell, might happen in any spatial compartment containing the plant cells under considerations:\\
$PlantCell: nucleus  \srewrites{k_1} nucleus \conc  mRNA$\\
$PlantCell: mRNA \conc cytoplasm  \srewrites{k_2} mRNA \conc  cytoplasm \conc protein$.
\end{example}

\smallskip \noindent
\emph{Spatial Events:} are described by rules that can be applied to specific spatial compartments. These rules allow to change the spatial label of the considered compartment. Spatial events are described by rules of the following shape:
$$
\Theta \triangleright \ell_s : \overline{\LeftPat} \srewrites{k}  \ell_s': \ov\RightPat
$$
\noindent Spatial rules can be applied only
within the spatial compartments with coordinates contained in the set $\Theta$ and with the spatial label $\ell_s$. The application of the rule may also change the label of the spatial compartments $\ell_s$ to $\ell_s'$.
This rule is translated into the CWC set of rules:\\
$\TOP : (  \mathtt{r_i.c_i} \conc x \into
\overline{\LeftPat}\conc X)^{\ell_s} \srewrites{k} (  \mathtt{r_i.c_i}\conc x
\into \ov{\RightPat}\conc X )^{\ell_s'} \quad \forall \mathtt{r_i.c_i} \in
\Theta.$

Note that spatial rules are analogous to non spatial ones. The only difference is the explicit indication of the set $\Theta$ which allows to write a single rule instead of a set of rule (one for each element of $\Theta$).

\begin{example}
If we suppose that the river of water in the middle of the grid defined in Example~\ref{EX_grid} has a downward streaming, we might consider the initial part of the river (framed by the first row \emph{\texttt{rect[1.4,1.7]}})
 to be a source of nitrates (as they are coming from a region which is not modelled in the actual considered grid). The spatial rule:\\
$
\mathtt{rect[1.4,1.7]}  \triangleright water: \emptyseq \srewrites{k_3} water: nitr
$\\
\noindent models the arrival of nitrates at the first modeled portion of the river (in this case it does not change the label of the spatial compartment involved by the rule).
\end{example}
 %

 \noindent
\emph{Spatial Movement Events:} are described by rules considering the content of two adjacent spatial compartments and are described by rules of the following shape:
$$
\Theta \triangleleft \Delta \triangleright \ell_{s_1}, \ell_{s_2}: \ov{p_1} ,\ov{p_2} \srewrites{k} \ell_{s_1}', \ell_{s_2}': \ov{o_1} ,\ov{o_2}
$$
\noindent This rule changes the
content of two adjacent (according to the possible directions contained in $\Delta$) spatial
compartments and thus allows to define the movement of objects. The
pattern matching is performed by checking the content of a spatial
compartment of type $\ell_{s_1}$ located in a portion of the
grid defined by $\Theta$ and the content of the adjacent spatial
compartment of type $\ell_{s_2}$. Such a rule could also change the
labels of the spatial compartments.
This rule is translated into the CWC set of rules:
$$\TOP : ( \mathtt{r_i.c_i}\conc x \into
\overline{\LeftPat_1} \conc X)^ {\ell_{s_1}}  ( dir(\mathtt{r_i.c_i}) \conc y \into
 \overline{\LeftPat_2} \conc Y) ^{\ell_{s_2}}
\srewrites{k} \\ ( \mathtt{r_i.c_i} \conc x \into  \ov{o_1}\conc X)
^{\ell_{s_1}'}  ( dir(\mathtt{r_i.c_i}) \conc y \into   \ov{o_2}\conc Y )^{\ell_{s_2}'}$$
for all $\mathtt{r_i.c_i} \in \Theta$ and for all $dir \in \Delta$.

\begin{example}
We assume that the flux of the river moves the nitrates in the water according to a downward direction in our grid and with a constant speed in any portion of the river with the following rule:\\
$
\mathtt{rect[1.4,9.7]} \triangleleft \{\mathtt{S}\} \triangleright water, water: nitr,\emptyseq \srewrites{k_4} water, water: \emptyseq, nitr
$\\
when nitrates reach the down-most row in our grid they just disappear (non moving event):\\
$
\mathtt{rect[10.4,10.7]} \triangleright water: nitr \srewrites{k_4} water: \emptyseq
$.\\
Moreover, nitrates streaming in the river may be absorbed by the soil on the riverside with the rule:\\
$
\mathtt{rect[1.4,10.7]} \triangleleft \{\mathtt{W}, \mathtt{E}\} \triangleright water, soil: nitr,\emptyseq \srewrites{k_5} water, soil: \emptyseq, nitr
$.\\
A graphical representation of these events is shown in Figure \ref{ex_grid_2}. Other rules can be defined to move the nitrates within the soil etc.
\end{example}

\subsection{An AM Symbiosis Model with the Surface Language}

Simultaneously to intraradical colonization, the fungus develops an extensive network of hyphae which explores and exploits soil microhabitats for nutrient acquisition. AM fungi have different hyphal growth patterns, anastomosis and branching frequencies which result in the occupation of different niche in the soil and probably reflect a functional diversity \cite{maherali2007influence} (see Figure \ref{fig:ERM}).

\begin{figure}
\center
\includegraphics[width=.8\textwidth]{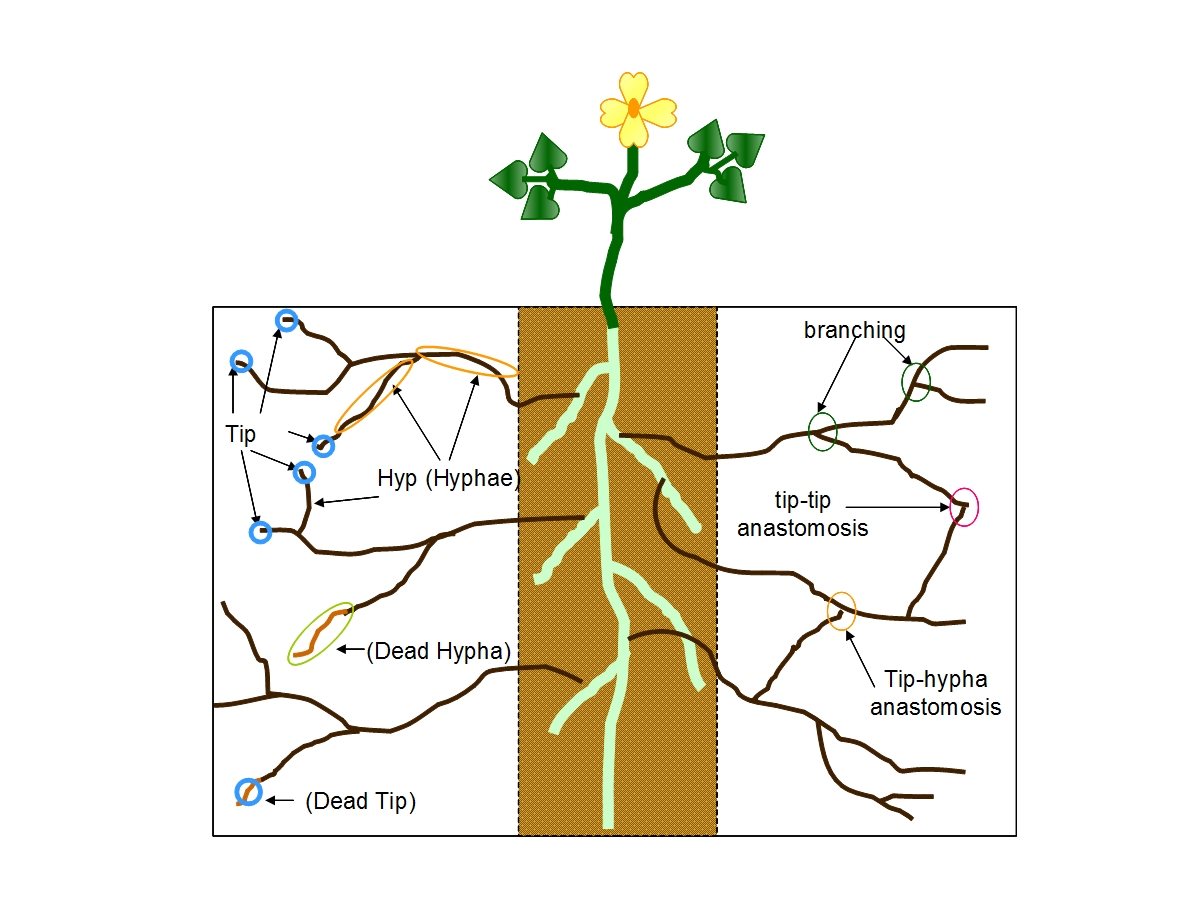}
\caption{Extraradical mycelia of an arbuscular mycorrhizal fungus.}\label{fig:ERM}
\end{figure}

The mycelial network that develops outside the roots is considered as the most functionally diverse component of this symbiosis. Extraradical mycelia (ERM) not only provide extensive pathways for nutrient fluxes through the soil, but also have strong influences upon biogeochemical cycling and agro-ecosystem functioning \cite{purin2008parasitism}. The mechanisms by which fungal networks extend and function remain poorly characterized. The functioning of ERM presumably relies on the existence of a complex regulation of fungal gene expression with regard to nutrient sensing and acquisition.
The fungal life cycle is then completed by the formation, from the external mycelium, of a new generation of spores able to survive under unfavourable conditions.

Investigations on carbon (C) metabolism in AM fungi have proved useful to offer some explanation for their obligate biotrophism. As mentioned above, an AM fungus relies almost entirely on the host plant for its carbon supply. Intraradical fungal structures (presumably the arbuscules) are known to take up photosynthetically fixed plant C as hexoses. Unfortunately, no fungal hexose transporter-coding gene has been characterized yet in AM fungi.

In order to quantify the contribution of arbuscular mycorrhizal (AM) fungi to plant nutrition, the development and extent of the external fungal mycelium and its nutrient uptake capacity are of particular importance. Shnepf and collegues \cite{schnepf2008growth} developed and analysed a model of the extraradical growth of AM fungi associated with plant roots considering the growth of fungal hyphae from a
cylindrical root in radial polar coordinates.

Measurements of fungal growth can only be made in the presence of plant. Due to this practical difficulty experimental data for calibrating the spatial and temporal explicit models are scarce.
Jakobsen and collegues \cite{jakobsen1992external} presented measurements of hyphal length densities of three AM fungi: \textit{Scutellospora calospora} (Nicol.\& Gerd.) Walker \& Sanders; \textit{Glomus} sp. associated with clover (\textit{Trifolium subterraneum} L.); these data appeared suitable for comparison with modelled hyphal length densities.

The model in  \cite{schnepf2008growth} describes, by means of a system of Partial Differential Equations (PDE), the development and distribution of the fungal mycelium in soil in terms of the creation and death of hyphae, tip-tip and tip-hypha anastomosis, and the nature of the root-fungus interface. It is calibrated and corroborated using published experimental data for hyphal length densities at different distances away from root surfaces. A good agreement between measured and simulated values was found for the three fungal species with different morphologies associated with \textit{Trifolium subterraneum} L. The model and findings are expected to
contribute to the quantification of the role of AM fungi in plant mineral nutrition and the
interpretation of different foraging strategies among fungal species.

\subsubsection{Surface CWC Model}

In this Section we describe how to model the growth of arbuscular mycorrhyzal fungi using the surface spatial CWC.
We model the  growth of AM fungal hyphae in a soil environment partitioned into $13$ different layers (spatial compartments with label $soil$) to account for the distance in centimetres between the plant root and the fungal hyphae where the soil layer at the interface with the plant root is at position $1.1$. We describe the mycelium by two atoms:
the hyphae (atom $Hyp$) related to the length densities (number of hyphae in a given compartment) and the hyphal tips (atom $Tip$). The plant root (atom $Root$) is contained in the $soil$ compartment at position $1.1$.

The tips and hyphae at the root-fungus interface proliferate according to the following spatial events:

$$\{ \mathtt{1.1}\} \triangleright soil : Root \srewrites{\tilde{a}} soil : Root \conc Hyp $$
$$\{ \mathtt{1.1}\} \triangleright soil : Root \srewrites{a} soil : Root \conc Tip $$

\noindent where $\tilde{a}$ and $a$ is the root proliferation factor for the hyphae and tips respectively.

Hyphal tips are important, because growth occurs due to the elongation of the region just behind the tips. Therefore, the spatial movement event describing the hyphal segment created during a tip shift to a nearby compartment is:\\

$$
[*] \triangleleft  \{\mathtt{E}, \mathtt{W}\} \triangleright soil, soil: Tip, \emptyseq \srewrites{v} soil, soil: Hyp , Tip
$$

\noindent where $v$ is the rate of tip movement. The hyphal length is related to tips movement, i.e. an hyphal trail is left behind as tips move through the compartments. We consider hyphal death to be linearly
proportional to the hyphal density, so that
the rule describing this spatial event is:\\

$$ [*] \triangleright soil : Hyp \srewrites{d_H} soil : \emptyseq $$

\noindent where $d_H$ is the rate of hyphal death.

Mycorrhizal fungi are known to branch mainly apically where one tip splits into two. In the simplest case, branching and
tip death are linearly proportional to the existing tips in that location modelled with the following spatial events:

$$ [*] \triangleright soil : Tip \srewrites{b_T} soil : 2*Tip $$
$$ [*] \triangleright soil : Tip \srewrites{d_T} soil : \emptyseq $$

\noindent where $b_T$ is the tip branching rate and $d_T$ is the tip death rate.

Alternatively, if we assume that branching decreases with increasing tip density and ceases at a given maximal tip density, we employ the spatial event:

$$ [*] \triangleright soil : 2* Tip \srewrites{c_T} soil : \emptyseq $$

\noindent where $c_T=\frac{b_T}{T_{max}}$. From a
biological point of view, this behaviour take into account the volume saturation when the tip density achieves the maximal number of tips $T_{max}$.

The fusion of two hyphal tips or a tip with a hypha
can create interconnected networks by means of
anastomosis:

$$ [*] \triangleright soil : 2* Tip  \srewrites{a_1} soil : Tip $$
$$ [*] \triangleright soil : Tip \conc Hyp \srewrites{a_2} soil : Tip $$

\noindent where $a_1$ and $a_2$ are the tip-tip and tip-hypha anastomosis rate constants, respectively.

The initial state of the system is given by the following grid cell definition:

$$\{ \mathtt{1.1}\},\;soil \boxplus  Root \conc T_0 * Tip \conc  H_0 * Hyp$$
$$\mathtt{rect}[\mathtt{1.2},\mathtt{1.13}],\;soil \boxplus  \emptyseq$$

\noindent where  $T_0$ and $H_0$ are the initial number of tips and hyphae respectively at the interface with the plant root.

\subsubsection{Results}

We run $60$ simulations on the model for the fungal species \textit{Scutellospora
calospora} and \textit{Glomus} sp. Figure \ref{res_fungi} show the mean values of hyphae (atoms $Hyp$) of the resulting stochastic simulations in function of the elapsed time in days and of the distance from the root surface. The rate parameters of the model are taken from~\cite{schnepf2008growth}.

The results for \textit{S. calospora} are in accordance with the linear PDE model of~\cite{schnepf2008growth}  which is characterized by linear branching with a relatively small net branching rate and both kinds of anastomosis are negligible when compared with the other species. This model imply that the fungus is mainly growing and allocating resources for
getting a wider catchment area rather than local expoloitation of mineral resources via hyphal branching.

The model for \textit{Glomus} sp. considers the effect of nonlinear
branching due to the competition between tips for space. The results obtained for \textit{Glomus} sp. are in accordance with the non--linear PDE model of~\cite{schnepf2008growth} which imply that local exploitation for resources via hyphal branching is important for
this fungus as long as the hyphal tip density is small. Reaching near the maximum tip density, branching decreases. 
Symbioses between a given host plant and different AM fungi have been shown to differ functionally~\cite{ravnskov1995functional}. 

\begin{figure}
\center
\includegraphics[angle=-90,width=.49\textwidth]{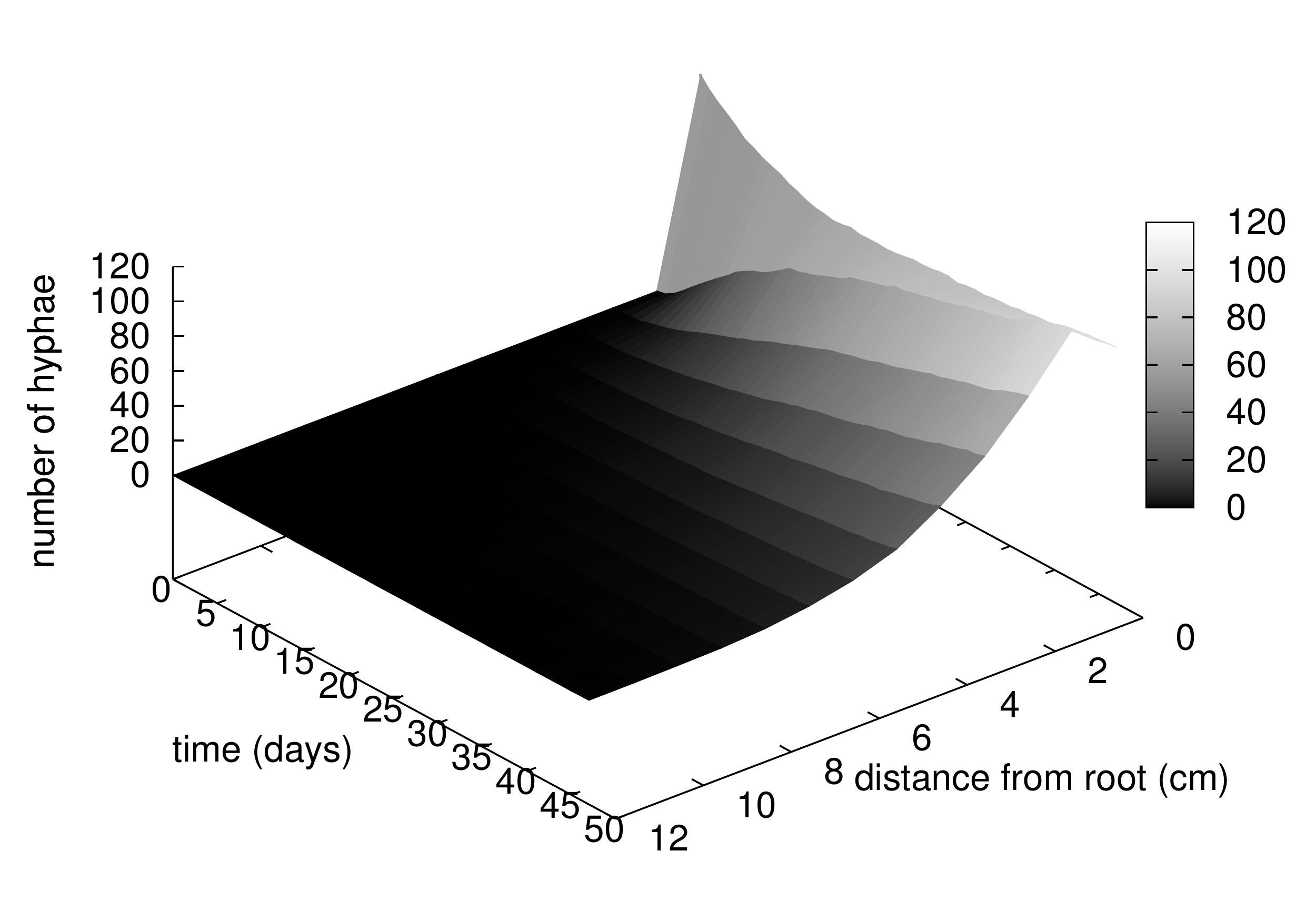}
\includegraphics[angle=-90,width=.49\textwidth]{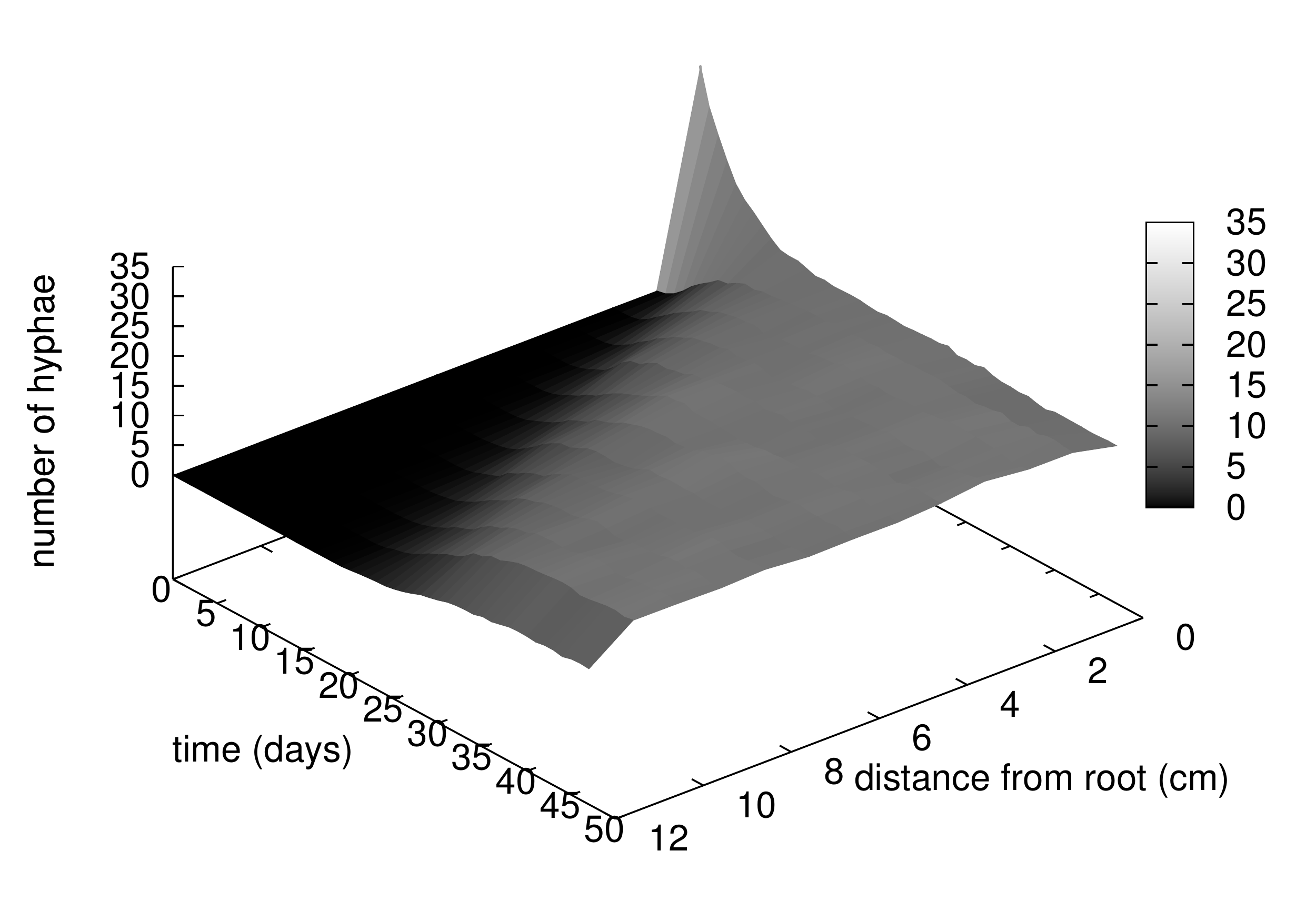}
\caption{Mean values of $60$ stochastic simulations of hyphal growth (atoms $Hyp$) results for S. calospora and Glomus sp. fungi.}\label{res_fungi}
\end{figure}

\chapter{A Hybrid Semantics}
\label{hybrid}
Metabolic networks involving large numbers of molecules are most often modelled deterministically through Ordinary Differential Equations (ODEs) but it may happen that a purely deterministic model does not accurately capture the dynamics of multi-stable system, and a stochastic description is needed.
However, while the ODEs method is extremely efficient, the computational cost of a discrete stochastic simulation often becomes an insurmountable obstacle. Thus, when the deterministic approach is applicable and provides a good approximation of the system behaviour, it might be profitable to take advantage of its efficiency, and move to the stochastic approach when this is not true any more.

In a hybrid model, some reactions are modelled in a discrete  way (i.e. computed probabilistically according
to an exact stochastic method) and others in a continuous way (i.e. computed in a deterministic way by a set of ODEs). These models for the simulation of biological systems have been presented in the last few years for purely mathematical models~\cite{SK05,GCPS06,CDR09}, i.e. models in which all reactions take place in a single, ``flat'', environment (without compartmentalisation). 

In this chapter we present an hybrid simulation algorithm for the CWC framework. We will apply our approach to a variant of Lotka-Volterra dynamics and an HIV-1 transactivation mechanism.

The CWC simulator has been enriched with a prototype implementation of the hybrid simulation algorithm~\cite{HCWC_SIM}. Hybrid simulations produce results consistent with those obtained by the exact stochastic simulation method, but with a considerable gain in computing time.

\section{Comparing Quantitative Methods}\label{SECT:running}

In this section we compare the two quantitative simulation methods based, respectively, on stochastic computations and on the deterministic solution of ODEs. We consider, as a running example, a toy case study derived from a Lotka-Volterra prey-predator dynamics. 
We let the prey-predator oscillatory dynamics to be confined into a compartment \emph{IN} interfered with rare events causing dramatic changes in the species evolution. 
A rare event could be represented as a viral epidemic factor entering and
exiting compartment \emph{IN} with a relatively slow rate. Once inside the compartment  \emph{IN} the
viral epidemic factor has the capability of killing some preys.

The set of \CalculusShortName\ rules modelling such an example is given in Figure~\ref{fig:LKvir-rules}. The preys (atoms $a$) and predators (atoms $b$) are located in compartment \emph{IN} and follow a dynamics given by the rules ($B_1$),($B_2$) and ($B_3$). The viral epidemic (atom \emph{Vir}) enters and leaves the compartment with rules ($N_1$) and ($N_2$) respectively, and kills the preys with rule ($B_4$).

\begin{figure}[h]
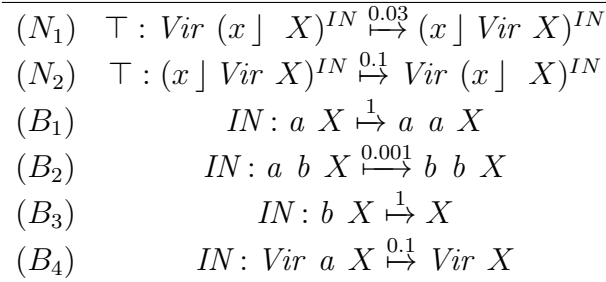

\hrule
$
\begin{array}{lc}
(N_1) & \TOP : \textit{Vir} \conc (x \into \conc X)^{IN} \xsrewrites{0.03} (x \into \textit{Vir} \conc X)^{IN}\\
(N_2) & \TOP : (x \into \textit{Vir} \conc X)^{IN} \xsrewrites{0.1} \textit{Vir} \conc (x \into \conc X)^{IN}\\
(B_1) & \textit{IN} : \textit{a} \conc X \xsrewrites{1} \textit{a} \conc \textit{a} \conc X\\
(B_2) & \textit{IN} : \textit{a} \conc \textit{b} \conc X \xsrewrites{0.001} \textit{b} \conc \textit{b} \conc X\\
(B_3) & \textit{IN} : \textit{b} \conc X \xsrewrites{1} X\\
(B_4) & \textit{IN} : \textit{Vir} \conc \textit{a} \conc X \xsrewrites{0.1} \textit{Vir} \conc X\\
\end{array}
$
 \caption{ \CalculusShortName\ rules for the prey-predator dynamics}
\label{fig:LKvir-rules}
\end{figure}

\subsection{Stochastic Simulation}

Simulations are performed for $60$ time units, with the starting term:
$$
\textit{Vir} \conc (\emptyseq \into 1200 \timesSilent \textit{a} \conc 1200 \timesSilent \textit{b})^\textit{IN}.
$$
Several stochastic simulations of the toy case study were performed, showing different possible system evolutions of the dynamics of the species inside the compartment \textit{IN} depending on the viral epidemic factor. Two of these runs are shown in Figure~\ref{fig:stoch_toy}.
\begin{figure*}[t]
\centering
  \includegraphics[width=.49\textwidth]{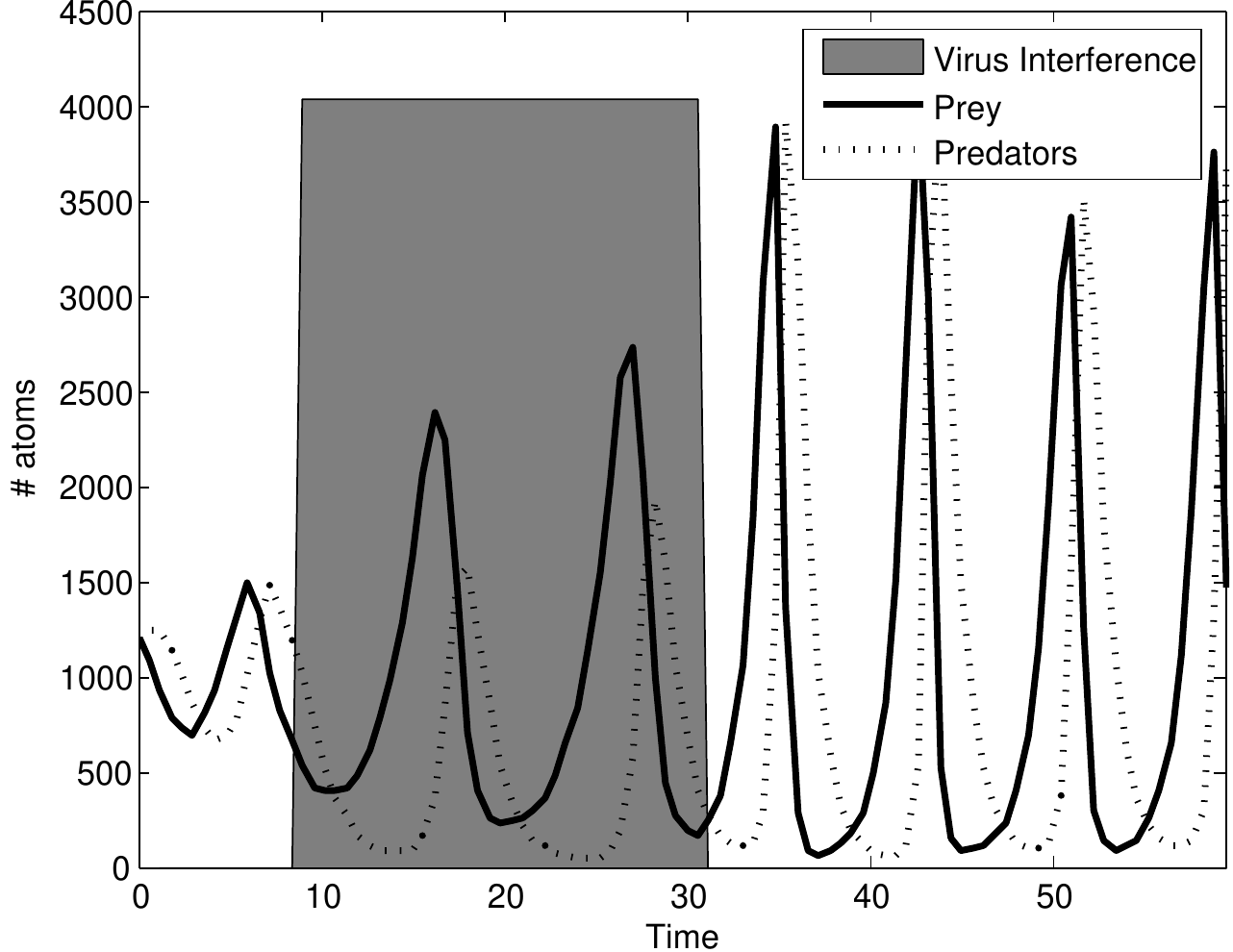}
 \includegraphics[width=.49\textwidth]{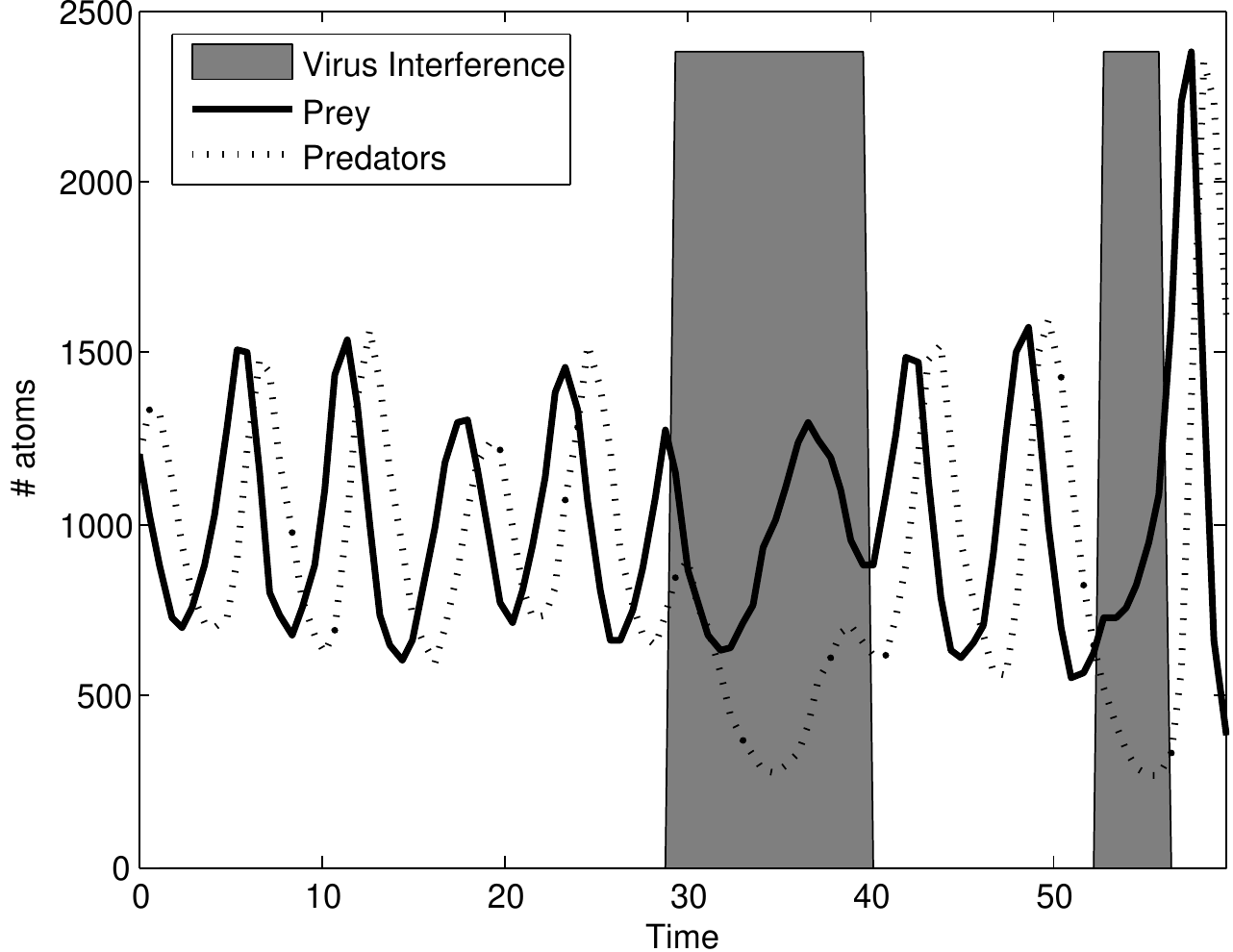}
  \caption{Two different runs obtained with a purely stochastic simulation showing the different behaviour of the dynamics of the species
  inside the compartment \textit{IN} depending on the viral epidemic factor}
\label{fig:stoch_toy}
\end{figure*}
A characteristic of this example is that the evolution of the system is
strongly determined by the viral epidemic random event that can change dramatically the dynamics of the species.

\subsection{Deterministic Simulation}\label{SECT:ODE_SEM}

A \short\ system $\QQ$ consisting of $r$ biochemical rewrite rules represents
 a system of $r$ biochemical reactions. Its deterministic semantics is defined by
extracting from $\QQ$ a system of ODEs  to be used for simulating the evolution of the involved multisets of atoms~\cite{HJ72}.
%
%
For every label $\ell$, let
\begin{itemize}
\item
$a_1,\ldots,a_{n_\ell}$ denote the $n_\ell$ species of atoms that may occur at top level within a compartment of type $\ell$, and
\item
$\QQ_\ell$ denote the set of rules with label $\ell$.
\end{itemize}
The $i$-th rule in the set $\QQ_\ell$ is denoted by
$$
  \ell :  \bar{a}_i \csrew{k_i}  \bar{b}_i \ \ \  i= 1, 2, \ldots, \LengthOf{\QQ_\ell}
$$
 For all species $a_j$ ($j=1, 2, \ldots, n_\ell$) let $\alpha^-_{i,j}$ be the number of atoms of species $a_j$ consumed by the $i$-th rule
 and $\alpha^+_{i,j}$ the number of atoms of species $a_j$ produced by the $i$-th rule.
 The $n_\ell \times \LengthOf{\QQ_\ell}$ stoichiometric matrix $\Lambda_\ell$ is defined
 by $\nu_{i,j}=\alpha^+_{i,j}-\alpha^-_{i,j}$.\footnote{Many of the $\alpha^-_{i,j},\;\alpha^+_{i,j}$ are usually $0$.}

Let $[a]$ denote the number of the atoms of species $a$ occurring at top level in a given compartment of type $\ell$. If
$\bar{a}_i= n_{i_1} a_{i_1} \ldots n_{i_{r_i}} a_{i_{r_i}} (r_i\geq 1)$,
%
the evolution of the given compartment of type $\ell$ is modelled by the following system of ODEs:
$$
  \ell : \frac{d[a_j]}{dt} = \sum_{i=1}^{\LengthOf{\QQ_\ell}} \nu_{i,j} \cdot k_i \cdot [a_{i_1}]^{n_{i_1}}\cdot \ldots\cdot [a_{i_{r_i}}]^{n_{i_{r_i}}}
$$

%
%
%
    %

Computationally, ODEs are well studied and understood. They can be solved
using a variety of numerical methods, from the Euler method to
higher-order Runge-Kutta methods or stiff methods, many of which are readily
available in software packages that can be easily incorporated
into existing simulation code. In all the examples presented in this chapter we use a GNU library implementing an explicit embedded Runge-Kutta Prince-Dormand method.

To perform a deterministic simulation of the toy case study we have to remodel the Lotka-Volterra dynamics by \emph{eliminating} the stochastic influence of the virus epidemics. The equation system governing the population evolution is:
\begin{eqnarray*}
\frac{da}{dt} = a - 10^{-3} a\cdot b \\
\frac{db}{dt} = -b + 10^{-3} a\cdot b
\end{eqnarray*}
with an initial condition of $a(0)=b(0)=1200$.

The evolution of the system is represented in Figure~\ref{fig:lk_det}. Notice that the mean of $100$ stochastic simulations of the Lotka-Volterra dynamics without the viral epidemics interference tends to the deterministic simulation.

\begin{figure*}
\centering
  \includegraphics[width=.49\textwidth]{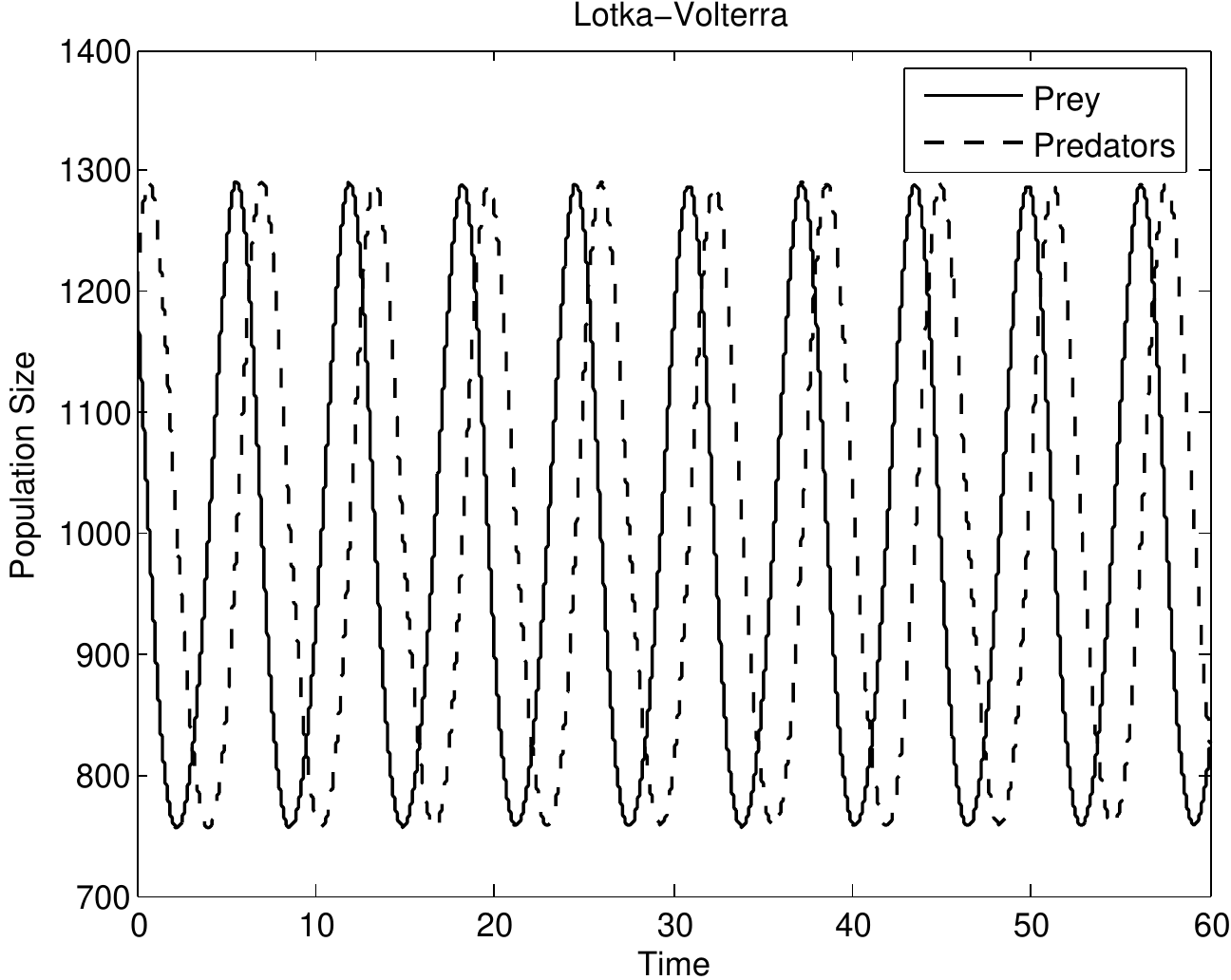}
  \includegraphics[width=.49\textwidth]{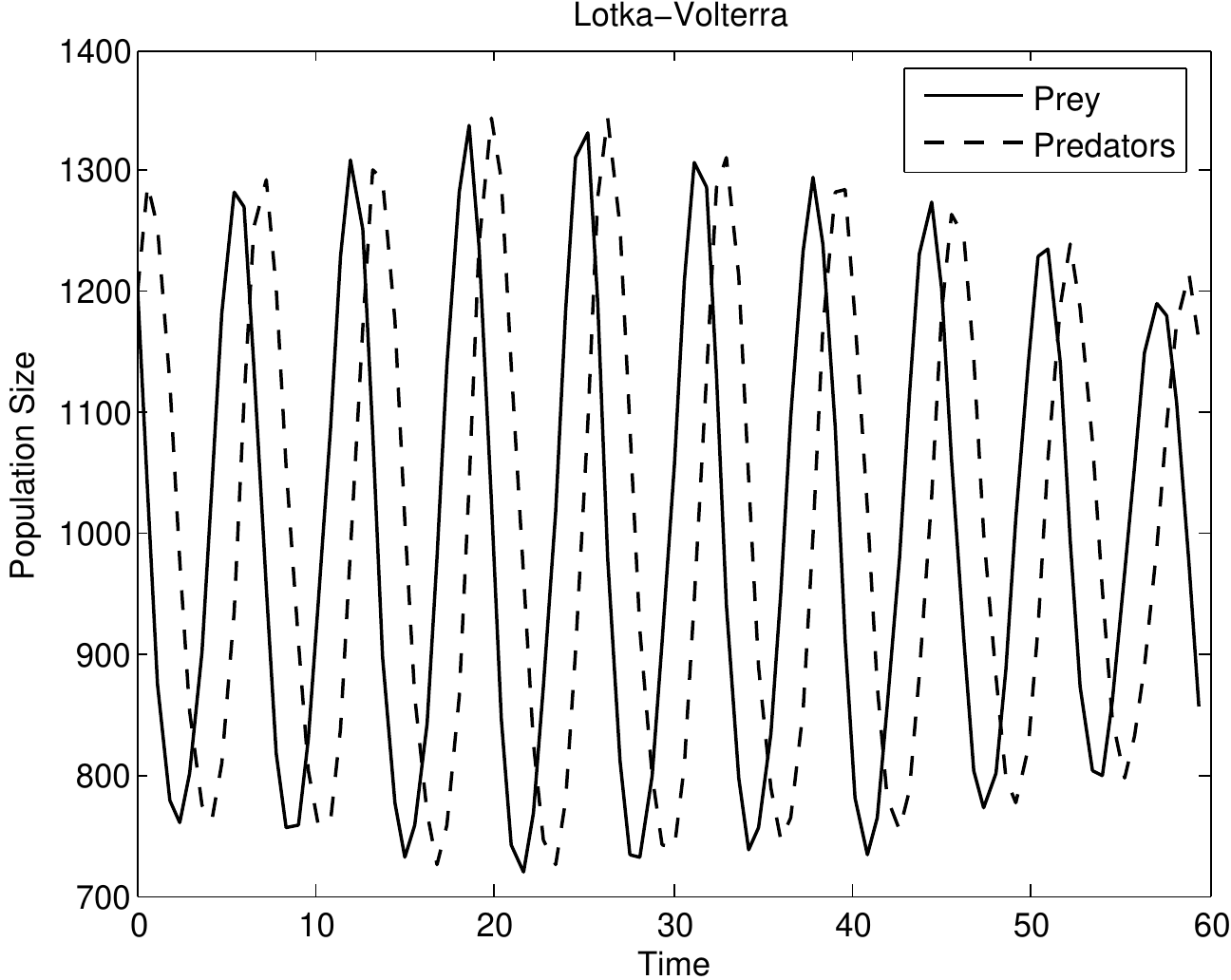}
  \caption{Deterministic simulation of the Lotka-Volterra dynamics (left figure) and the mean of $100$ stochastic simulations (right figure)}
\label{fig:lk_det}
\end{figure*}

\section{Hybrid Evolution}\label{sect:Hybrid}

The stochastic approach is based on a probabilistic simulation method that
 manages the evolution of exact integer quantities and often requires a huge
computational time to complete a simulation. The ODEs numerical approach computes a unique deterministic and
 fractional evolution of
the species involved in the system and achieves very efficient computations.
%

A numerical background for the hybrid evolution in
which fast reactions are approximated using a deterministic
evolution while slow reactions are treated as stochastic events
can be found in~\cite{haseltine2005origins, haseltine2002approximate}. This schema allows to accurately solve
fast reactions using an ODE solver at the thermodynamic limit. This condition is of course ideal and unattainable
in biological systems. However, the analytical knowledge of the system allows the use of this approximation
if the variation of the species affected by slow reactions are almost insensitive with respect to species affected by fast reactions.

In this section we adapt this approximation method within \short, defining a hybrid simulation technique.
%
%
Given a \short\ system $\QQ$ we partition it into a set of biochemical rewrite rules $\BB$ and a set of non-biochemical rewrite rules $\NN$.  Rules in $\NN$
are always applied by using the stochastic method. Rules in $\BB$ might be applied with the ODEs approach. In general $\BB$ might contain both rules that
model evolution of large numbers of molecules according to very fast reactions (whose execution is suitable to be correctly computed with ODEs) and rules that model very slow reactions or reactions that involve a very
small number of reagents. In the latter case
it is convenient to compute the execution of
the associated rule according to the stochastic approach.

According to the state of the system, a rule might be dynamically interpreted either as stochastic or deterministic. For instance, during a simulation, it might happen that a given biochemical rewrite rule $\ell:
\bar{a}_i \srewrites{k_i} \bar{b}_i \in \BB$ is applied initially according to the stochastic semantics, since the associated compartment contains a very
small number of reagents. After the system has evolved for some time, however, the number of the reagents involved in the rule can be substantially
increased and it becomes convenient to model the corresponding reaction according to the deterministic approach.

 Actually, at the beginning of each simulation
step we build, for each compartment in the term, a system of ODEs for the simulation of the biochemical rules in that compartment which (1) are
sufficiently fast and (2) involve a sufficient number of reagents. For the remaining rules the evolution is determined by the stochastic simulation
algorithm.

In order to describe the hybrid semantics we assume that, given a \short\ term $\ov{t}$, each compartment of $\ov{t}$ is univocally identified by an
index $\iota$. The index of the (implicit) compartment at the top level will be denoted by $\iota_0$. The \emph{biochemical reagents} of a compartment $(\ov{a}\into \ov{t})^\ell$ with index
$\iota$, written $\BR(\iota)$, are expressed by the multiset of the atomic elements appearing in the top level of $\ov{t}$. For example, given the term
\[
\ov{t} = 2a \conc (b \into (c \conc d \into \emptyseq)^{\ell'} \conc e )^\ell \conc (b \into d \conc e )^\ell
\]
and assuming that the compartment $(b \into (c \conc d \into \emptyseq)^{\ell'} \conc e )^\ell$ has index $\iota_1$, the compartment $ (c \conc d \into \emptyseq)^{\ell'}$ has index $\iota_2$ and the compartment $(b \into d \conc e )^\ell$ has index $\iota_3$, we have that $\BR(\iota_0)=2a$, $\BR(\iota_1)=e$, $\BR(\iota_2)=\emptyseq$ and $\BR(\iota_3)=d \conc e$ where $\iota_0$ is the index of the to level compartment.

\begin{figure}
\hrule $\;$ \\
Let $\ov{t}$ denote the whole term and let $I$ denote the set of compartment indexes occurring in $\ov{t}$.
\begin{enumerate}
\item For each compartment $\iota\in I$:   
\begin{itemize}
\item Let $\ell$ be the label of $\iota$, let $\RR_\iota=\BB_\ell$ and let $\DD_\iota=\emptyset$.
\item For each biochemical rule 
$B_i=\ell:\bar{a}_i \srewrites{k_i} \bar{b}_i \in \BB_\ell$ let $\bar{a}_i= n_{i_1} a_{i_1} \ldots n_{i_{r_i}} a_{i_{r_i}} (r_i\geq 1)$ and let
$[a_{i_j}]_\iota$ denote the number of $a_{i_j}$ atoms occurring in $\BR(\iota)$. Let $\pi_i^{\iota}=k_i \cdot ([a_{i_1}]_\iota^{n_{i_1}} / n_{i_1}) \cdot \ldots
\cdot ([a_{i_{r_i}}]_\iota^{n_{i_{r_i}}} / n_{i_{r_i}})$ be the rate of the rule $B_i$ in the compartment $\iota$. If $\pi_i^{\iota} > \phi$ and $\min_{j=1}^{r_i}{[a_{i_{j}}]_\iota} > \psi$ remove $B_i$ from $\RR_\iota$ and put it into $\DD_\iota$.
\end{itemize}
\item Considering the rules in $\bigcup_{\iota \in I} \RR_\iota \cup \NN$ selected according to Gillespie's method and to the semantics given in Section~\ref{SECT:running}
a stochastic transition step $C[\Pat\sigma] \ltrans{f(\sigma)} C[o\sigma]$, where $R=\ell:\Pat \srewrites{f}  o
\in \RR_{\iota'} \cup \NN_\ell$.
\\
Let $\tau$ be the corresponding time interval.$^\dagger$
\item
For each compartment $\iota$ in $I$:
\begin{itemize}
\item
Let $\Ode_{\iota}$ denote the system of ODEs for the rules in $\DD_\iota$ in the compartment $\iota$ as explained in Section~\ref{SECT:ODE_SEM} without
considering, in the compartment $\iota'$ where the stochastic transition step takes place, the active reagents appearing in the left part $\Pat$ of the
stochastically applied rule. (If $\DD_\iota=\emptyset$ then $\Ode_{\iota}=\emptyset$.)
\item Apply the system of ODEs $\Ode_{\iota}$ to the biochemical reagents $\BR(\iota)$
 of the compartment for a time duration $\tau$.
\end{itemize}
\item Update the term $\ov{t}$ according to the right part $o$ of the chosen stochastic rule and to the applications of the systems of ODEs.
\end{enumerate}
\underline{\hspace{4cm}}\\
 \footnotesize{$^\dagger$ It may happen that no rule in $(\bigcup_{\iota \in I} \RR_\iota)\cup\NN$ is applicable.
In such cases the evolution of the system must be determined for some time $\tau$ according to the deterministic semantics only. In our implementation we
choose as $\tau$ the maximum time calculated by Gillespie's algorithm for each of the applicable biochemical rules in $\bigcup_{\iota \in I}\DD_\iota$.}
 \caption{Steps performed by an hybrid simulation iteration}
\label{fig:CWM-hybrid-step}
\end{figure}

A basic point of our hybrid approach is the criterion to determine, at each computation stage, the reductions to compute in the stochastic or in the
deterministic way. In this paper we have chosen simply to put a threshold on the number of possible reagents and on the speed of the reaction, but other
more sophisticated criteria should be investigated.

These two thresholds are named:
\begin{description}
	\item[$\phi$] defining a minimum rate to consider a rule deterministically,
	\item[$\psi$] defining a minimum quantity to consider the involved rule deterministically.
\end{description}
Notice that these thresholds must be hold during all deterministic evolution in order to validate the hybrid approach.

Given as input a term $\ov{t}$ to reduce, a rate threshold $\phi$ and a quantity threshold $\psi$, each iteration of the hybrid reduction semantics performs the four steps listed in Figure~\ref{fig:CWM-hybrid-step}.
For every label $\ell$, the subsets of $\BB$ and $\NN$ containing the rules with label $\ell$ are denoted by
$\BB_\ell$ and $\NN_\ell$, respectively. The first step identifies, for each compartment $\iota\in I$ (where $I$ is the set of all compartment indexes occurring in $\ov{t}$), two disjoint sets of biochemical rules, namely
$\DD_\iota$ (to be applied deterministically) and $\RR_\iota$ (to be applied, together with the rules in $\NN$, according to the stochastic method). The
second step selects, considering only the rules in $\bigcup_{\iota \in I} \RR_\iota \cup \NN$, the next rule to be applied stochastically. When the stochastic transition is chosen, we ``lock'' the reagents involved in such a reaction. They will not contribute to the ODEs evolution (since they are already active in the stochastic sense), and their product is added at the end of the stochastic time step. The third step computes a system of ODEs $\Ode_\iota$ for each compartment $\iota\in I$ and applies
the ODEs for the time duration selected by the stochastic step. The calculation are performed in the same units as the stochastic computation, namely in terms of number of molecules and time. The fourth step updates the terms according to the results of the simulation.

Note that during the deterministic step there was the implicit assumption that the evolution of the species involved in fast reactions, calculated with the ODEs, did not significantly alter the propensity of the slow reactions changing their priority with respect to fast reactions. Omission of this hypothesis requires a control on the evolutionary trajectory of species calculated by the ODEs in order to stop deterministic evolution if the priorities of the reactions have changed its structure. 


Before computing the next stochastic step after the deterministic one, the fractional numbers of molecules computed by the ODEs need to be converted into integer numbers. As suggested by \cite{vasudeva2004adaptive}, fractions can be handled by the two following methods: rounding to the nearest integer and probabilistic rounding. Rounding to the nearest integer may introduce a bias. For example, if the number of molecules is $200.3$, it will always be rounded to $200$ and this bias may be amplified during the whole simulation. With a probabilistic rounding, instead, a molecular number of $200.3$ is rounded up to $201$ with probability $30\%$, and down to $200$ with probability $70\%$. Using such a technique, the average number of the molecular counts becomes identical to the one obtained using deterministic calculations.

In general, if reactions are fast enough, the deterministic ODEs simulation approximate better the exact stochastic simulations. This is the idea behind the use of the threshold $\phi$. The use of $\psi$, instead, allows to prevent the rounding approximation error that may derive when we are dealing with species of few elements. Combined together, the thresholds $\phi$ and $\psi$ affect the level of approximation we want to use in our simulations. Notice that with $\phi=+\infty$ all reactions will be considered \emph{too slow} and the simulation will be computed with the purely stochastic method.

\subsection{Running Example: Hybrid Simulations}\label{lab_toy_hybrid}
Hybrid simulations of the toy example presented in the previous section were performed by using thresholds $\phi = 0.5$ and $\psi = 10$. In Figure~\ref{fig:hybr_toy} we report two runs of the hybrid simulations showing two different evolutions of the species.

\begin{figure*}[t]
\centering
  \includegraphics[width=.49\textwidth]{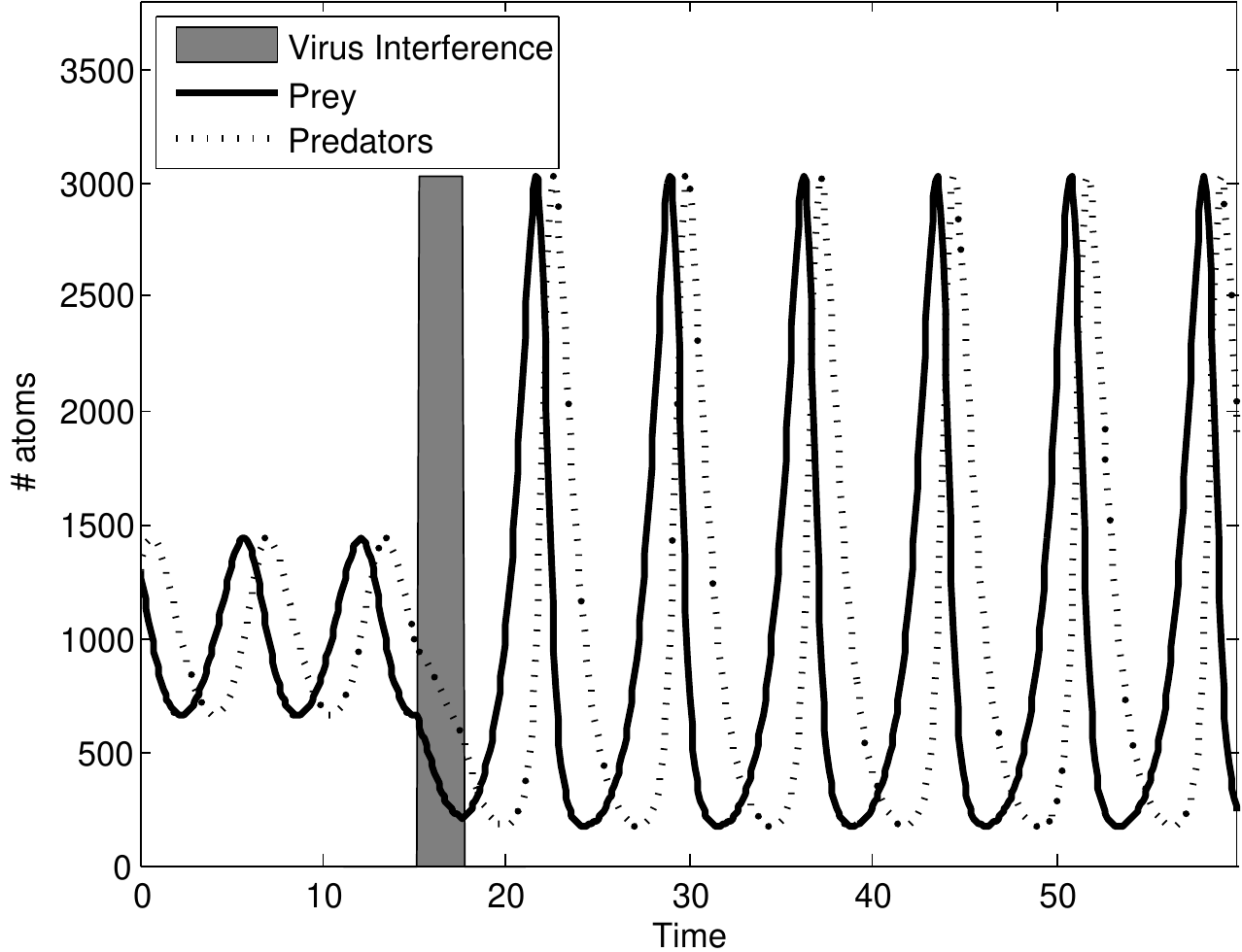}
 \includegraphics[width=.49\textwidth]{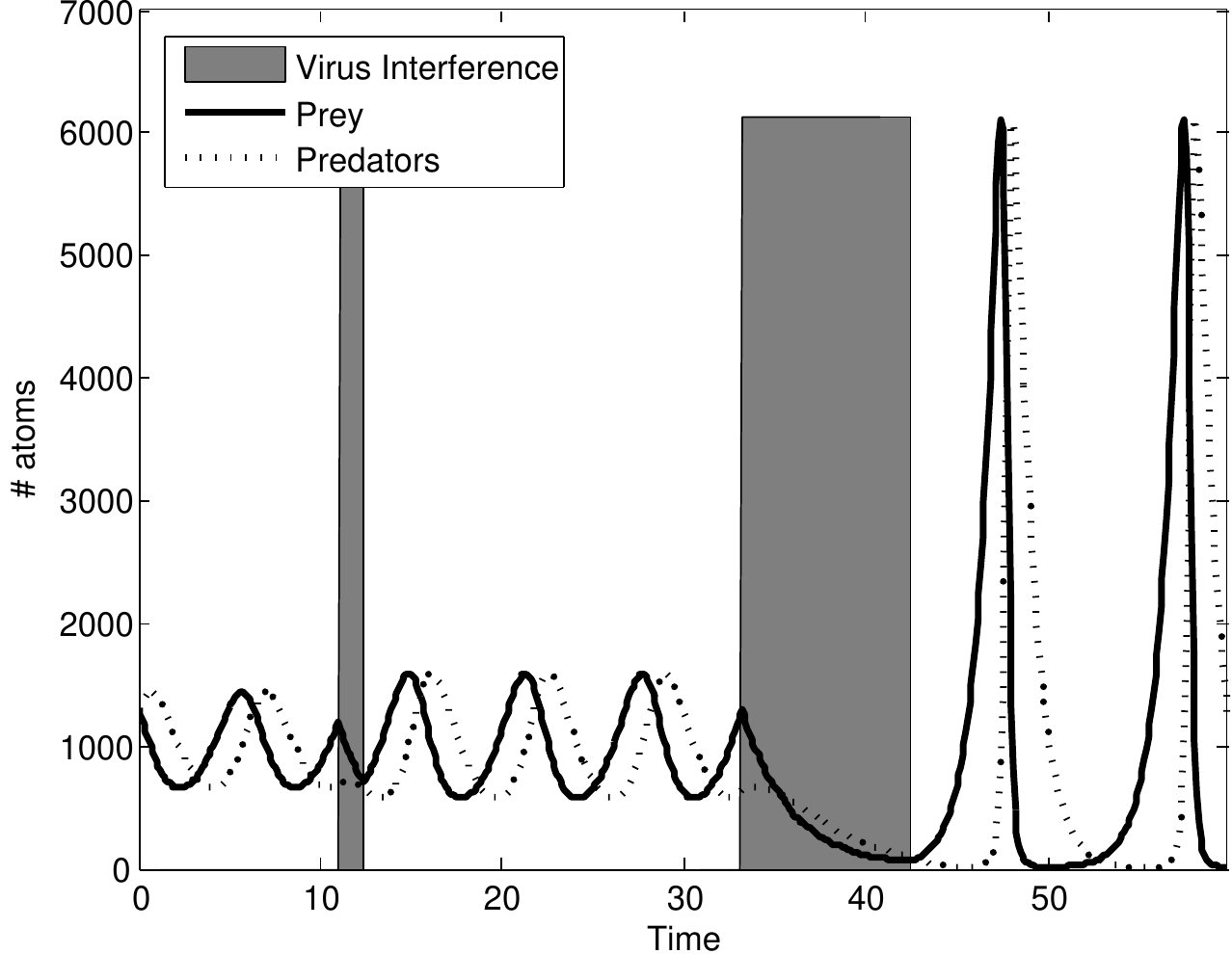}
  \caption{Two different runs of the hybrid simulations showing the different behaviour of the dynamics of the species
  inside the compartment \textit{IN} depending on the virus epidemic factor}
\label{fig:hybr_toy}
\end{figure*}

Notice that until the viral epidemic factor does not reach the compartment, the prey-predator dynamics (rules ($B_1$),($B_2$) and ($B_3$) of the model in Figure \ref{fig:LKvir-rules}) is treated deterministically since high propensities drive the sub-system. Conversely, a rare viral outbreak introduces a stochastic change inside the compartment influencing the amount of preys (and, consequently, the predators dynamics) through the rule ($B_4$).

The long period of evolution, waiting for the epidemic rare event given by the rules ($N_1$) and ($N_2$), allows an efficient deterministic computation. 
In this case, other approximation techniques, such as $\tau$-leaping, could not take the full advantage of ODEs since in such an oscillating scenario the propensities of the fast rules change rapidly and repeatedly over the time intervals and the length of leap is bound by the frequency of the oscillations.

A comparison of the computational time needed to perform $100$ runs using the hybrid method versus
the stochastic simulation technique provided a dramatic improvement on the computational effort. The relative speed-up, measured as the ratio of the computational run time of the stochastic simulations to that of hybrid simulations, was about $80$.

\subsection{Further Examples}

In this section the results of the hybrid algorithm,
when compared with the stochastic simulation algorithm, are presented on two benchmark models collected from the literature. The two models are a simple crystallisation system~\cite{SK05,haseltine2002approximate,GCPS06} and a model of intracellular viral infection~\cite{SYSY02,haseltine2002approximate,GCPS06} whose quantitative behaviours have been accurately reproduced by the hybrid simulation.

The first benchmark is a simplified model for the crystallisation of species $A$, consisting of two reactions, one occurring
many more times than the other. The rules  and rates are taken from \cite{SK05}, Table~3. Our partitioning scheme, setting the thresholds $\phi = 6 \cdot 10^{-4}$ and $\psi = 100$, allow to classify the fast reaction as continuous and the other as discrete.

Figure \ref{crist} compares the averages of species $A$ and $B$ computed by $100$ runs of the hybrid and pure stochastic simulation algorithms. In this example an implicit step (instead of an explicit one) for the ODEs solver could further improve the approximation.

\begin{figure*}[t]
\centering
\subfigure[Simple Crystallisation results]{
  \includegraphics[width=.43\textwidth]{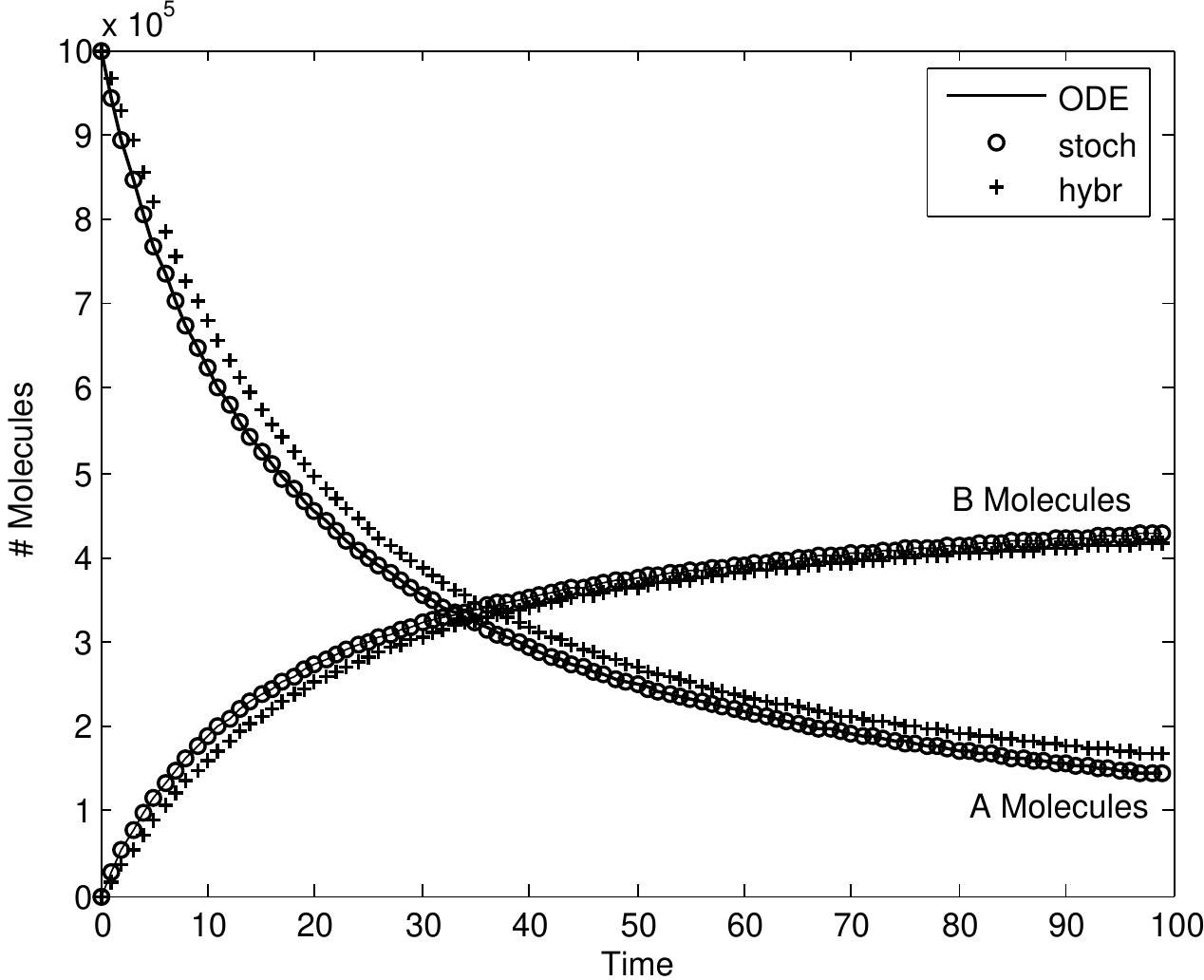}
\label{crist}}
\subfigure[Virus infection results]{
 \includegraphics[width=.4\textwidth]{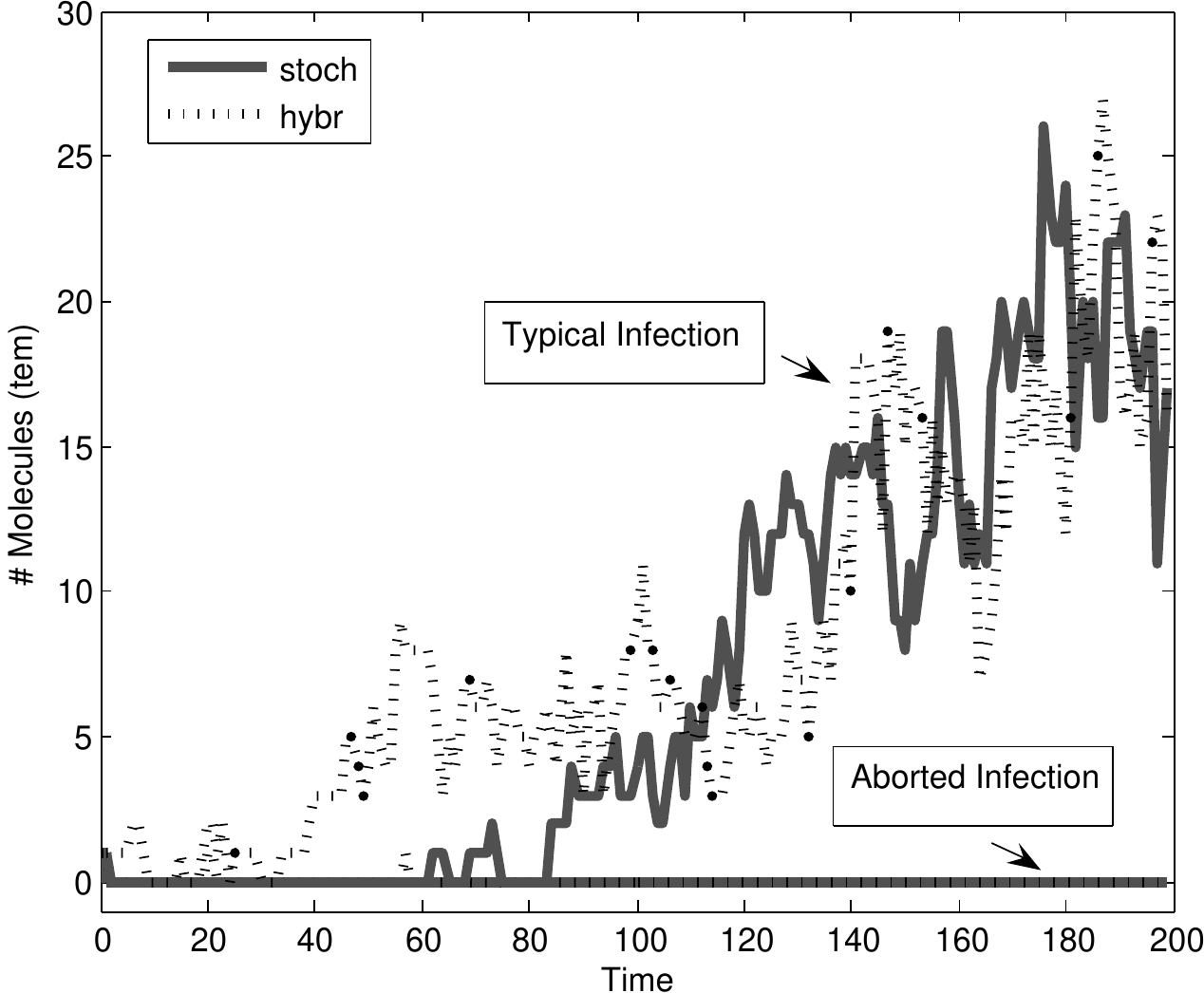}
  \label{virus}
}
  \caption{Comparison of the hybrid and pure stochastic simulations on the examples. Fig. (a) shows the average of $100$ runs of the hybrid and pure stochastic simulations compared with the ODE solution. Fig. (b)  shows two exemplificative solutions computed by the hybrid and pure stochastic simulations.}
\label{fig:examples}
\end{figure*}

The second benchmark is a general model of the infection of a cell by a virus. This model includes the genomic and template viral nucleic acids (\emph{gen} and \emph{tem} respectively) and the viral structural protein (\emph{struct}). This model has two interesting features: (i) the three components of the model exhibit fluctuations that vary by several orders of magnitude; and (ii) the model solution exhibits a bimodal distribution either a ``typical'' infection in which all species become populated, or an ``aborted'' infection in which all species are eliminated from the cell. The rules and rates are taken from \cite{SYSY02}, Table~1.

Figure \ref{virus} shows two exemplificative solutions computed by the hybrid (using $\phi = 10$ and $\psi = 25$) and pure stochastic simulations corresponding to a ``typical'' and ``aborted'' infection of species \emph{tem}.

Based on $100$ runs of the hybrid and pure stochastic simulation algorithms the speedup was around $500\%$ for the first benchmark and around $200\%$ for the second benchmark.
Note, however, that the presence of data structures to handle compartmentalisation introduces in our simulator an overhead which can reduce the efficiency differences between the two algorithms with respect to other hybrid simulators.

\section{A Real Model of Different Cellular Fate}\label{sect:tat}

To assess the soundness and efficiency of our hybrid approach on a real biological
problem we decided to apply it to a well known system where stochastic effects
play a fundamental role in determining its development: the HIV-1 transactivation mechanism.

After a cell has been infected, the retrotransposed DNA of the virus is integrated in the host genome and it begins its transcription into \textit{mRNA} and
then the translation to yield viral proteins; the initial speed of this mechanism, however, is fairly slow. The speedup of the viral production process is
determined by a regulation system driven by the viral protein \textit{TAT}: this protein is capable of binding cellular factors of the host to produce the
\textit{pTEFb} complex which in its acetylated form is able to bind to the integrated viral genome and speed up the transcription machinery, thus ending
in more viral proteins and, therefore, more \textit{TAT}, determining a positive loop.

The time scale during which this loop is triggered is affected by several factors e.g. the initial low \textit{TAT} production and the rate of its degradation,
the equilibrium between the active (acetylated) and inactive form of \textit{pTEFb}. As a consequence, the stochastic fluctuations in this
events are considered pivotal in determining when viral proteins are produced in a sufficient quantity to determine cellular lysis and viral spreading.
Since HIV is known to stay dormant and inactive in some types of cells and since the time between the infection and the high viral production rate related to the
active phase of AIDS is variable, this transactivation mechanism is of great interest.
Viral latency is also believed to be the cause of the persistent low level viremia observed in patients
treated with antiretroviral therapies, therefore understanding its molecular bases is fundamental in order
to be able to circumvent it and hopefully find a way to completely eradicate the virus avoiding lifelong therapies~\cite{Dahl2010}.

We decided to follow the direction taken in a previous study about this system (see~\cite{WBTAS05}), in which an experimental setting is developed where a
fluorescent protein, \textit{GFP}, is the only one encoded by an engineered viral genome, along with \textit{TAT}. In~\cite{WBTAS05} they were able to identify
different evolutions in the \textit{GFP} level over time: cellular clones with exactly the same genome showed two different behaviours, one produced a high
quantity of \textit{GFP} (they called it ``bright'') and the other one with very little \textit{GFP} (``off''). This work also reported that a purely
stochastic simulation was able to individuate this bistability; a later work (see~\cite{GCPS06}) confirmed these results performing purely stochastic and
mixed deterministic-stochastic simulations.

\begin{figure}[t]
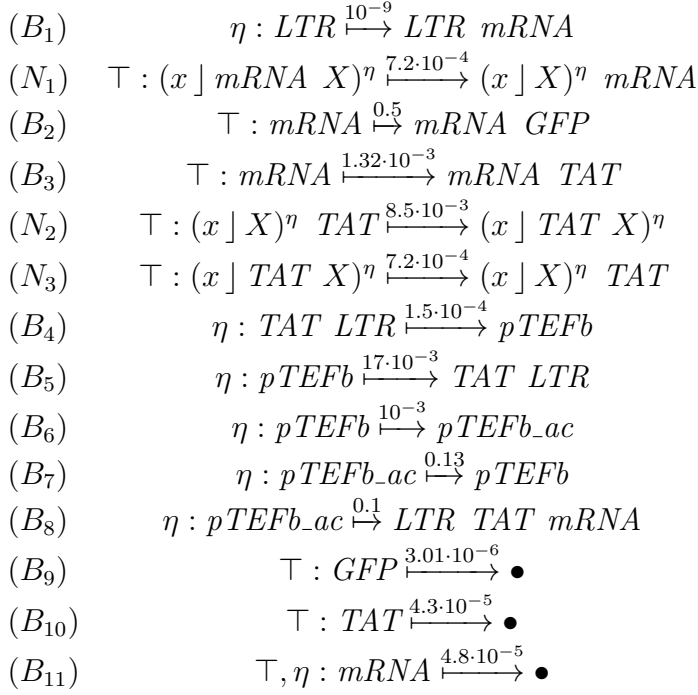

\hrule
$
\begin{array}{lc}
(B_1) & \eta : \textit{LTR} \xsrewrites{10^{-9}} \textit{LTR} \conc \textit{mRNA} 
\\
(N_1) & \TOP : (x \into \textit{mRNA} \conc X)^\eta \xsrewrites{7.2 \cdot 10^{-4}} (x \into X)^\eta \conc \textit{mRNA} 
\\
(B_2) & \TOP : \textit{mRNA} \xsrewrites{0.5} \textit{mRNA} \conc \textit{GFP} 
\\
(B_3) & \TOP : \textit{mRNA} \xsrewrites{1.32 \cdot 10^{-3}} \textit{mRNA} \conc \textit{TAT} 
\\
(N_2) & \TOP : (x \into X)^\eta \conc \textit{TAT} \xsrewrites{8.5 \cdot 10^{-3}} (x \into \textit{TAT} \conc X)^\eta 
\\
(N_3) & \TOP : (x \into \textit{TAT} \conc X)^\eta \xsrewrites{7.2 \cdot 10^{-4}} (x \into X)^\eta \conc \textit{TAT} 
\\
(B_4) & \eta : \textit{TAT} \conc \textit{LTR} \xsrewrites{1.5 \cdot 10^{-4}} \textit{pTEFb} 
\\
(B_5) & \eta : \textit{pTEFb} \xsrewrites{17 \cdot 10^{-3}}  \textit{TAT} \conc \textit{LTR}
\\
(B_6) & \eta : \textit{pTEFb} \xsrewrites{10^{-3}} \textit{pTEFb\_ac} 
\\
(B_7) & \eta : \textit{pTEFb\_ac} \xsrewrites{0.13} \textit{pTEFb} 
\\
(B_8) & \eta : \textit{pTEFb\_ac} \xsrewrites{0.1} \textit{LTR} \conc \textit{TAT} \conc \textit{mRNA} 
\\
(B_9) & \TOP : \textit{GFP} \xsrewrites{3.01 \cdot 10^{-6}} \emptyseq 
\\
(B_{10}) & \TOP : \textit{TAT} \xsrewrites{4.3 \cdot 10^{-5}} \emptyseq 
\\
(B_{11}) & \TOP, \eta: \textit{mRNA} \xsrewrites{4.8 \cdot 10^{-5}} \emptyseq 
\end{array}
$
 \caption{ \CalculusShortName\ rules for the TAT transactivation system}
\label{fig:TAT-rules}
\end{figure}

Since \short\ systems are able to represent compartments, we slightly modified the original set of rules used in these works to explicitly represent the
cytoplasm and the nucleus of an infected cell; all the kinetic rates were maintained, the one for \textit{TAT} nuclear import has been determined from
the literature (see~\cite{NLRB09}). The set of rules we adopted is given in Figure~\ref{fig:TAT-rules}, where we refer to the cytoplasm as the $\TOP$
compartment while $\eta$ is the label used for the nucleus. As regards the rules:
 $(B_1)$ 
 represents the slow basal rate of viral \textit{mRNA} transcription;
 $(N_1)$ 
 describes the \textit{mRNA} export from the nucleus to the cytoplasm;
 $(B_2)$ 
 and $(B_3)$ 
 express the translations of this \textit{mRNA} into
\textit{GFP}  and \textit{TAT} proteins, respectively;
 $(N_2)$ 
 and $(N_3)$ 
 represent the nuclear import and export of \textit{TAT};
 $(B_4)$ and $(B_5)$ model
the binding and unbinding of \textit{TAT} with (not represented here) host cellular factors and the viral genome portion \textit{LTR} that forms \textit{pTEFb} which, when
acetylated (by rule $(B_6)$) 
 determines an higher transcriptional activity, which is represented in $(B_8)$ 
 by the unbinding that releases
 \textit{LTR} and \textit{TAT} and creates an \textit{mRNA} molecule (note the higher rate with respect to
   $(B_1)$); 
   $(B_7)$  
   represents the \textit{pTEFb} deacetylation and
 $(B_9)$, 
 $(B_{10})$ 
 and $(B_{11})$ 
 model the degradation processes of the proteins and the \textit{mRNA} (note that \textit{mRNA} degrades both in the nucleus
and in the cytoplasm, the other proteins only degrade in the cytoplasm; also note how the compartment labelling mechanism allows to express this fact in a
simple and elegant way).

\begin{figure*}[h]
\centering
\includegraphics[width=.49\textwidth]{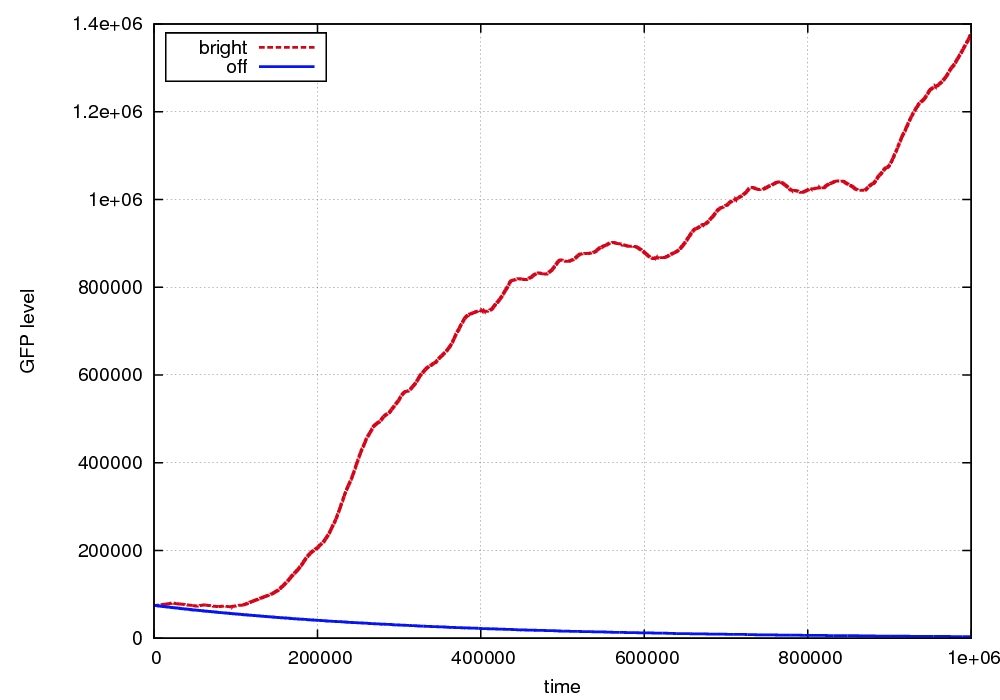}
\includegraphics[width=.49\textwidth]{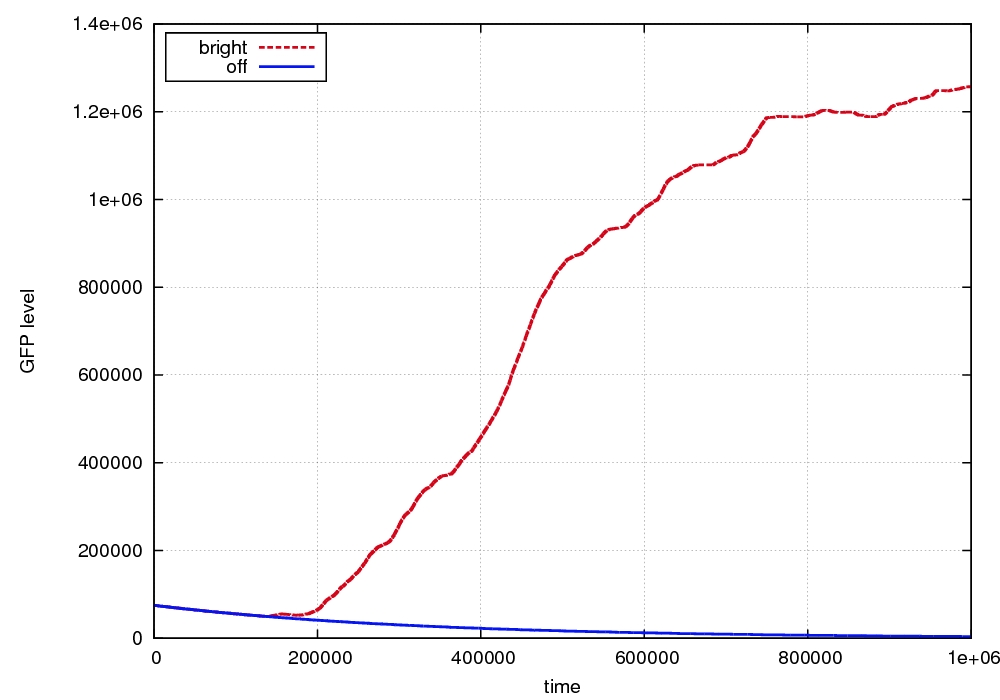}
\caption{Two different simulations pure stochastic (on the left) and hybrid (on the right) started with the same parameters: ``bright'' and ``off'' behaviour}
\label{fig:giano}
\end{figure*}

\begin{figure*}[h]
\centering
\includegraphics[width=.6\textwidth]{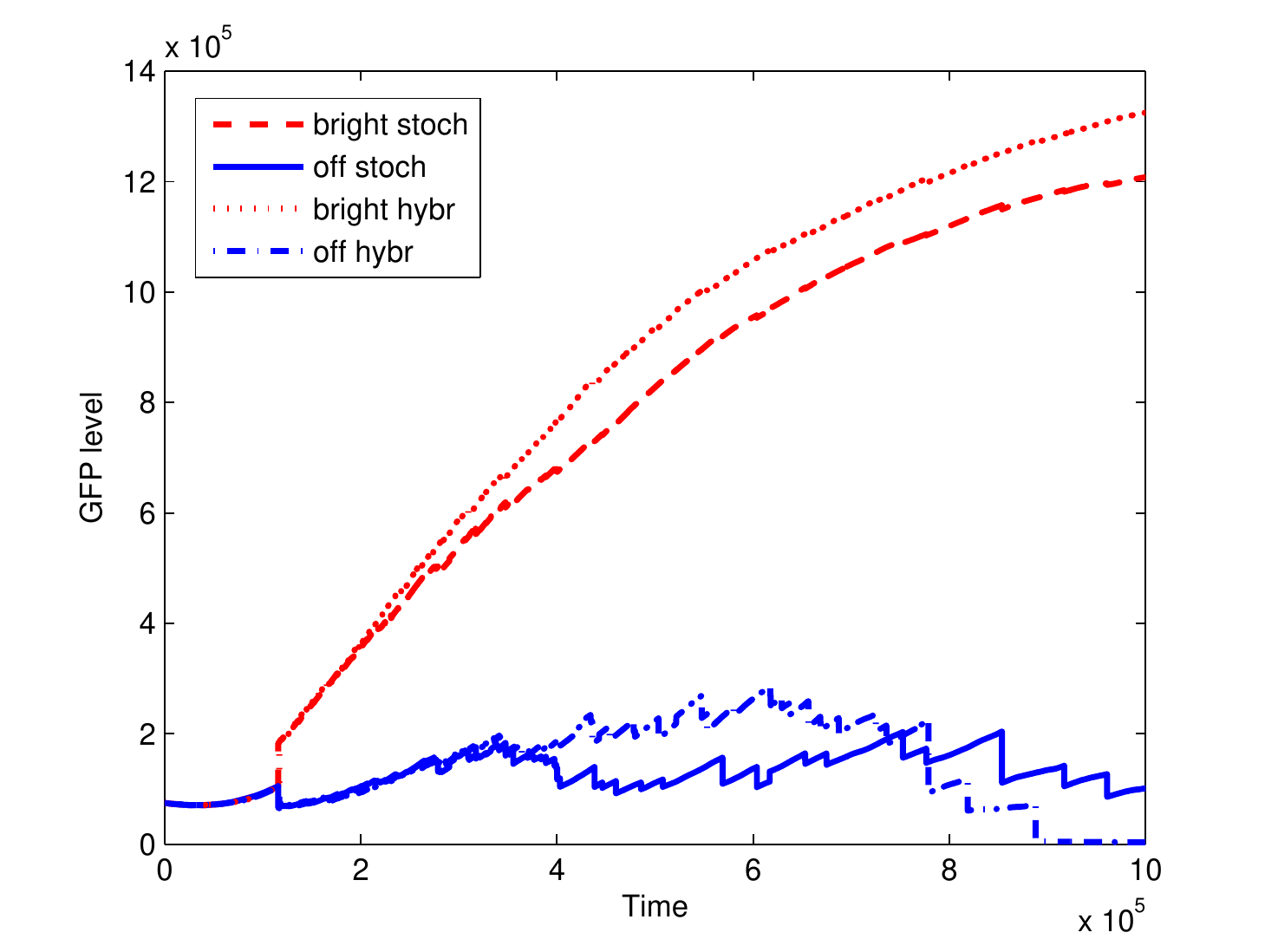}
\caption{Centroids comparison of the two clustered results obtained from $200$ simulations pure stochastic and hybrid showing the ``bright'' and ``off'' behaviour}
\label{fig:tat_clusters}
\end{figure*}

We performed $200$ purely stochastic simulations (i.e. setting $ \phi = +\infty$) and $200$ hybrid simulations (using $ \phi = 0.5$ and $\psi = 10$).
The initial term of our simulations is represented by the
\short\ term
$$
75000 \timesSilent \textit{GFP} \conc 5 \timesSilent \textit{TAT} \conc (\emptyseq \into \textit{LTR})^\eta ,
$$
 while the time interval of our simulations has been fixed to
$10^6$ seconds (the same parameters are used in~\cite{WBTAS05,GCPS06}). Both our stochastic and hybrid simulations clearly showed the two
possible evolutions of the system which correspond to the ``bright'' and the ``off'' cellular populations (in order to display the double destiny, almost
all the biochemical rewrite rules have to be simulated with the stochastic approach). Figure~\ref{fig:giano} shows two exemplificative runs of ``bright'' and ``off'' behaviour resulting from pure stochastic and hybrid simulations. Figure~\ref{fig:tat_clusters} reports the centroids of two clusters obtained using the k-means algorithm \cite{hartigan1979k} on the $200$ runs performed using the pure stochastic and hybrid simulations using a sampling step size $\Delta t = 10$. In both cases a $4\%$ of simulations showed an ``off'' behaviour and this results confirm the statistical analysis reported in~\cite{GCPS06}. The stochasticity of the centroids corresponding to the ``off'' behaviour is due to the few points belonging to the cluster.
As could be seen in Figures~\ref{fig:giano} and \ref{fig:tat_clusters}, the hybrid simulations
are comparable to the purely stochastic ones and, even with the relatively high thresholds used in this particular case, the hybrid simulations were
computationally more efficient (almost 40\% faster).\footnote{Comparisons are made using the same stochastic engine, in both cases with no particular optimisation.}

\chapter{Conclusions and Related Works}
\label{conclu}
This Chapter draws our concluding remarks, starting with a brief description of the tool developed for running CWC simulations.

\section{CWC Simulator}

The CWC simulator is available, as an open source project, from~\cite{HCWC_SIM}. It features:
\begin{itemize}
\item Gillespie-like simulations over CWC models,
\item Hybrid stochastic/deterministic semantics,
\item Dynamic cutoff-based semantical reconfiguration,
\item Multiple rating functions (mass action, Michaelis-Menten, Hill),
\item Multiple instances of simulation,
\item Online statistics (mean and variance under Grubb's test, quantiles, clustering)
\end{itemize}

The CWC simulator has also bean enriched in order to feature optimization for parallel multi-core computations based on the FastFlow framework (see~\cite{ACDDTT_PDP11,ACDDSSTT_HIBB11,ACCDDSSTT_BMRI}). A substantial increase in performance has been shown for both the shared-memory implementation and the distributed version.

\subsection{The Graphical User Interface}

The back-end tool
can steer the CWC simulation-analysis pipeline  either via a command line or a graphical user
interface. This latter makes it possible to design the biological
model, to run simulations and analysis and to view partial results
during the run. Also, the front-end  makes it possible to control the
simulation workflow from a remote machine. Two screenshots of the
graphical front-end are reported in Fig. \ref{fig:screenshot}.

\begin{figure}
\center
  \subfigure{
 \includegraphics[width=0.45\textwidth]{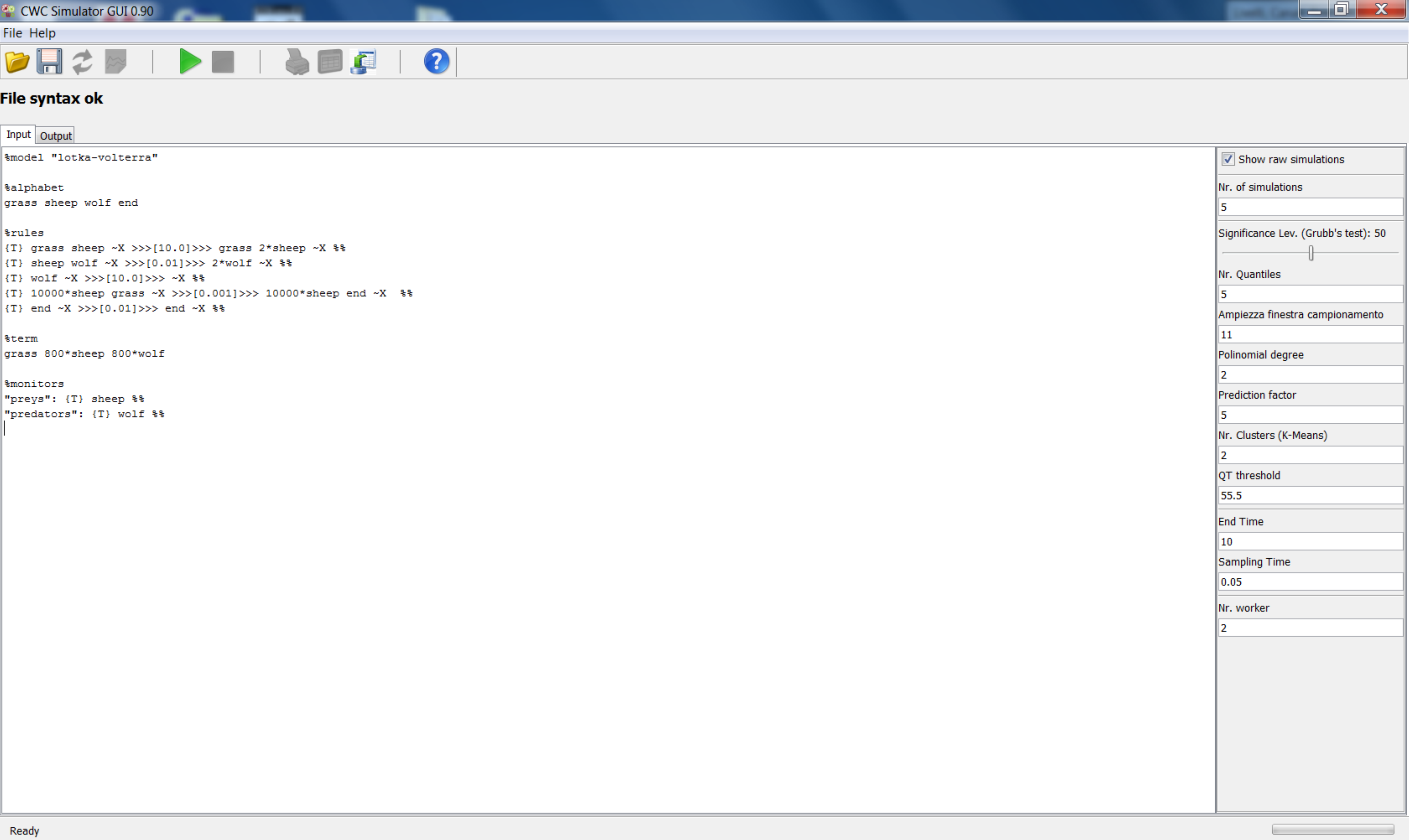}
}
\subfigure{
 \includegraphics[width=0.45\textwidth]{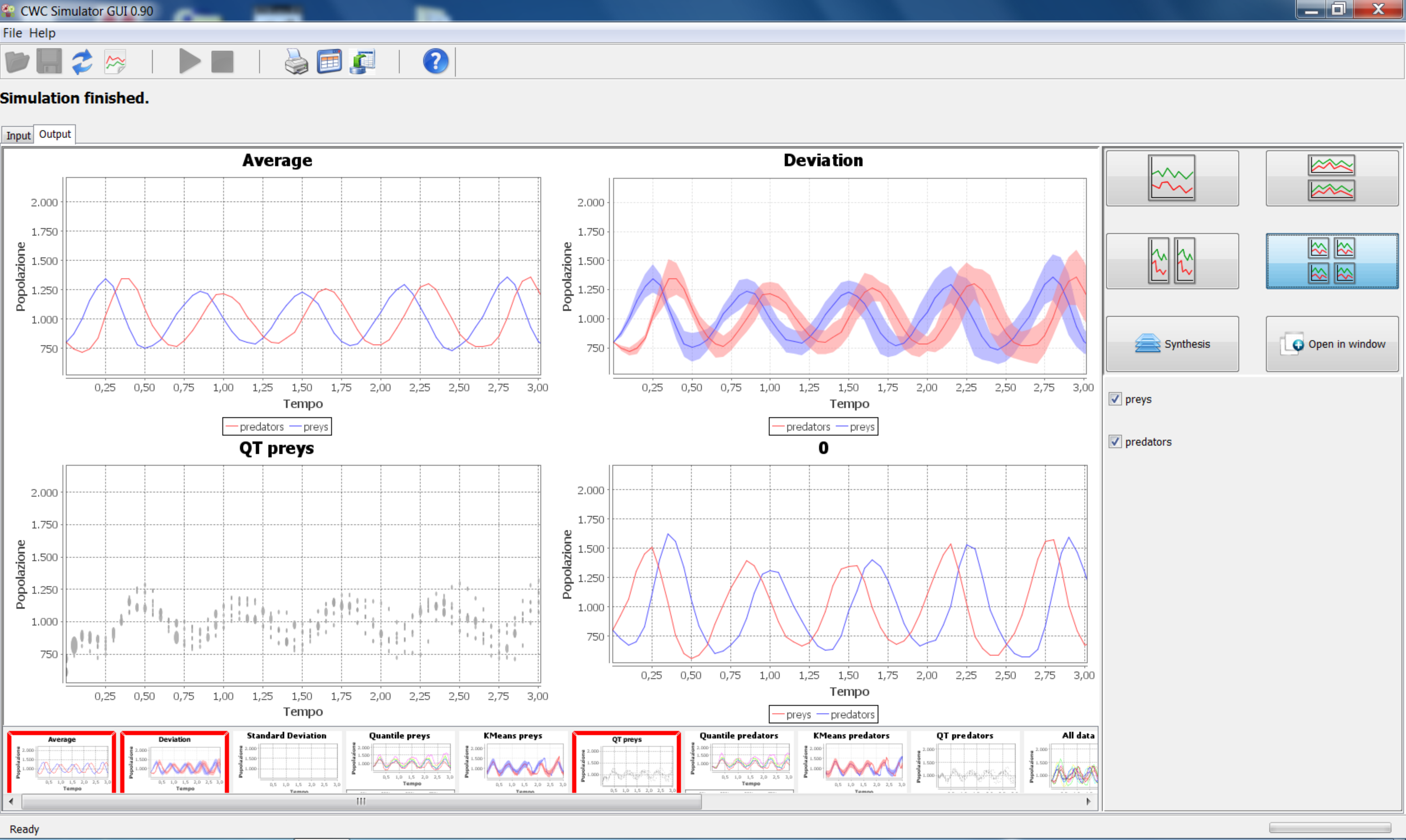}
}
  \caption{Screenshots of the simulation tool interface\label{fig:screenshot}}
\end{figure}

\section{Systems Biology}

This work reports on the use of CWC to simulate and understand
biological behaviour which is still unclear to biologists.
It also constitues an attempt to predict biological
behaviour and give some directions to biologists for future
experiments.

\subsection{Transporters in the AM Symbiosis}
CWC has been particularly suitable to model the AM symbiosis (where
substances flow through different cells) thanks to its feature of
modelling compartments, and their membranes, in a simple and natural
way. Our simulations have confirmed some of the latest
experimental results about the LjAMT2;2 transporter~\cite{GUE09}
and also support some of the hypotheses about the energy source
for the transport and the meaning of LjAMT2;2 overexpression
in arbusculated cells. These are the first steps towards a complete
simulation of the symbiosis and open some interesting paths that
could be followed to better understand the nutrient exchange.

As demonstrated by heterologous complementation experiments in
yeast, AM-specific plant transporters~\cite{Har02,GUE09}
show different uptake efficiencies under varying pH conditions.
Looking on this and the results from the simulations with
different pH conditions in the periarbuscular space new
experiments for a determination of the uptake kinetics
($K_m$-values) under a range of pH seems mandatory. Another
conclusion from the model is that, for accurate simulation, exact
in vivo concentration measurements have to be carried out even
though at the moment this is a difficult task due to technical
limitations.

As shown by various studies, many transporters on the plant side
show a strong transcriptionally regulation and the majority of
them are thought to be localized at the plant-fungus
interface~\cite{GUE09a}. Consequently a quantification of these
proteins in the membrane will be a prerequisite for an accurate
future model.

It is known that some of these transporters (e.g. PT4 phosphate
transporters) obtain energy from proton gradients established by
proton pumps (Figure~\ref{fig:symbschem}) whereas others as the
LjAMT2;2 and aquaporines~\cite{Gon06} use unknown energy supplies
or are simply facilitators of diffusion events along gradients.
Thus they conserve the membranes electrochemical potential for the
before mentioned proton dependent transport processes. To make the
story even more complex some of these transporters, which are known
to be regulated in the AM symbiosis (e.g. aquaporines), show
overlaps in the selectivity of their
substrates~\cite{Hol05,JMZ04,Tye02,NT00}. Consequently the integration
of the involved transporters will be a necessary
and challenging task in a future model. It is quite probable that also transporters
for other macronutrients (potassium and sulfate) which might be
localized in the periarbuscular membrane influence the
electrochemical gradients for the known transport processes.

With the ongoing sequencing work~\cite{MGHL07} on the arbuscular
model fungus {\it Glomus intraradices} transporters on the fungal
side of the symbiotic compartment are likely to be identified and
characterized soon. Data from such future experiments could be
integrated in the model and help to answer the question whether
transporter mediated diffusion or vesicle based excretion events
lead to the release of ammonium into the periarbuscular
space~\cite{Cha06}.

Further questions about the plant nutrient uptake and competitive
fungal reimport process~\cite{Cha06,Bal07} might be answered.
Based on the transport properties of orthologous transporters from
different fungal and/or plant species, theories could be developed
which explain different mycorrhiza responsiveness of host plants;
meaning why certain plant-AM fungus combinations have rather a
disadvantageous than a beneficial effect for the plant.

In future research on AM and membrane transport processes in
general values from measurements of concentrations or kinetics will be included in the model and will show how the whole
system is influenced by these values. Vice versa simulations could
be the base for new hypotheses and experiments.
Consequently simulated variations could help to
assess which type of experiments might be the most promising ones
to do.

\section{Computational Ecology}
The long-term goal of Computational Ecology is the development of methods to predict the response of ecosystems to changes in their physical, chemical and biological components. Computational models, and their ability to understand and predict the biological world, could be used to express the mechanisms governing the structure and function of natural populations, communities, and ecosystems. Until recent times, there was insufficient computational power to run stochastic, individually-based, spatially explicit models. Today, however, some of these techniques could be investigated~\cite{PP12}.

The main objective of Computational Ecology is twofold. First, we aim at discovering new models and theories within
Computer Science inspired by the natural world, and at producing techniques and tools to
deal with much more complex problems than those addressable with
current technology. Second, we plan to provide ecologists with an environment for attacking
problems at a system level, not addressable without using Information Technology. This
environment will provide ecologists with modelling, analysis and simulation tools capable of
handling complex behaviour and of representing emerging properties.

Calculi developed to describe process interactions in a compartmentalised setting are well suited for the description and analysis of the evolution of ecological systems. The topology of the ecosystem can be directly encoded within the nested structure of the compartments. These calculi can be used to represent structured natural processes in a greater detail, when compared to purely numerical analysis. As an example, food webs can give rise to combinatorial interactions resulting in the formation of complex systems with emergent properties (as signalling pathways do in cellular biology), and, in some cases, giving rise to chaotic behaviour.

Also, \emph{ecological niches} describe how organisms or populations respond to the distribution of resources and competitors~\cite{LB98}, and arise as the spatial sectors where organisms and populations tend to distribute, often forming geographical clusters of dominant species. Using the Spatial Calculus of Wrapped Compartments~\cite{BCCDSST11,spatial_COMPMOD11}, we might be able to extend our analysis to the case of ecological niches and show how they  could emerge from complex ecological interactions and then become a fundamental characteristic of an ecosystem.

As a final remark about ecological modelling with a framework based on stochastic rewrite rules, we underline an important compositional feature. How can we test an hypothetical scenario where a grazing species is introduced in the model of our case study? A possibility could be to represent the grazing species with a new CWC atom (e.g. $s$) and then just add the new competitive rules to the previously validated model (e.g. the rule $P: s \conc (h \conc x \into X)^c \srewrites{} s \conc ( x \into X)^c$). Changing in the same sense a model based on ordinary differential equations would, instead, result in a complete new model were all previous equations should be rewritten.

\section{Related works}

Computer simulations play
a rather unique role in Ecology, compared with other sciences, and there are several reasons for the extensive growth in
the number of simulation--based studies in ecological
mathematical modelling. In particular, the
problem is that, although a field experiment is a common
research approach in ecology, replicated experiments
under controlled conditions (which is a cornerstone of
all natural sciences providing information for theory
development and validation) are rarely possible because
of the transient nature of the environment: just consider, for instance,
the impossibility to reproduce the same weather
pattern for a repeated experiment in an ecological study. We also mention it here that
large--scale ecological experiments are costly and, in the situation
when consequences are poorly understood, can have
adverse effects on some species, put in danger the biodiversity of the ecosystem,
and may even pose a threat to human well--being.
Capturing the complexity of real systems through tractable
experiments is therefore logistically not feasible.
Mathematical modelling and computer simulations
create a convenient ``virtual environment'' and hence
can provide a valuable supplement, or sometimes even
an alternative, to the field experiment.

\subsection{Ecological Modelling}
It has long been recognised that numerical modelling and computer simulations can be used
as a powerful research tool to understand, and sometimes to predict, the tendencies and
peculiarities in the dynamics of populations and ecosystems. It has been, however, much
less appreciated that the context of modelling and simulations in ecology is essentially different
from those that normally exist in other natural sciences~\cite{PP12}.

Ecology became a quantitative and
theory--based science since the seminal studies by Lotka \cite{Lot25}, Volterra \cite{Vol26} and Gause \cite{Gau34}, who were the first to use mathematical
tools for ecological problems. General principles of ecosystem
organisation were later refined and systematised by
Odum \& Odum~\cite{OO53}, while the mathematical theory was
further developed by Skellam~\cite{Ske51} and Turing~\cite{Tur52}, who
emphasised the importance of the spatial aspect.
The bright mathematical ideas of those seminal works
sparked a huge fire. The last quarter of the twentieth century
saw an outbreak of interest in mathematical ecology
and ecological modelling \cite{WO76,OMWM89,DMW91,MWD93,SLF95,DMW98}. Especially over the last
decade, more and more complicated models have been
developed with a generic target to take into account
the ecological interactions in much detail and hence to
provide an accurate description of ecosystems dynamics.
Owing to their increased complexity, many of the models
had to be solved numerically, a development that was
inspired and made possible by the simultaneous advances in computer
science and technology.

The modelling approaches can be very different in terms of the mathematics used and depending on the goals of the study, and there are several ways to classify them. For instance, there is an apparent difference between statistical models \cite{Cza98} and ``mechanistic'' models \cite{Mur93,OL01}, although simulation--based studies may sometimes include both. Taken from another angle, two qualitatively different modelling streams are rule--based approaches (such as individual-based models and cellular automata) \cite{GR05} and equation--based ones.

Another way to sort out the models used in ecology is to consider the level of complexity involved. Depending on the purposes of the ecological study, there have been two different streams in model building \cite{May74}. In case the purpose is to predict the system's state (with a certain reasonable accuracy), the model is expected to include as many details as possible. This approach is often called \emph{predictive modelling}. The mathematical models arising in this way can be very complicated and analytically intractable in an exact way (in these cases the model can still be partially analysed via a limited number of runs of computer simulations) \cite{PCSKR96,NCO11,Pas05,GW92}. Alternatively, the purpose of the study can be to understand the current features of the ecosystem, e.g. to identify the factors responsible for a population decline or a population outbreak, but not necessarily to predict their development quantitatively. We will call this approach a \emph{conceptual modelling}. In this case, the corresponding models can be pretty simple, even if their exact solutions are still not always possible; therefore, they often have to be solved by simulations as well~\cite{OMWM89,DMW91,SLF95}.
These two streams of theoretical--ecological research can be clearly seen in the literature, even though it is not always straightforward to distinguish between them as sometimes simple models may show a certain predictive power and, on the contrary, complicated ones are used for making a qualitative insight into some subtle issues.

\subsection{Formal Computational Frameworks}
As P-Systems~\cite{Pau00,Pau02} and the Calculus of Looping Sequences (CLS, for short)~\cite{BMMT07}, the Calculus of Wrapped Compartments is a framework modelling topological compartmentalisation inspired by biological membranes, and with a semantics given in terms of rewrite rules.

CWC has been developed as a simplification of CLS, focusing on stochastic multiset rewriting. The main difference between CWC and CLS consists in the exclusion of the sequence operator, that constructs ordered strings out of the atomic elements of the calculus. While the two calculi keep the same expressiveness, some differences arise on the way systems are described. On the one hand, the Calculus of Looping Sequences allows to define ordered sequences in a more succinct way (for examples when describing sequences of genes in DNA or sequences of amino acids in proteins).\footnote{An ordered sequence can be expressed in CWC as a series of nested compartments, ordered from the outermost compartment to the innermost one.} On the other hand, CWC reflects in a more realistic way the fluid mosaic model of the lipid bilayer (for example in the case of cellular membrane description, where proteins are free to float), and, the addition of compartment labels allows to characterise the properties peculiar to given classes of compartments. Ultimately, focusing on multisets and avoiding to deal explicitly with ordered sequences (and, thus, variables for sequences) strongly simplifies the pattern matching procedure in the development of a simulation tool.

The Calculus of Looping Sequences has been extended with type systems in~\cite{ADT09,DGT09,DGT09b,BDMMT10,BDGT12}. As an application to ecology, stochastic CLS (see~\cite{BMMTT08}) is used in~\cite{BCBMMR10} to model population dynamics.

P-Systems have been proposed as a computational model inspired by biological structures. They are defined as a nesting of membranes in which multisets of objects can react according to pre defined rewrite rules. Maximal-parallelism is the key feature of P-Systems: at each evolution step all rewrite rules, in all membranes, are applied as many times as possible. Such a feature makes P-Systems a very powerful computational model and a versatile instrument to evaluate expressiveness of languages. However, it is not practical to describe stochastic systems with a maximally-parallel evolution: exact stochastic simulations based on race conditions model systems evolutions as a sequence of successive steps, each of which with a particular duration modelled by an exponential probability distribution.

There is a large body of literature about applications of P-Systems to ecological modelling. In~\cite{CCMPPPS11,CCMPPS09,CCPSM09}, P-Systems are enriched with a probabilistic semantics to model different ecological systems in the Catalan Pyrenees. Rules could still be applied in a parallel fashion since reduction durations are not explicitly taken into account. In~\cite{BCPM07,BCPM08,BCPM10}, P-Systems are enriched with a stochastic semantics and used to model metapopulation dynamics. The addition of \emph{mute rules} allows to keep a form of parallelism reducing the maximal consumption of objects.

While all these calculi allow to manage systems topology through nesting and compartmentalisation, explicit spatial models are able to depict more precise localities and \emph{ecological niches}, describing how organisms or populations respond to the distribution of resources and competitors~\cite{LB98}. The spatial extensions of CWC~\cite{BCCDSST11}, CLS~\cite{BMMP11} and P-Systems~\cite{BMMPT11} could be used to express this kind of analysis allowing to deal with spatial coordinates.

\subsection{Hybrid Models}

Stochastic formulation of chemical kinetics is mainly based on Gillespie's algorithm~\cite{G77}, which explicitly accounts for the individual reactive
collisions among the molecules. However, it is problematic to use exact simulation methods to study systems containing a large number of molecules affected by fast reactions due to the computational cost of accounting for individual molecular collisions.

The problem of efficient simulation of systems involving reactions varying across multiple scales of time and molecular concentrations employing a mixed
stochastic--deterministic method to approximate system dynamics is not new, and has been already addressed by several recent studies. Gillespie, in
\cite{gillespie2001approximate}, presented the ``$\tau$--leap'' method, an approximate technique for accelerating stochastic simulation, in which the
occurrence of some fast reactions can be eliminated by taking time steps that are larger than a single reaction. An improved procedure for selecting the
value of $\tau$ has also been presented in \cite{cao2006efficient}.

Haseltine and Rawlings~\cite{haseltine2002approximate} partition the system into the
subsets of \textit{slow} and \textit{fast} reactions, and approximate the fast reactions either deterministically or as Langevin equations. In the method
of Rao and Arkin~\cite{rao2003stochastic}, some of the reactions are explicitly simulated with Gillespie's algorithm whereas others are described by
random variables distributed according to the probability density functions at quasi-stationary state. The last two methods require direct intervention of
the modeller to partition the system into reaction sets covering different time and concentration regimes.
Similarly to our hybrid approach Salis and Kaznessis~\cite{SK05} proposed two parameters to define how many reactions occur within a time step and how fine grained the species must be to appear continuous-valued. They approximate fast reactions using a chemical Langevin equation. Our method, instead, simplifies the stochastic integration ignoring
the fluctuations in fast reaction dynamics using ODEs.

Bortolussi and Policriti~\cite{BP10} provide a Stochastic Concurrent Constraint Programming (sCCP) algebra with a semantics based on hybrid automata combining discrete and continuous steps. A similar technique is developed in~\cite{BP09} for the Stochastic $\pi$-Calculus. In~\cite{SHB10}, a hybrid analysis technique, combining stochastic
simulations with ODEs is presented in the context of PEPA precess algebra.

In~\cite{GHB08}, the HYPE process algebra, developed to model hybrid systems in which the continuous behaviour of a subsystem
does not need to be understood in advance of the modelling process, is used to model the repressilator genetic regulatory network.

In~\cite{FOA11} a hybrid technique, computing ODE with the Runge-Kutta numerical
approximation method, is adapted
in the Real-Time Maude rewriting logic.

The state of the art approaches, methods and tools in hybrid modelling  cut down the computational cost of large stochastic exact simulations. However, they introduced new forms of complexities. In particular: (i) when different scales are taken into consideration, a deep consistency study on the system under analysis should be carried out, (ii) new parameters have to be defined to control the degree of the approximation, (iii) there is no commonly understood policy to partition the set of rules.

\bibliographystyle{plain}
\bibliography{fmb}
\end{document}